\documentclass[11pt]{article}
\usepackage{hyperref}
\pdfoutput=1
\usepackage{fullpage}
\usepackage{amsmath}
\usepackage{amsfonts}
\usepackage{amssymb}
\usepackage{amsthm} 
\usepackage{mathtools}
\usepackage{thmtools, thm-restate}
\usepackage[utf8]{inputenc}
\usepackage{cellspace}
\usepackage[small,sc]{caption}
\usepackage{booktabs, multicol, multirow}
\usepackage[shortlabels]{enumitem}
\usepackage{xspace}
\usepackage{comment}
\usepackage{units}
\usepackage{booktabs}
\usepackage[dvipsnames, usenames]{xcolor}
\usepackage{placeins}
\usepackage{booktabs}
\usepackage{accents}

\hypersetup{
	colorlinks=true,
	linkcolor=Sepia,
	citecolor=Sepia,
	filecolor=black,
	urlcolor=Sepia
}
\usepackage[capitalize]{cleveref}

\crefformat{equation}{(#2#1#3)}

\Crefformat{chapter}{#2{}Chapter #1{}#3}
\Crefformat{section}{#2{}Section #1{}#3}
\Crefformat{appendix}{#2{}Appendix #1{}#3}
\Crefformat{figure}{#2{}Figure #1{}#3}
\Crefformat{table}{#2{}Table #1{}#3}
\Crefformat{theorem}{#2{}Theorem #1{}#3}
\Crefformat{definition}{#2{}Definition #1{}#3}
\Crefformat{lemma}{#2{}Lemma #1{}#3}
\Crefformat{corollary}{#2{}Corollary #1{}#3}
\Crefformat{restatable}{#2{}Theorem #1{}#3}
\Crefformat{item}{#2{}Theorem #1{}#3}
\Crefformat{Item}{#2{}Theorem #1{}#3}

\usepackage{caption}

\usepackage{tikz}
\usetikzlibrary{shapes.geometric, arrows}
\tikzstyle{startstop} = [rectangle, rounded corners, minimum width=3cm, minimum height=1cm,text centered, draw=black, fill=white!30]
\tikzstyle{arrow} = [thick,->,>=stealth]

\setlist[1]{itemsep=-3pt}

\newcommand\xqed[1]{%
  \leavevmode\unskip\penalty9999 \hbox{}\nobreak\hfill
  \quad\hbox{#1}}
\newcommand\qeddef{\xqed{\raisebox{-1.5pt}{$\triangle$}}}

\newtheorem{theorem}{Theorem}
\newtheorem{lemma}{Lemma}

\newtheorem{corollary}{Corollary}

\theoremstyle{definition}
\newtheorem{definition}{Definition}

\def\HS{{\sc Halld{\'o}rsson-Szegedy Guessing Game}\xspace} 

\def\OPT{\texttt{OPT}\xspace}

\def\ALG{\texttt{ALG}\xspace}
\def\DET{\texttt{D}\xspace}

\def\ra{\texttt{R}\xspace}
\def\piw{\ensuremath{p_{i\vert w}}\xspace}

\def\P{{\textsc P}\xspace}
\def\Prs{\ensuremath{\textsc{P}^*_\Sigma}\xspace}
\def\Prm{\ensuremath{\textsc{P}^*_\lor}\xspace}

\def\gs{\ensuremath{\mathcal{S}}\xspace}
\def\gt{\ensuremath{\mathcal{T}}\xspace}
\def\posr{\ensuremath{\mathbb{R}_{\geq 0}}\xspace}

\def\pea{\ensuremath{{p\displaystyle_{\mkern+1mu \alpha, \varepsilon}}}\xspace}
\def\pear{\ensuremath{{p\textstyle^r\displaystyle_{\mkern-8mu\alpha, \varepsilon}}}\xspace}
\def\init{\ensuremath{s}\xspace}
\def\lea{\ensuremath{{L\displaystyle_{\mkern+1mu \alpha, \varepsilon}}}\xspace}

\DeclareMathOperator{\E}{\mathbb{E}}

\DeclareMathOperator{\prob}{Pr}
\DeclareMathOperator{\Prob}{Pr}

\DeclareMathOperator{\KL}{KL}
\DeclareMathOperator{\supp}{supp}
\DeclareMathOperator{\cost}{cost}
\DeclareMathOperator{\score}{score}

\DeclareMathOperator{\dist}{dist}
\DeclareMathOperator{\pair}{pairs}

\def\dri{2^{n^{O(1)}}}

\newcommand{\ab}[1] {\left\vert #1\right\vert}
\DeclarePairedDelimiter\paren{\lparen}{\rparen}

\usepackage{authblk}

\begin{document}

\title{Randomization can be as helpful as a glimpse of the future in online computation}

\author{
	Jesper W. Mikkelsen\thanks{Supported in part by the Villum Foundation and the Danish Council for Independent Research, Natural Sciences.}\\
	\small \texttt{jesperwm@imada.sdu.dk}\\
	\small University of Southern Denmark
}

\date{\vspace{-5ex}}

\maketitle

\begin{abstract}
We provide simple but surprisingly useful direct product theorems for proving lower bounds on online algorithms with a limited amount of advice about the future. Intuitively, our direct product theorems say that if $b$ bits of advice are needed to ensure a cost of at most $t$ for some problem, then $r\cdot b$ bits of advice are needed to ensure a total cost of at most $r\cdot t$ when solving $r$ independent instances of the problem.
Using our direct product theorems, we are able to translate decades of research on randomized online algorithms to the advice complexity model. Doing so improves significantly on the previous best advice complexity lower bounds for many online problems, or provides the first known lower bounds. For example, we show that
\begin{itemize}
\item A paging algorithm needs $\Omega(n)$ bits of advice to achieve a competitive ratio better than $H_k=\Omega(\log k)$, where $k$ is the cache size. Previously, it was only known that $\Omega(n)$ bits of advice were necessary to achieve a constant competitive ratio smaller than $5/4$.
\item Every $O(n^{1-\varepsilon})$-competitive vertex coloring algorithm must use $\Omega(n\log n)$ bits of advice. Previously, it was only known that $\Omega(n\log n)$ bits of advice were necessary to be optimal.
\end{itemize} 
For certain online problems, including the MTS, $k$-server, metric matching, paging, list update, and dynamic binary search tree problem, we prove that randomization and sublinear advice are equally powerful (if the underlying metric space or node set is finite). This means that several long-standing open questions regarding randomized online algorithms can be equivalently stated as questions regarding online algorithms with sublinear advice. For example, we show that there exists a deterministic $O(\log k)$-competitive $k$-server algorithm with sublinear advice if and only if there exists a randomized $O(\log k)$-competitive $k$-server~algorithm~without~advice.

Technically, our main direct product theorem is obtained by extending an information theoretical lower bound technique due to Emek, Fraigniaud, Korman, and Ros{\'e}n [ICALP'09]. 
\end{abstract}

\newpage
\tableofcontents
\newpage

\section{Introduction}
The model of online computation deals with optimization problems where the input arrives sequentially over time. Usually, it is assumed that an online algorithm has no knowledge of future parts of the input. While this is a natural assumption, it leaves open the possibility that a tiny amount of information about the future (which might be available in practical applications) could dramatically improve the performance guarantee of an online algorithm. Recently, the notion of advice complexity \cite{A1, A2, A3, A4} was introduced in an attempt to provide a quantitative and problem-independent framework for studying \emph{semi-online} algorithms with limited (instead of non-existing) knowledge of the future. In this framework, the limited knowledge is modeled as a number of advice bits provided to the algorithm by an oracle (see \cref{cadef}). The goal is to determine how much advice (measured in the length, $n$, of the input) is needed to achieve a certain competitive ratio. In particular, one of the most important questions is how much advice is needed to break the lower bounds for (randomized) online algorithms without advice.
It has been shown that for e.g.\ bin packing and makespan minimization on identical machines, $O(1)$ bits of advice suffice to achieve a better competitive ratio than what is possible using only randomization~\cite{DBLP:conf/wads/AngelopoulosDKR15,Amakespan}. On the other hand, for a problem such as edge coloring, it is known that $\Omega(n)$ bits of advice are needed to achieve a competitive ratio better than that of the best deterministic online algorithm without advice~\cite{Aedge}. However, for many online problems, determining the power of a small amount of advice has remained an open problem.

With a few notable exceptions (e.g.\ \cite{sg,A2,Ak-server,AOC}), most of the previous research on advice complexity has been problem specific.
In this paper, we take a more complexity-theoretic approach and focus on developing techniques that are applicable to many different problems. Our main conceptual contribution is a better understanding of the connection between advice and randomization.
Before explaining our results in details, we briefly review the most relevant previous work.

Standard derandomization techniques~\cite{DBLP:journals/ipl/Newman91,DBLP:conf/soda/ChattopadhyayEEP12}
imply that a randomized algorithm can be converted (maintaining its competitiveness) into a deterministic algorithm with advice complexity $O(\log \log I(n)+\log n)$, where $I(n)$ is the number of inputs of length $n$~\cite{Ak-server} (see also \cref{sec:derand}). Clearly, there are problems where even a single bit of advice is much more powerful than randomization, and so we cannot in general hope to convert an algorithm with advice into a randomized algorithm. However, using machine learning techniques, Blum and Burch have shown that a metrical task system algorithm with advice complexity $O(1)$ can be converted into a randomized algorithm without advice~\cite{BBmts}. Problems such as paging, $k$-server, and list update can be modeled as metrical task systems (see e.g.\ \cite{BE98b}).

\subsection{Overview of results and techniques} 
 We will give a high-level description of the results and techniques introduced in this paper. For simplicity, we restrict ourselves to the case of minimization problems\footnote{All of our results (and their proofs) are easily adapted to maximization problems. See \cref{app:sec:max} for details.}.
\paragraph*{Direct product theorems.} Central to our work is the concept of an \emph{$r$-round input distribution}. Informally, this is a distribution over inputs that are made up of $r$ rounds such that the requests revealed in each round are selected independently of all previous rounds. Furthermore, there must be a fixed upper bound on the length of each round (see \cref{rround}).

As the main technical contribution of the paper, we prove direct product theorems for $r$-round input distributions. Intuitively, a direct product theorem says that if $b$ bits of advice are needed to ensure a cost of at most $t$ for each individual round, then $rb$ bits of advice are needed to ensure a cost of at most $rt$ for the entire input. This gives rise to a useful technique for proving advice complexity lower bounds. In particular, it follows that a linear number of advice bits are needed to get a (non-trivial) improvement over algorithms without any advice at all.

We provide two different direct product theorems formalizing the above intuitive statement in different ways: The first theorem (\cref{mainthm}) is based on an information theoretical argument similar to that of \cite{A2}. The second theorem (\cref{dpmartin}) uses martingale theory.

\paragraph*{Repeatable online problems.} We combine our direct product theorems with the following very simple idea: Suppose that we have a lower bound on the competitive ratio of randomized online algorithms without advice. Often, such a lower bound is proved by constructing a hard input distribution and then appealing to Yao's principle. For some online problems, it is always possible to combine (in a meaningful way) a set of input sequences $\{\sigma_1,\ldots , \sigma_r\}$ into one long input sequence $\sigma=g(\sigma_1,\ldots , \sigma_r)$ such that serving $\sigma$ essentially amounts to serving the $r$ smaller inputs separately and adding the costs incurred for serving each of them. We say that such problems are \emph{$\Sigma$-repeatable} (see \cref{re}). For a $\Sigma$-repeatable online problem, an adversary can draw $r$ input sequences independently at random according to some hard input distribution. This gives rise to an $r$-round input distribution. By our direct product theorem, an online algorithm needs linear advice (in the length of the input) to do better against this $r$-round input distribution than an online algorithm without advice.
Thus, for $\Sigma$-repeatable online problems, we get that it is possible to translate lower bounds for randomized algorithms without advice into lower bounds for algorithms with sublinear advice. More precisely, we obtain the following theorem.
\begin{restatable}{theorem}{restateone}
Let $\P$ be a $\Sigma$-repeatable online minimization problem and let $c$ be a constant. Suppose that for every $\varepsilon>0$ and every $\alpha$, there exists an input distribution $\pea:I\rightarrow [0,1]$ with finite support such that $\E_{\pea}[\DET(\sigma)]\geq (c-\varepsilon)\E_{\pea}[\OPT(\sigma)]+\alpha$ for every deterministic algorithm $\DET$ without advice. Then, every randomized algorithm reading at most $o(n)$ bits of advice on inputs of length $n$ has a competitive ratio of at least $c$.
\label{maint1}
\end{restatable}
Much research has been devoted to obtaining lower bounds for randomized algorithms without advice. \cref{maint1}{} makes it possible to translate many of these lower bounds into advice complexity lower bounds, often resulting in a significant improvement over the previous best lower bounds (see \cref{maintable}).

For a $\Sigma$-repeatable problem, the total cost has to be the sum of costs incurred in each individual round. It is also possible to consider another kind of repeatable problems, where the total cost is the \emph{maximum cost} incurred in a single round. We call such problems $\lor$-repeatable (see \cref{sec:maxrep}). Many online coloring problems are $\lor$-repeatable. For $\lor$-repeatable problems, we show in \cref{lorrep} that under certain conditions, a constant lower bound on the competitive ratio of \emph{deterministic} algorithms without advice carries over to randomized algorithms, even if the randomized algorithms have advice complexity $o(n)$. The proof of this result is straightforward and does not rely on our direct product theorems for $\Sigma$-repeatable problems. However, the result improves or simplifies a number of previously known advice complexity lower bounds for $\lor$-repeatable problems.

In \cref{maintable}, we have listed most of the repeatable online problems for which lower bounds on algorithms with sublinear advice existed prior to our work, and compared those previous lower bounds with the lower bounds that we obtain in this paper. We have also included two examples of repeatable online problems for which \cref{maint1}{} provides the first known advice complexity lower bounds.

It is evident from \cref{maintable} that there are many repeatable online problems. On the other hand, let us mention that e.g.\ bin packing and makespan minimization are examples of problems which are provably not repeatable.

\begin{table}[ht!]
\begin{center}
\begin{minipage}{\textwidth}
\renewcommand{\thempfootnote}{$\mathsection$}
\begin{center}
\begin{tabular}{|l|lr|lr|}
\cline{1-5}
\rule{0pt}{5ex}

& \multicolumn{4}{c|}{\shortstack{Lower bound for algorithms \\ with advice complexity $o(n)$}}  \\
\hline
\hline
\multicolumn{1}{|c|}{$\Sigma$-repeatable problem} & \multicolumn{2}{c|}{\bf \cref{maint1}} & \multicolumn{2}{c|}{Previous best} \\
\hline
\rule{0pt}{3ex} Paging & $\Omega\left(\log k\right)$ & \cite{BLSmetric}& $5/4$ &\cite{A1}\\
\rule{0pt}{3ex} $k$-Server & $\Omega\left(\log k\right)$&\cite{BLSmetric}& $3/2$&\cite{Smula} \\
\rule{0pt}{3ex} $2$-Server & $1+e^{-1/2}$&\cite{2serverlower} & $3/2$&\cite{Smula} \\
\rule{0pt}{3ex} List Update & $3/2$&\cite{Teia93} & $15/14$&\cite{AListupdate} \\
\rule{0pt}{3ex} Metrical Task Systems & $\Omega\left(\log N\right)$&\cite{BLSmetric}& $\Omega\left(\log N\right)$&\cite{A2}\\
\rule{0pt}{3ex} Bipartite Matching & $e/(e-1)$&\cite{optbip}& $1+\varepsilon$&\cite{Abipfinal} \\

\rule{0pt}{3ex} Reordering Buffer Management & $\Omega\left(\log \log k\right)$&\cite{STOCrbm} & $1+\varepsilon$&\cite{Arbm} \\
\rule{0pt}{3ex} 2-Sleep States Management & $e/(e-1)$&\cite{KMMO90} & $7/6$&\cite{Assm} \\
\rule{0pt}{3ex} Unit clustering & $3/2$&\cite{UnitClustering} & $-$& \\
\rule{0pt}{3ex} Max-SAT & $3/2$&\cite{Max-SAT} & $-$& \\
\hline
\hline
\multicolumn{1}{|c|}{$\lor$-repeatable problem} & \multicolumn{2}{c|}{\bf \cref{lorrep}} & \multicolumn{2}{c|}{Previous best} \\
\hline
\rule{0pt}{3ex} Edge Coloring & $2$&\cite{Bar-Noy} & $2$&\cite{Aedge} \\
\rule{0pt}{3ex} $L(2,1)$-Coloring on Paths & $3/2$&\cite{AL21} & $3/2$&\cite{AL21} \\
\rule{0pt}{3ex} $2$-Vertex-Coloring & $\omega(1)$&\cite{ColorBean} & $2$&\cite{Avbipartite} \\
\hline
\end{tabular}
\end{center}
\end{minipage}
\end{center}
\caption{New and previously best lower bounds on the competitive ratio for algorithms reading $o(n)$ bits of advice on inputs of length $n$. For the $\Sigma$-repeatable problems, the lower bounds in this table are obtained by combining \cref{maint1}{} with known lower bounds for randomized online algorithms without advice. For the $\lor$-repeatable problems, the lower bounds are obtained by combining our general result on $\lor$-repeatable problems with known lower bounds for \emph{deterministic} algorithms without advice. In both cases, references to these previously known lower bounds for online algorithms without advice are provided in the second column of the table. See \cref{sec:appmaint1} and \cref{sec:maxrep} for a more detailed explanation of the entries in \cref{maintable}, and for a comparison with the current upper bounds. \\
Bipartite matching and Max-SAT are maximization problems, and hence the lower bound is obtained via \cref{maint1:max} (which is the maximization version of \cref{maint1}). For paging and reordering buffer management, $k$ denotes the cache/buffer size. For metrical task systems, $N$ is the number of states. For the metrical task system problem and the $k$-server problem, the bounds are for a worst-case metric. It is also possible to use \cref{maint1}{} together with known lower bounds for specific metric spaces. The lower bound for unit clustering is for the one-dimensional case.
}
\label{maintable}
\end{table}
\newpage
\paragraph*{Compact online problems.} 
Note that when translating a lower bound on randomized algorithms without advice to a lower bound on algorithms with $o(n)$ bits of advice via \cref{maint1}{}, we had to make some assumptions on the lower bound. This begs the following question: Are there online problems where every lower bound (after suitable modifications) satisfies these assumptions? 
Recall that for finite games, Yao's principle works in both directions: In particular, if there is a lower bound for randomized computation, then there exists an input distribution witnessing this lower bound. Unfortunately, online problems are almost always infinite games (since there is no bound on the length of the input), and so this direction of Yao's principle do not necessarily apply.  However, for certain online problems, it follows from previous results (see \cite{ambuhl, momke} and \cref{subsec:lts}) that this direction of Yao's principle \emph{do} apply in the sense that the (approximately) best possible lower bound for randomized algorithms without advice is witnessed by a family of hard input distributions with finite support. We call such online problems \emph{compact} (\cref{compdef}). 
If an online problem is both $\Sigma$-repeatable and compact, we get the following result: Let $c$ be a constant and let $\varepsilon>0$. If $c$ is a lower bound on the competitive ratio of randomized online algorithms without advice, then $c-\varepsilon$ is also a lower bound on the competitive ratio of online algorithms with sublinear advice. Equivalently, by contraposition, the existence of a $c$-competitive algorithm with sublinear advice implies the existence of a $(c+\varepsilon)$-competitive randomized algorithm without advice. Combining this with existing derandomization results \cite{Ak-server} (see also \cref{sec:derand}) yields the following complexity theoretic equivalence between randomization and sublinear advice (previously, only the forward implication was known~\cite{Ak-server}):
\begin{theorem}
\label{mmtt}
Let $\P$ be a compact and $\Sigma$-repeatable minimization problem with at most $\dri$ inputs of length $n$, and let $c$ be a constant independent of $n$. The following are equivalent:
\begin{enumerate}
\item For every $\varepsilon>0$, there exists a randomized $(c+\varepsilon)$-competitive $\P$-algorithm without advice.
\item For every $\varepsilon>0$, there exists a deterministic $(c+\varepsilon)$-competitive $\P$-algorithm with advice~complexity $o(n)$.
\end{enumerate}
\end{theorem}

In this paper, we use a technique due to Amb{\"u}hl~\cite{ambuhl} and M{\"{o}}mke~\cite{momke} to show that the class of compact and $\Sigma$-repeatable problems contains all problems which can be modeled as a metrical task system (MTS) with a finite number of states and tasks. This means that e.g.\ the $k$-server, list update, paging, and dynamic binary search tree problem are all compact and $\Sigma$-repeatable, assuming that the underlying metric space or node set is finite. For all these problems, it is known that it is possible to achieve a constant competitive ratio with respect to the length of the input. Also, the number of inputs of length $n$ for each of these problems is at most $\dri$ (this bound holds since when we apply \cref{mmtt}{} to e.g.\ the $k$-server problem, the metric space will be fixed and not a part of the input). Thus, for each of these problems, \cref{mmtt}{} applies. Furthermore, for all the problems just mentioned (except paging), determining the best possible competitive ratio of a randomized algorithm without advice is regarded as important open problems~\cite{BE98b,DBLP:journals/siamcomp/DemaineHIP07}. For example, the randomized $k$-server conjecture says that there exists a randomized $O(\log k)$-competitive $k$-server algorithm~\cite{DBLP:journals/csr/Koutsoupias09}. \cref{mmtt}{} shows that this conjecture is equivalent to the conjecture that there exists a deterministic $O(\log k)$-competitive $k$-server algorithm with advice complexity $o(n)$.  Currently, it is only known how to achieve a competitive ratio of $O(\log k)$ using $2n$ bits of advice~\cite{Ak-server,Ak-serverRR}.

We also show that there are compact and $\Sigma$-repeatable problems which cannot be modeled as a MTS. One such example is the metric matching problem (on finite metric spaces)~\cite{DBLP:journals/jal/KalyanasundaramP93}.

Pessimistically, \cref{mmtt}{} may be seen as a barrier result which says that (for compact and $\Sigma$-repeatable online problems) designing an algorithm with sublinear advice complexity and a better competitive ratio than the currently best randomized algorithm without advice might be very difficult. Optimistically, one could hope that this equivalence might be useful in trying to narrow the gap between upper and lower bounds on the best possible competitive ratio of randomized algorithms without advice. In all cases, \cref{mmtt}{} shows that understanding better the exact power of (sublinear) advice in online computation would be very useful.

\subsection{Lower bounds for specific online problems}
The previously mentioned applications of our direct product theorem treat the hard input distribution as a black-box. However, it is also possible to apply our direct product theorem to explicit input distributions. Doing so yields some interesting lower bounds which cannot be obtained via \cref{maint1}{}. In what follows, we give three such examples. 

\paragraph*{Repeated matrix games.} 
Let $q\in\mathbb{N}$ and let $A\in \posr^{q\times q}$ be a matrix defining a two-player zero-sum game. Let $V$ denote the value of the game defined by $A$. The \emph{repeated matrix game (RMG)} with cost matrix $A$ is an online problem where the algorithm and adversary repeatedly plays the game defined by $A$. The adversary is the row-player, the algorithm is the column-player, and the matrix $A$ specifies the cost incurred by the online algorithm in each round. This generalizes the string guessing problem \cite{sg} and the generalized matching pennies problem \cite{A2} (both of these essentially corresponds to the RMG with a $q\times q$ matrix $A$ where $A(i,j)=1$ if $i\neq j$ and $A(i,i)=0$). Using our direct product theorem, we easily get that for every $\varepsilon>0$, an online algorithm which on inputs of length $n$ is guaranteed to incur a cost of at most $(V-\varepsilon)n$ must read $\Omega(n)$ bits of advice.
\begin{restatable}{theorem}{restateRMGmainthm}
\label{m:RMGmainthm}
\label{RMGmainthm}
Let $\ALG$ be an algorithm for the RMG with cost matrix $A$. Furthermore, let $V$ be the value of the (one-shot) two-person zero-sum game defined by $A$ and let $0<\varepsilon\leq V$ be a constant. If $\E[\ALG(\sigma)]\leq (V-\varepsilon)n$ for every input $\sigma$ of length $n$, then $\ALG$ must read at least 
\begin{equation}
\label{RMGmainthmeq}
b\geq \frac{\varepsilon^2}{2\ln(2)\cdot\|A\|_\infty^2}n=\Omega(n)
\end{equation}
bits of advice. 
\end{restatable}
 Furthermore, we also show how a more careful application of our direct product theorem to some particular repeated matrix games yields good trade-off results for the exact amount of advice needed to ensure a cost of at most $\alpha n$ for $0<\alpha<V$. See \cref{sec:repmat} for details and for the proof of \cref{RMGmainthm}.

\paragraph{A better bin packing lower bound via repeated matrix games.} We use our results on repeated matrix games to prove the following advice complexity lower bound for bin packing:
\begin{restatable}{theorem}{restatebetterbp}
Let $c<4-2\sqrt{2}$ be a constant. A randomized $c$-competitive bin packing algorithm must read at least $\Omega (n)$ bits of advice.
\label{thm:betterbp}
\end{restatable}
Previously, Angelopoulos et al.\ showed that a bin packing algorithm with a competitive ratio of $c<7/6$ had to use $\Omega(n)$ bits of advice by a reduction from the binary string guessing problem~\cite{DBLP:conf/wads/AngelopoulosDKR15, Abp}. From our results on repeated matrix games, we obtain a lower bound for \emph{weighted} binary string guessing. Using the same reduction as in~\cite{DBLP:conf/wads/AngelopoulosDKR15}, but reducing from weighted binary string guessing instead, we improve the lower bound for bin packing algorithms with sublinear advice to $4-2\sqrt{2}$. Thus, even though bin packing itself is not repeatable, we can obtain a better lower bound via a reduction from a repeated matrix game. The proof of \cref{thm:betterbp} can be found in \cref{appendix:binpackproof}.

\paragraph*{Superlinear lower bounds for graph coloring.} 
We obtain the following superlinear lower bound for online graph coloring by applying our direct product theorem to an ingenious hard input distribution due to Halld{\'o}rsson and Szegedy \cite{MS} (they show that a randomized graph coloring algorithm \emph{without} advice must have a competitive ratio of at least $\Omega(n/\log^2 n)$):
\begin{restatable}{theorem}{restategraphcollow}
\label{m:thm:graphcollow}
\label{thm:graphcollow}
Let $\varepsilon>0$ be a constant. A randomized $O(n^{1-\varepsilon})$-competitive online graph coloring algorithm must read at least $\Omega(n\log n)$ bits of advice.
\end{restatable}

Previously, it was only known that $\Omega(n\log n)$ bits of advice were necessary to be $1$-competitive \cite{Agraphc}.
Note that $O(n\log n)$ bits of advice trivially suffice to achieve optimality for graph coloring. Furthermore, it is not hard to prove that for every $c=n^{1-o(1)}$, there exists a $c$-competitive graph coloring algorithm reading $o(n \log n)$ bits of advice. Thus, our lower bound saying that $\Omega(n\log n)$ bits are needed to be $O(n^{1-\varepsilon})$-competitive for every constant $\varepsilon>0$ is essentially tight. See \cref{sec:graphc} for the proof of \cref{thm:graphcollow}.




\subsection{Relation to machine learning}
\label{rtml}
\cref{mmtt}{} is very closely related to previous work on combining online algorithms \cite{BBmts,FiatFKRRV98,AzarBM93}. Let $A_1,\ldots , A_m$ be $m$ algorithms for a MTS of finite diameter $\Delta$. Based on a variant of the celebrated machine learning algorithm Randomized Weighted Majority (RWM) \cite{WMA}, Blum and Burch obtained the following result \cite{BBmts, Burch}: For every $\varepsilon>0$, it is possible to combine the $m$ algorithms into a single randomized MTS-algorithm, $R$, such that
\begin{equation}
\label{bbeq}
\E[R(\sigma)]=(1+2\varepsilon)\cdot\min_{1\leq i\leq m}A_i(\sigma)+\left(\frac{7}{6}+\frac{1}{\varepsilon} \right)\Delta\ln m,
\end{equation}
for every input $\sigma$. An algorithm with $b$ bits of advice corresponds to an algorithm which runs $m=2^b$ algorithms in parallel (and selects the best one at the end). Thus, equation \cref{bbeq} immediately implies that given a $c$-competitive MTS-algorithm with advice complexity $b=O(1)$, we can convert it to a randomized $(c+\varepsilon)$-competitive algorithm without advice.

\cref{mmtt}{} improves on the result of Blum and Burch in two ways. First of all, it allows us to convert algorithms with sublinear instead of only constant advice complexity. Furthermore, \cref{mmtt}{} applies to all compact and $\Sigma$-repeatable online problems, not just those which can be modeled as a MTS.

It is natural to ask if it is possible to use the technique of Blum and Burch in order to obtain a constructive proof of \cref{mmtt}{}. To this end, we remark that the result of Blum and Burch relies fundamentally on the fact that the cost incurred by RWM when switching from the state of algorithm $A_i$ to the state of algorithm $A_j$ for $i\neq j$ is bounded by a constant independent of $n$. Thus, it does not seem possible to extend the result to those compact and $\Sigma$-repeatable problems which does not satisfy this requirement (such as the metric matching problem). On the other hand, in \cref{sec:constructmain}, we show that by combining the result of Blum and Burch with the ideas that we use to prove that the MTS problem is compact, it is possible to use a variant of RWM to algorithmically convert a $c$-competitive MTS-algorithm with advice complexity $o(n)$ (instead of just $O(1)$) into a randomized $(c+\varepsilon)$-competitive algorithm without advice. This yields a constructive version of \cref{mmtt}{} for problems that can be modeled as a MTS.
\FloatBarrier

Very shortly after and independently of our work, B\"{o}ckenhauer et al.\ also considered applications of machine learning techniques to advice complexity of online algorithms~\cite{aML}.

\subsection{Other related work}
\label{sec:related}
\paragraph{Compact online problems.} The existing results on (what we call) the compactness of online problems were not motivated by advice complexity. If an online problem is compact, this essentially gives a way of constructing an algorithm which can approximate the best possible competitive ratio for that online problem. Such an algorithm is known as a competitive ratio approximation scheme~\cite{DBLP:conf/soda/GuntherMMW13}. In general, it is not clear how construct such a scheme since online problems can have arbitrarily long input sequences (or an infinite number of possible requests in each round). For deterministic algorithms, this idea goes back at least to a paper by Manasse et al. \cite{DBLP:conf/stoc/ManasseMS88}. See also \cite{MR1165342, LPforOP}. More recently, Amb{\"u}hl showed that the list update problem is compact~\cite[Theorem 2.3]{ambuhl}. Furthermore, is has been shown by M{\"{o}}mke in \cite{momke} that the $k$-server problem is compact on finite graphs. In \cite{momke}, M{\"{o}}mke also observes that the same technique he uses to show compactness of the $k$-server problem can be used to show that a larger class of online problems (see \cite{DBLP:conf/stacs/KommKKM14}), including the MTS problem on finite metrics, are compact. However, \cite{momke} does not contain a full proof of this. For completeness, we show in \cref{subsec:lts} that e.g.\ metrical task systems with finite state spaces are compact. Our proof technique is essentially the same as the one used by both Amb{\"u}hl and M{\"{o}}mke.

\paragraph{The high-entropy technique of Emek et al.} As already mentioned, our information theoretical direct product theorem (\cref{mainthm}) is an extension of a previous lower bound technique due to Emek et al.~\cite{A2}. We will call this technique the \emph{high-entropy technique}. The only known application of this technique is a lower bound for generalized matching pennies \cite{A2}, but it is clear that the technique could also be used (together with an $r$-round input distribution) for other online problems.

Consider the set of $r$-round input distributions with support $I$. The crux of the high-entropy technique is to show that for any $r$-round input distribution over $I$, if the entropy of the distribution is sufficiently high, then every online algorithm must incur a large cost. This means that in order to achieve a low cost, an online algorithm must read enough advice so that the conditional entropy of the input sequence with respect to the advice becomes sufficiently low. In order to use the high-entropy technique, it must be the case that every probability distribution over $I$ with high entropy yields a good lower bound. While it often happens that the uniform distribution over $I$ (which is the distribution over $I$ with maximum entropy) yields a good lower bound, this is far from always the case. Sometimes, good lower bounds can only be obtained by using non-uniform distributions. The main advantage of our direct product theorem over the high-entropy technique is that it can be used to obtain good lower bounds even if the uniform distribution is not a hard input distribution. In particular, this allows us to prove lower bounds for randomized online problems while treating the hard input distribution as a black-box.

\subsection{Further work}
We expect the techniques and results introduced in this paper to find other applications, besides those that we provide. Our main result on $\Sigma$-repeatable online problems, \cref{maint1}, is very easy to apply and reduces the problem of proving a linear advice complexity lower bound to that of proving a lower bound against randomized algorithms without advice. Specifically, it allows one to quickly deduce what is (implicitly) already known about the advice complexity of such problems from previous results on randomized algorithms without advice. Similarly, the results on repeated matrix games might prove useful as a starting point for showing hardness results via advice-preserving reductions. 

As stated, \cref{maint1} can only be used to show a constant lower bound on the competitive ratio of algorithms with sublinear advice, even if there exists a lower bound of $c(n)$ on the competitive ratio of randomized algorithms without advice, where $c(n)$ is an increasing function of the input length $n$. It is possible to use the techniques from this paper to obtain a more general version of \cref{maint1} which given such a lower bound on randomized algorithms yields an advice complexity lower bound for achieving a competitive ratio which may depend on $n$. This result will appear in the final version of the paper.

\section{Preliminaries}
\label{Preliminaries}
Throughout the paper, $n$ always denotes the length of the input (number of requests). As a consequence, we use $N$ (instead of $n$) to denote the number of nodes in a graph or number of points in a metric space. We usually denote an online algorithm by $\ALG$. If we want to emphasize that an algorithm is deterministic (resp. randomized), we use the notation $\DET$ (resp. $\ra$) instead. We let $\OPT$ be an optimal offline algorithm for the problem at hand. Furthermore, we let $\log$ denote the binary logarithm $\log_2$, we let $\ln$ denote the natural logarithm $\log_e$, and we let $\posr=[0,\infty)$.
\subsection{The advice complexity model}
We start by formally defining competitive analysis and advice complexity. Since we are very much interested in sublinear advice, we will use the \emph{advice-on-tape model}~\cite{A1, A3}. In this model, the algorithm is allowed to read an arbitrary number of advice bits from an advice tape. There is an alternative model, the \emph{advice-with-request model}~\cite{A2}, where a fixed number of advice bits is provided along with each request.

We give the formal definition of the advice-on-tape model and competitive analysis only for minimization problems, but it can easily be adapted to maximization problems.
\begin{definition}[Advice complexity \cite{A1,A3} and competitive ratio \cite{CompRatio1,CompRatio2}]\label{cadef}
The input to an online problem, \P, is a sequence $\sigma=(\init, x_1,\ldots , x_n)$. We say that $s$ is the \emph{initial state} and $x_1,\ldots , x_n$ are the \emph{requests}. A \emph{deterministic online algorithm with advice}, \ALG, computes the output $\gamma=( y_1,\ldots , y_n)$, under the constraint that $y_i$ is computed from $\varphi,\init, x_1,\ldots , x_i$, where $\varphi$ is the content of the advice tape. The \emph{advice complexity}, $b(n)$, of \ALG is the largest number of bits of $\varphi$ read by \ALG over all possible inputs of length at most $n$. 

For a request sequence $\sigma$, $\ALG(\sigma)$ $(\OPT(\sigma))$ denotes the non-negative \emph{cost} of the output computed by $\ALG$ $(\OPT)$ when serving $\sigma$. 
We say that $\ALG$ is \emph{$c$-competitive} if there exists a constant $\alpha$ such that $\ALG (\sigma)\leq c\cdot\OPT (\sigma)+\alpha$ for all request sequences $\sigma$. If the inequality holds with $\alpha=0$, we say that $\ALG$ is \emph{strictly $c$-competitive}.

A \emph{randomized online algorithm $\ra$ with advice complexity $b(n)$} is a probability distribution over deterministic online algorithms with advice complexity at most $b(n)$. We say that $\ra$ is $c$-competitive if there exists a constant $\alpha$ such that $\E[\ra(\sigma)]\leq c\cdot \OPT(\sigma)+\alpha$ for all request sequences $\sigma$.
\label{def}
\qeddef
\end{definition}


Suppose that $\ALG$ is a deterministic algorithm with advice complexity $b(n)$. For every fixed input length $n$, $\ALG$ can be converted into $2^{b(n)}$ algorithms $\ALG_1,\ldots , \ALG_{2^{b(n)}}$ without advice, one for each possible advice string, such that $\ALG(\sigma)=\min_{i}\ALG_i(\sigma)$ for every input $\sigma$ of length at most $n$.
Another basic observation is that the advice partitions the set of inputs of length $n$ into $2^{b(n)}$ disjoint subsets.

In order to make the definition more readable, there are some (unimportant) technicalities not addressed in \cref{def}. In particular, we will assume that whenever $\ALG$ has to compute some $y_i$, there is a non-empty set of \emph{valid} values that $y_i$ may take (this set may depend on $s, x_1,\ldots , x_i$ and $y_1,\ldots , y_{i-1}$. As an example, on a page-fault, a paging algorithm must evict a page which is currently in its cache.).  An algorithm is required to compute a valid answer in each round. 

In \cref{def}, the initial state, $s$, is part of the input and known to the online algorithm before the first request $x_1$ arrives. Note that the initial state is not counted as a request itself. For the $k$-server problem on a metric space $M$, the initial state is a placement of the $k$ servers in $M$. For the bin packing problem, there is only one possible initial state (the empty state). For problems where there is only one possible initial state, we will usually omit it and simply write $\sigma=(x_1,\ldots , x_n)$ when specifying an input.


There are several possible definitions of a randomized algorithm with advice. The one we give in \cref{cadef} is as close as possible to the standard definitions used for randomized online algorithms without advice~\cite{BBKTW} and randomized online algorithms in the advice-per-request model~\cite{A2}. In \cref{app:rando}, we discuss some of the alternative definitions and how they relate to \cref{cadef}.

Finally, we remark that the advice is provided on an infinite tape to prevent the algorithm from learning anything from the length of the advice~\cite{A1, A3}. However, all lower bounds proved in this paper also holds even if the algorithm simply receives (at the very beginning) a finite advice string of length $b(n)$ where $n$ is the input length. See \cref{app:aab} for a more detailed discussion of this point.

\paragraph{Parameter of the problem or part of the input?} When applying our main results, \cref{maint1,mmtt}, to a problem $\P$, some care needs to be taken regarding exactly how $\P$ is defined. Take as an example the $k$-server problem. In this paper, we will not consider this problem to be one single problem, but instead a family of problems parameterized by a number, $k$, of servers and a finite metric space $(M,d)$. Let $P_{k,(M,d)}$ denote the $k$-server problem on the metric space $(M,d)$. A $P_{k,(M,d)}$-input $\sigma=(s_0,r_1,\ldots , r_n)$ consists of an initial placement, $s_0$, of the $k$ servers and $n$ requests. The number of possible length $n$ inputs is $\ab{M}^k\cdot \ab{M}^{n}=\ab{M}^{k+n}$, which is $\dri$ since $k$ and $\ab{M}$ are fixed. When applying \cref{maint1} to $P_{k,(M,d)}$, we are allowed to set $c=\log k$ since $k$ is fixed and therefore $\log k$ is a constant not depending on the input $\sigma$.

Regarding the additive constant $\alpha$ in \cref{def}, we do \emph{not} allow $\alpha$ to depend on the initial state $s$ of the input. For several problems, allowing $\alpha$ to depend on $s$ would trivialize the problem (for example, this is the case for the matching problems discussed in \cref{app:bipmath,app:metmatch}). We remark that several papers on the $k$-server problem do allow the additive constant to depend on the initial configuration. While this is important when dealing with unbounded metric spaces, it does not matter when the metric space is finite. Indeed, suppose that a $k$-server algorithm is $c$-competitive on a metric space $(M,d)$ with an additive term $\alpha(s_0)$ depending on the initial placement $s_0$ of the $k$ servers. If the metric space is finite, there are only finitely many possible initial placements of the servers. Thus, we may remove the dependency on $s_0$ by replacing $\alpha(s_0)$ with the constant $\max_{s_0}\alpha(s_0)$. Generally, for any problem $\P$ with a finite number of possible initial states, our requirement that $\alpha$ depends only on $\P$ and not the initial state does not make any difference for the competitive ratio of $\P$-algorithms.

\subsection{Information theoretic preliminaries}
\label{appendix:inft}
We will define the quantities of information theory that we are going to use (see also \cite{inftheory}). In order to simplify notation, we will use the following common conventions:  $0\log 0=0$, $0\log(q/0)=0$, and $0\log (0/q)=0$ for any fixed $q\geq 0$.

Let $\Omega$ be a finite set. In this paper, most probability spaces will be finite (they will consist of a finite number of input sequences). A \emph{probability distribution} over $\Omega$ is a mapping $\mu:\Omega\rightarrow [0,1]$ such that $\sum_{\omega\in\Omega}\mu (\omega)=1$. The \emph{support} of $\mu$ is defined as $\supp(\mu )=\{\omega\in\Omega : \mu(\omega)>0\}$. Formally, a random variable $X$ over a finite set $\Omega$ is a mapping $X:\Omega\rightarrow \mathcal{X}$ from $\Omega$ into some set $\mathcal{X}$. We say that $\mathcal{X}$ is the \emph{range} of $X$. Random variables are always denoted by capital letters and their range by the calligraphic version of that letter. We always assume that the range\footnote{Our definition of the range $\mathcal{X}$ of $X$ is equivalent to what is often called the support of $X$.} of a random variable $X$ over a finite set $\Omega$ is chosen so that for every $x\in\mathcal{X}$, it holds that $\Pr[X=x]>0$.

The \emph{entropy}, $H(X)$, of a random variable $X$ is defined as

\begin{equation}
H(X)=\sum_{x\in\mathcal{X}}\Pr[X=x]\log\paren*{\frac{1}{\prob[X=x]}}.
\end{equation}

One important property of entropy is that $0\leq H(X)\leq \log\ab{\mathcal{X}}$ with equality if and only if $X$ is uniformly distributed. The \emph{joint entropy} $H(X_1,\ldots , X_n)$ of a sequence $X_1,\ldots , X_n$ of random variables is the entropy of the random variable $(X_1,\ldots , X_n)$. Let $Y$ be a random variable. For any fixed $y\in\mathcal{Y}$, we define $H(X\vert Y=y)$ to be the entropy of the random variable $X$ conditioned on the event $Y=y$. The \emph{conditional entropy} $H(X\vert Y)$ is defined as the expected value of $H(X\vert Y=y)$ over all $y\in\mathcal{Y}$:
\begin{equation}
  H(X\vert Y)=\E_y[H(X\vert Y=y)]=\sum_{y\in\mathcal{Y}}\prob[Y=y]H(X\vert Y=y).
\end{equation}
It holds that $H(X,Y)=H(X)+H(Y\vert X)$. The chain rule for conditional entropy states that
\begin{equation}
H(X_1, X_2,\ldots , X_n \vert Y)=\sum_{i=1}^nH(X_i\vert X_1,\ldots , X_{i-1}, Y).
\label{eq:chain}
\end{equation}

The \emph{mutual information} is defined as $I(X;Y)=H(X)+H(Y)-H(X,Y)$. It follows from the chain rule that $I(X;Y)=H(Y)-H(Y\vert X)=H(X)-H(X\vert Y)$. Also, it can be shown that $I(X;Y)\geq 0$.

Let $\mu, \nu :\Omega\rightarrow [0,1]$ be two probability distributions on some finite set $\Omega$. The \emph{Kullback-Leibler divergence} between $\mu$ and $\nu$ is defined as
\begin{equation}
  D_{KL}(\mu\| \nu)=\sum_{\omega\in\Omega}\mu(\omega)\log\paren*{\frac{\mu(\omega)}{\nu(\omega)}}.
\end{equation}
If $\mu(\omega)=0$ for some $\omega\in\Omega$, then the term in the sum corresponding to $\omega$ is zero according to our conventions. On the other hand, if $\nu(\omega)=0$ and $\mu(\omega)>0$ for some $\omega\in\Omega$, then $D(\mu\| \nu)$ is undefined. In particular, $D(\mu\| \nu)$ is only defined if $\supp(\mu)\subseteq\supp(\nu)$.

\paragraph{Pinsker's inequality.} Let $\mu$ and $\nu$ be two probability distributions on a finite set $\Omega$. Define the \emph{$L_1$-distance} between $\mu$ and $\nu$ as 
\begin{equation}
\|\mu-\nu\|_{1}=\sum_{x\in\Omega}\ab{\mu(x)-\nu(x)}.
\label{L1dist}
\end{equation}
Let $f$ be a function from $\Omega$ into $\mathbb{R}$ and let $\|f\|_\infty = \max_{x\in\Omega} \ab{f(x)}$. Then, 
\begin{equation}
\big| \mathbb{E}_{\mu}[f(x)]-\mathbb{E}_{\nu}[f(x)]\big|=\ab{\sum_{x\in\Omega}f(x)\mu (x)-\sum_{x\in\Omega}f(x)\nu(x)}\leq \|f\|_{\infty}\cdot \|\mu-\nu\|_{1}.
\label{etv}
\end{equation}
Pinsker's inequality (\cite[Lemma 11.6.1]{inftheory})\footnote{Some care needs to be taken when comparing different statements of Pinsker's inequality. Please note that we measure the Kullback-Leibler divergence in bits (whereas it is often measured in nats). Furthermore, the inequality is sometimes stated for a different distance measure which is half the $L_1$-distance.} says this if the Kullback-Leibler divergence between $\mu$ and $\nu$ is small, then $\mu$ and $\nu$ must be close in $L_1$-norm:
\begin{align}
\notag \|\mu-\nu\|_{1}&\leq \sqrt{2\ln 2\cdot D_{\KL}(\mu\|\nu)}\\
\label{pinsker}&=\sqrt{\ln 4\cdot D_{\KL}(\mu\|\nu)}.
\end{align}

\paragraph{The function $K_y(x)$.}

We define a family of functions related to the Kullback-Leibler divergence.
\begin{definition}
\label{def:k}
For $0<y< 1$, define $K_y:[0,1]\rightarrow\mathbb{R}$ by \begin{equation}K_y(x)=x\log\left(\frac{x}{y}\right)+(1-x)\log\left(\frac{1-x}{1-y}\right).\end{equation}
\end{definition}
 Note that $K_y(x)$ is the Kullback-Leibler divergence $D(\mu\| \nu)$ between $\mu$ and $\nu$ where $\mu$ (resp. $\nu$) is the Bernoulli distribution with parameter $x$ (resp. $y$). 

For any fixed $y$, the function $K_y(x)$ is convex. Furthermore, it is decreasing for $x\in[0,y]$ and increasing for $x\in[y,1]$. Also, $K_y(0)=\log (1/(1-y))$, $K_y(y)=0$, and $K_y(1)=\log (1/y)$. We will often need to consider inverse functions of $K_y$. Unfortunately, there is no standard notation for such functions. We define $K_{y,l}^{-1}$ to be the inverse of $K_{y}\vert_{[0,y]}$ and $K_{y,r}^{-1}$ to be the inverse of $K_{y}\vert_{[y,1]}$. In order to simplify some statements, we use the following conventions: $K_{y,l}^{-1}(t)=0$ if $t>K_y(0)=\log(1/(1-y))$ and $K_{y,r}^{-1}(t)=1$ if $t>K_y(1)=\log (1/y)$.

Recall that for $q\in\{2,3,\ldots \}$, the \emph{$q$-ary entropy function} $h_q:[0,1]\rightarrow [0,1]$ is defined as $h_q(x)=x\log_{q}(q-1)-x\log_q\paren*{x}-(1-x)\log_q\paren*{1-x}$. A straightforward calculation reveals that $K_{1/q}(x)=(1-h_q(1-x))\log q$ for $q\in\mathbb{N}, q\geq 2$.

The following lemma shows that if we have an upper bound on the Kullback-Leibler divergence between $\mu$ and $\nu$, then we can bound how much $\mu(\omega)$ can differ from $\nu(\omega)$.  
 
\begin{lemma}
\label{Klem}
Let $q\geq 2$ and let $\mu,\nu:\{x_1,x_2,\ldots, x_q\}\rightarrow [0,1]$ be two probability distributions such that $0<\nu (i)<1$ for all $i$. Assume that $D_{KL}(\mu\| \nu)\leq d$ for some $d\geq 0$. Then, for any $1\leq i\leq q$, it holds that 
\begin{equation}
\mu(i)\geq K_{\nu (i),l}^{-1}(d),
\label{Klem1}
\end{equation}
and 
\begin{equation}
\mu(i)\leq K_{\nu (i),r}^{-1}(d).
\label{Klem2}
\end{equation}
\end{lemma}
\begin{proof}
Fix $i$ such that $1\leq i\leq q$. Then,
\begin{align*}
D(\mu\| \nu)&=\sum_{1\leq j\leq q} \mu(j)\log\left(\frac{\mu(j)}{\nu(j)}\right)\\
&=\mu(i)\log\left(\frac{\mu(i)}{\nu(i)}\right)+\sum_{j\neq i}\mu(j)\log\left(\frac{\mu(j)}{\nu(j)}\right)\\
&\geq \mu(i)\log\left(\frac{\mu(i)}{\nu(i)}\right)+(1-\mu(i))\log\left(\frac{1-\mu(i)}{1-\nu(i)}\right)\tag*{\{Log sum inequality\}}\\
&=K_{\nu(i)}(\mu(i)).
\end{align*}
Thus, $K_{\nu(i)}(\mu(i))\leq d$. We will show that (\ref{Klem1}) is true. Note that (\ref{Klem1}) is trivially true if $\nu(i)\leq \mu(i)$ since the range of $K^{-1}_{\nu(i),l}$ is $[0, \nu(i)]$. Similarly, if $d>K_{\nu(i)}(0)$, then $K^{-1}_{\nu(i),l}(d)=0$ and so (\ref{Klem1}) is also trivially true. Assume therefore that $\mu(i)<\nu(i)$ and that $d\leq K_{\nu(i)}(0)$. In this case, it follows from $K_{\nu(i)}(\mu(i))\leq d$ that $\mu(i)\geq K_{\nu(i),l}^{-1}(d)$ since $K_{\nu(i)}$ is decreasing on $[0,\nu(i)]$ and since $K^{-1}_{\nu(i),l}$ is the inverse of $K_{\nu(i)}$ on this interval. By a symmetric argument, the inequality (\ref{Klem2}) also holds.
\end{proof}

\begin{center}

\end{center}

\subsection{Advice complexity lower bounds via Yao's principle.}

Yao's principle shows that we can prove lower bounds on randomized online algorithms with advice by proving lower bounds on deterministic algorithms with advice against a fixed input distribution. Please note that this direction of Yao's principle does not rely on the minimax theorem or other game-theoretical results (it simply relies on interchanging the order of summation).
\begin{theorem}[Yao \cite{Yao}]
\label{Yao}
Let $\P$ be a minimization problem, let $b(n)$ be a function and let $c>1$. Assume that for every $\varepsilon>0$ and every $\alpha$, there is an input distribution $\pea$ such that if $\DET$ is a deterministic $\P$-algorithm with advice complexity $b(n)$ then \begin{equation}\E_{\pea}[\DET(\sigma)]\geq (c-\varepsilon)\cdot \E_{\pea}[\OPT(\sigma)]+\alpha. \label{yaoeq}\end{equation}

Then, the competitive ratio of every randomized $\P$-algorithm with advice complexity $b(n)$ is at least $c$.
\end{theorem}
\begin{proof}[Proof (following \cite{BLSmetric}).]
Assume by way of contradiction that there exists a randomized $c'$-competitive algorithm $\ra$ with advice complexity $b(n)$ for some $c'<c$. By \cref{cadef}, $\ra$ defines a probability distribution $\mu$ over deterministic algorithms with advice complexity at most $b(n)$. Thus, we may write $\E[\ra (\sigma)]=\E_{\DET\sim\mu}[\DET(\sigma)]$ for the expected cost incurred by $\ra$ on a fixed input $\sigma$. By definition, there must exist a constant $\alpha'$ such that $\E[\ra(\sigma)]=\E_{\DET\sim\mu}[\DET(\sigma)]\leq c'\cdot\OPT(\sigma)+\alpha'$ for every input sequence $\sigma$. 

Choose $\varepsilon=c-c'$ and $\alpha=\alpha'+1$. Let $\pea$ be the probability distribution satisfying (\ref{yaoeq}) for this choice of $\varepsilon$ and $\alpha$. Since $\ra$ is $c'$-competitive, we have that $E_{\sigma\sim \pea}[\E_{\DET\sim\mu}[\DET(\sigma)]]\leq c'\E_{\sigma\sim \pea}[\OPT(\sigma)]+\alpha'$. On the other hand, since costs are by \cref{def} non-negative,
\begin{align*}
\E_{\sigma\sim \pea}[\E_{\DET\sim\mu}[\DET(\sigma)]]&=\E_{\DET\sim\mu}[\E_{\sigma\sim \pea}[\DET(\sigma)]]\tag*{\{Tonelli's theorem\}}\\
&\geq \E_{\DET\sim\mu}[(c-\varepsilon)\E_{\sigma\sim\pea}[\OPT(\sigma)]+\alpha]\tag*{\{by (\ref{yaoeq})\}}\\
&=\E_{\DET\sim\mu}[c'\E_{\sigma\sim \pea}[\OPT(\sigma)]+\alpha'+1]\\
&=c'\E_{\sigma\sim \pea}[\OPT(\sigma)]+\alpha'+1,
\end{align*}
which is a contradiction.
\end{proof}

\section{Direct product theorems for online computation with advice}
The main results of this section are two direct product theorems, namely \cref{mainthm,dpmartin}. Before proving these theorems, we need to define an $r$-round input distribution.

\subsection{Definition of \texorpdfstring{$r$}{r}-round input distributions}
Recall that informally, an $r$-round input distribution is an input distribution which can naturally be split up into $r$ rounds such that in each round, the requests which are revealed are chosen independently from the requests in previous rounds. Furthermore, we are interested in the case where one can associate to each round a cost function such that the total cost of the output computed by an algorithm is simply the sum of the costs incurred in each round. We remark that a round can consist of more than one request. 

Given $\P$-inputs $\sigma_1,\ldots , \sigma_r$, we define $\sigma=\sigma_1\ldots \sigma_r$ to be the $\P$-input obtained by concatenating the requests of the $r$ inputs and using the same initial state as $\sigma_1$. For example, if $\sigma_1=(s,x_1,\ldots , x_n)$ and $\sigma_2=(s',x_1',\ldots , x_{n'}')$, then $\sigma_1\sigma_2=(\init,x_1,\ldots , x_n, x_1',\ldots , x_{n'}')$. We are now ready to formally define an $r$-round input distribution. We will also define a number of random variables associated to such a distribution.
\begin{definition}[$r$-round input distribution]
\label{rround}
Let $\P$ be a minimization problem and let $r\in\mathbb{N}$. For each $1\leq i\leq r$, let $I_i$ be a finite set of $\P$-inputs such that the following holds: If $\sigma_1,\ldots, \sigma_r$ are such that $\sigma_i\in I_i$ for $1\leq i\leq r$, then $\sigma=\sigma_1\sigma_2\ldots \sigma_r$ is a valid $\P$-input. Furthermore, let $I^r=I_1\times\cdots\times I_r=\{\sigma_1\ldots\sigma_i\ldots\sigma_r\mid\sigma_i\in I_i \text{ for $1\leq i\leq r$}\}$.

For each $1\leq i\leq r$, let $\cost_i$ be a function which maps an output $\gamma$ computed for an input $\sigma\in I^r$ to a non-negative real number $\cost_i(\gamma,\sigma)$. We say that $\cost_i$ is the $i$th round \emph{cost function}.

Let $p_i:I_i\rightarrow [0,1]$ be a probability distribution over $I_i$ and let $p^r:I^r\rightarrow [0,1]$ be the (pro\-duct) probability distribution which maps $\sigma_1\sigma_2\ldots\sigma_r\in I^r$ into $p_1(\sigma_1)p_2(\sigma_2)\cdots p_r(\sigma_r)$. We say that $p^r$ (together with the associated cost functions $\cost_i$) is an \emph{$r$-round input distribution.} For $1\leq i\leq r$, we say that $p_i$ is the $i$th round input distribution of $p^r$.
\qeddef
\end{definition}

Given an $r$-round input distribution and a fixed deterministic algorithm, we introduce a number of random variables which are useful when analyzing how the algorithm performs against the $r$-round input distribution. These random variables will be used throughout the paper. Recall that $I^r$ is by \cref{rround} a finite set, and so these random variables are finite.
\begin{definition}
\label{rround2}

Let $\P$ be a minimization problem and let $r\in\mathbb{N}$. Let $p^r:I^r\rightarrow [0,1]$ be an $r$-round input distribution with associated cost functions $\cost_i$ and with $i$th round input distributions $p_i$. Let $\ALG$ be a deterministic $\P$-algorithm with advice. For $\sigma\in I^r$, let $\varphi(\sigma)$ denote the advice read by $\ALG$ on input $\sigma$. Assume that $\ab{\varphi(\sigma)}\leq b$ for all $\sigma\in I^r$, that is, assume that $\ALG$ reads at most $b$ bits of advice when the input belongs to $I^r$. We define the following random variables:

\begin{itemize}
\item For $1\leq i\leq r$, let $X_i$ be the requests\footnote{The variable $X_1$ will in addition to the requests revealed in round $1$ also contain the initial state of the input.} revealed in round $i$. Formally, the random variable $X_i:I^r\rightarrow I_i$ maps $\sigma=\sigma_1\ldots\sigma_r\in I^r$ to $\sigma_i\in I_i$. Furthermore, let $X=X_1\ldots X_r$.
\item Let $B$ be the advice bits read by $\ALG$. Formally, the random variable $B:I^r\rightarrow \{0,1\}^{b}$ maps $\sigma$ to the advice $\varphi(\sigma)\in\{0,1\}^{b}$ read by $\ALG$ on input $\sigma$.
\item For $1\leq i\leq r$, let $\cost_i(\ALG)$ be the $i$th round cost function applied to the output computed by $\ALG$. Formally, the random variable $\cost_i(\ALG):I^r\rightarrow \mathbb{R}$ maps $\sigma\in I^r$ to $\cost_i(\gamma,\sigma)\in\mathbb{R}$ where $\gamma$ are the answers computed by $\ALG$ when given $\sigma$ as input. Furthermore, let $\cost(\ALG)$ be the total cost incurred by $\ALG$, that is, $\cost(\ALG)=\sum_{i=1}^r \cost_i(\ALG)$.
\item Let $W_i=(X_1,\ldots , X_{i-1}, B)$ be the information available to $\ALG$ when the $i$th round begins. This information consists of all previous requests and the advice read by $\ALG$. Formally, $W_i:I^r\rightarrow I_1\times\cdots\times I_{i-1}\times \{0,1\}^{b(n)}$ maps $\sigma\in I^r$ to $(\sigma_1,\sigma_2,\ldots , \sigma_{i-1}, \varphi(\sigma))$.
\end{itemize}
Finally, for every $1\leq i\leq r$ and every $w\in\mathcal{W}_i$, we define the conditional $i$th round input distribution $\piw:I_i\rightarrow [0,1]$ as follows\footnote{Recall that we are assuming that the range $\mathcal{W}_i$ of the finite random variable $W_i$ is chosen to be smallest possible, that is, if $w\in\mathcal{W}_i$, then $\Pr[W_i=w]>0$.}:
\begin{align*}
\piw(x)=p_i(x\vert W_i=w)=\frac{\prob(X_i=x, W_i=w)}{\prob (W_i=w)}.
\end{align*}
\qeddef
\end{definition}
Even if $\ALG$ is an algorithm without advice, we will still use the random variables introduces in \cref{rround2}. In this case, $B$ is the empty string and hence $W_i$ is simply the \emph{history} of requests revealed in previous rounds. This will be relevant when proving and applying \cref{dpmartin}.

When proving lower bounds via $r$-round input distributions, we always assume that the algorithm receives all of the advice at the very beginning. This only strengthens the lower bounds. However, some care has to be taken due to the fact that the length (number of requests) of an input drawn from an $r$-round input distribution is a random variable. By \cref{rround}, there is only a finite number of possible request sequences in each round and, hence, there exists an upper bound on the total length of inputs from the $r$-round distribution. Thus, the value $b$ (the upper bound on the number of advice bits read by $\ALG$) in \cref{rround2} is well-defined. When defining $B$ and $W_i$, we assume that the algorithm $\ALG$ receives a $b$-bit advice string at the beginning. Please note that even if $\sigma$ is such that $\ALG$ does not actually read $b$ bits of advice when serving $\sigma$, we still assume that $\ALG$ receives the first $b$ bits of the advice-tape prepared by its oracle for $\sigma$ (again, this can only help the algorithm).

\subsection{An information-theoretic direct product theorem}
We prove an information theoretical direct product theorem by extending the high-entropy lower bound technique due to Emek et al.~\cite{A2}. We will make use of the notation of \cref{rround,rround2}.
\begin{theorem}
\label{mainthm}
\label{m:mainthm}
Let $\P$ be a minimization problem and let $p^r$ be an $r$-round input distribution with associated cost functions $\cost_i$. Furthermore, let $\ALG$ be a deterministic algorithm reading at most $b$ bits of advice on every input in the support of $p^r$.
Assume that there exists a convex and decreasing function $f:[0,\infty]\rightarrow\mathbb{R}$ such that for every $1\leq i\leq r$ and $w\in\mathcal{W}_i$, the following holds:
\begin{equation}
\E [\cost_i(\ALG)\vert W_i=w]\geq f\left(D_{KL}(p_{i\vert w}\| p_i)\right).
\label{mainthmas}
\end{equation}
Then, $\E[\cost(\ALG)]\geq  r f(b/r)$.
\end{theorem}

Before sketching the proof of \cref{m:mainthm}{}, we informally discuss the theorem. Recall that $W_i$ is the information available to the algorithm when round $i$ begins (the advice read and the history of previous requests). Without any advice or knowledge of the history, the probability of $x\in \mathcal{X}_i$ being selected as the round $i$ request sequence is $p_i(x)$. However, this probability may change given that the algorithm knows $W_i$ (in the most extreme case, the advice could specify exactly the request sequence in round $i$). For any fixed $w\in\mathcal{W}_i$, the probability of $x$ being selected in round $i$ given that $W_i=w$ is denoted $\piw (x)$. Assumption (\ref{mainthmas}) informally means that the closer $\piw$ and $p_i$ are to each other, the better a lower bound we must have on the expected cost incurred by $\ALG$ in round $i$ given that $\ALG$ knows $w$. We remark that the convexity assumption on $f$ is automatically satisfied in most applications. The conclusion of \cref{mainthm}{} essentially says that under these assumptions, an algorithm with $b$ bits of advice for the entire input can do no better than an algorithm with $b/r$ bits of advice for each individual round. In particular, this will allow us to conclude that $\Omega(r)$ bits of advice are needed to get a non-trivial improvement over having no advice at all. 

\begin{proof}[Proof (of \cref{mainthm})]

First of all, note that $H(B)$, the entropy of $B$, is at most $b$ since the algorithm $\ALG$ reads at most $b$ bits of advice. It follows that 
\begin{equation}
I(X;B)=H(B)-H(B\vert X)\leq b.
\label{atmostb}
\end{equation}
On the other hand, using (in the third equality below) that $X_1,\ldots , X_r$ are independent random variables and the chain rule of conditional entropy, we get
\begin{align}
\notag I(X; B) &= I(X_1,X_2,\ldots , X_r ; B)\\
\notag &=H(X_1,X_2,\ldots , X_r)-H(X_1,X_2,\ldots , X_r\vert B)\\
\notag &=\sum_{i=1}^rH(X_i) - \sum_{i=1}^rH(X_i\vert X_{1},\ldots , X_{i-1}, B)\\
\notag &=\sum_{i=1}^r I(X_i ; B, X_{1}, X_2,\ldots , X_{i-1})\\
&=\sum_{i=1}^r I(X_i; W_i).
\label{Ichain}
\end{align}
Combining (\ref{Ichain}) and (\ref{atmostb}) gives
\begin{equation}
I(X; B)=I(X_1;W_1)+I(X_2; W_2)+\cdots + I(X_r;W_r)\leq b. 
\label{crux}
\end{equation}
Fix $i$ such that $1\leq i\leq r$. By definition, for any $w\in\mathcal{W}_i$,
$$\piw(x)=p_i(x\vert W_i=w)=\frac{\prob(X_i=x, W_i=w)}{\prob (W_i=w)}.$$
Note that $\piw$ is, for each fixed $w\in\mathcal{W}_i$, a probability distribution satisfying $\supp(\piw)\subseteq \supp (p_i)$. Thus, the Kullback-Leibler divergence between $\piw$ and $p_i$ is well-defined and finite.
Using the law of total expectation, the assumption (\ref{mainthmas}), and Jensen's inequality, we get that
\begin{align}
\E[\cost_i(\ALG)]&=\E_w[\E[\cost_i(\ALG)\vert W_i=w]]\notag \\ &\geq \E_w\big[f\big(D_{\KL}(\piw\| p_i)\big)\big]\notag\\
&\geq f\left(\E_w[D_{\KL}(\piw\| p_i)]\right).\label{mainthmckl}
\end{align}
As the following calculation shows, the mutual information $I(X_i; W_i)$ is the expected Kullback-Leibler divergence between between $\piw$ and $p_i$ over all $w\in\mathcal{W}_i$:
\begin{align}
I(X_i; W_i)&=\;\;\sum_{\mathclap{x\in I_i,\, w\in\mathcal{W}_i}}\,\Prob(X_i=x,W_i=w)\cdot\log\paren*{\frac{\Prob(X_i=x, W_i=w)}{\Prob (W_i=w)\Prob (X_i=x)}}\notag\\
&=\;\;\sum_{\mathclap{x\in I_i,\, w\in\mathcal{W}_i}}\,\Prob(W_i=w)\cdot \piw(x)\cdot\log\paren*{\frac{\piw(x)}{p_i(x)}}\notag\\
&=\sum_{w\in \mathcal{W}_i}\Prob(W_i=w)\cdot D_{\KL}(\piw\| p_i)\notag\\
&=\E_w[D_{\KL}(\piw\| p_i)].\label{mainthmmikl}
\end{align}
Combining (\ref{mainthmckl}) and (\ref{mainthmmikl}) and using linearity of expectation, we get that
\begin{align*}
\E[\cost(\ALG)]&=\sum_{i=1}^r\E[\cost_i(\ALG)]\\
&\geq \sum_{i=1}^r f\big(\E_w[D_{\KL}(\piw \| p_i)]\big)\\
&= \sum_{i=1}^r f(I(X_i; W_i)).
\end{align*}
We claim that the sum $\sum_{i=1}^rf(I(X_i; W_i))$ is at least $rf(b/r)$. In order to prove this claim, suppose that $h_1,\ldots , h_r$ are non-negative real numbers such that $\sum_{i=1}^r h_i\leq b$. Then, by Jensen's inequality and the assumption that $f$ is convex and decreasing, it follows that
\begin{displaymath}
f\left(\frac{b}{r}\right)\leq f\left(\frac{\sum_{i=1}^r h_i}{r}\right)\leq\frac{1}{r}\sum_{i=1}^r f(h_i).
\end{displaymath}
Thus, $rf(b/r)\leq \sum_{i=1}^{r}f(h_i)$. Since $I(X_1; W_1), \ldots , I(X_r, W_r)$ are non-negative real numbers which according to (\ref{crux}) sums to at most $b$, we conclude that $\E[\cost(\ALG)]\geq rf(b/r)$.
\end{proof}
For a simple example of how to apply \cref{mainthm}, we refer to \cref{thmreprove}.
\subsection{A martingale-theoretic direct product theorem}
We will now provide an alternative direct product theorem using the Azuma-Hoeffding inequality.
\begin{theorem}
\label{dpmartin}
Let $\P$ be a minimization problem and let $p^r$ be an $r$-round input distribution with associated cost functions $\cost_i$.
Assume that there exist constants $t,s\geq 0$ such that for every deterministic $\P$-algorithm $\DET$ without advice, it holds that for every $i$ and every history $w\in\mathcal{W}_i$,
\begin{align}
\E[\cost_i(\DET)\vert W_i=w]&\geq t,\label{asume}\\
\E [(\cost_i(\DET)-t)^2\vert W_i=w]&\leq s^2.\label{asumv}
\end{align}
Then, for every algorithm $\ALG$ reading at most $b$ bits of advice and every $\varepsilon>0$,
\begin{equation}
\label{asumconc}
\Pr \left[\text{cost}(\ALG)\leq (t-\varepsilon)r\right]\leq \exp_2\left(b-K_{\mbox{\Large$\frac{\gamma}{1+\gamma}$}}\left(\frac{\alpha+\gamma}{1+\gamma}\right)r\right),
\end{equation}
where $\gamma=\dfrac{s^2}{t^2}, \alpha=\dfrac{\varepsilon}{t}$, $\exp_2(x)=2^x$, and $K_y(x)=x\log\left(x/y\right)+(1-x)\log\left((1-x)/(1-y)\right)$.
\end{theorem}
It is very important to note that (\ref{asume}) and (\ref{asumv}) must only hold for an algorithm \emph{without advice}. In this case, $W_i$ is simply the history of requests in the previous rounds (and does not contain any advice about the input). Assumption (\ref{asume}) says that the expected cost of a deterministic algorithm (without advice) should be at least $t$ in each round, no matter what has happened in previous rounds. Regarding assumption (\ref{asumv}), note that if $M$ is an upper bound on the cost incurred by any algorithm in a single round, then (\ref{asumv}) is satisfied by taking $s=M$. However, in some cases, it is possible to obtain a better estimate and thereby a better advice complexity lower bound. Azuma-Hoeffding's inequality essentially shows that under these assumptions, the probability that a single fixed algorithm without advice incurs a total cost smaller than $(t-\varepsilon)r$ decreases exponentially in $r$. Viewing an algorithm with $b$ bits of advice as $2^b$ algorithms without advice then gives the result. 

Note that if $b<K_{\mbox{$\frac{\gamma}{1+\gamma}$}}\left(\frac{\alpha+\gamma}{1+\gamma}\right)r$ then the probability in (\ref{asumconc}) is strictly smaller than $1$. Thus, by the probabilistic method, there exists an input $\sigma$ such that the algorithm with $b$ bits of advice incurs a cost of at least $(t-\varepsilon)r$. In particular, we again get that in order to achieve a non-trivial improvement over an algorithm without advice, it must be the case that $b=\Omega(r)$.
\begin{proof}[Proof (of \cref{dpmartin})]
Let $\DET$ be a deterministic algorithm without advice. Recall that for $1\leq j\leq r$, the random variable $\cost_j(\DET)$ is the cost incurred by $\DET$ in round $j$ when the input is drawn from the $r$-round input distribution $p^r$. Let $V_0=0$ and for $1\leq i\leq r$, let $V_i=i\cdot t-\sum_{j=1}^i\cost_j(\DET)$. We claim that $V_0, V_1, V_2,\ldots, V_r$ is a supermartingale. To this end, note that for $0\leq i<r$,
\begin{align*}
\E[V_{i+1}\vert V_1,\ldots , V_i]&=\E[V_i+t-\cost_{i+1}(\DET)\vert V_1,\ldots , V_i]\\
&=V_i+t-\E[\cost_{i+1}(\DET)\vert V_1,\ldots , V_i].
\end{align*}
Since we assume that $\E[\cost_{i+1}(\DET)\vert W_{i+1}=w]\geq t$ for every $w\in \mathcal{W}_{i+1}$, it follows that the conditional expectation $\E[\cost_{i+1}(\DET)\vert V_1,\ldots , V_i]$ is a random variable which is always at least $t$. Thus, $\E[V_{i+1}\vert V_1,\ldots , V_i]\leq V_i$, which proves that $V_0,V_1,\ldots , V_r$ is a supermartingale.

Since costs are always non-negative, we have that $V_{i+1}-V_{i}=t-\cost_{i+1}(\DET)\leq t$. Furthermore, from (\ref{asumv}) we get that $\E[(V_{i+1}-V_{i})^2\vert V_1,\ldots , V_i]=\E[(t-\cost_{i+1}(\DET))^2\vert V_1,\ldots , V_i]\leq s^2$. Thus, by the Azuma-Hoeffding inequality (see \cref{azuma} in \cref{app:sec:az}),
\begin{align*}
\Pr[\cost(\DET)\leq (t-\varepsilon)r)]&=\Pr[rt-\cost(\DET(X))\geq \varepsilon r]\\
&=\Pr\left[V_r\geq \varepsilon r\right]\\
&\leq \exp_2\left(-K_{\mbox{\Large$\frac{\gamma}{1+\gamma}$}}\left(\frac{\alpha+\gamma}{1+\gamma}\right)r\right).
\end{align*}
Let $\ALG$ be a deterministic algorithm reading at most $b$ bits of advice for every input in $\supp (p^r)$. Recall that $\ALG$ can be converted into $2^b$ deterministic algorithms without advice. Above, we computed an upper bound on the probability that a single fixed deterministic algorithm without advice will incur a cost of at most $(t-\varepsilon)r$. Using the union bound, we can get an upper bound on the probability that at least one of the $2^b$ algorithms will incur at most this cost:
\begin{align*}
\Pr[\cost(\ALG)\leq (t-\varepsilon)r)]\leq\exp_2\left(b-K_{\mbox{\Large$\frac{\gamma}{1+\gamma}$}}\left(\frac{\alpha+\gamma}{1+\gamma}\right)r\right).
\end{align*}
\end{proof}

\cref{dpmartin} can be easier to apply than \cref{mainthm}, but it is also less flexible. In this paper, we mainly use \cref{mainthm}, but most of the results could also be proved using \cref{dpmartin}. In some cases (such as \cref{thm:wsgb}), \cref{dpmartin} would yield a slightly weaker bound than \cref{mainthm}. Our main application of \cref{dpmartin} is a lower bound for online graph coloring given in \cref{sec:graphc}.

\section{Lower bounds for \texorpdfstring{$\Sigma$}{Sigma}-repeatable online problems}
In order to prove \cref{maint1}, we first define the repeated version, $\Prs$, of an online problem $\P$. This is a new online problem consisting of some number of (variable-length) rounds of $\P$. Just before each round, a restart takes place and everything is reset to some initial state of $\P$. While we are usually not interested in the problem $\Prs$ itself, we introduce it as a stepping stone for proving lower bounds for $\P$.

Given $\P$-inputs $\sigma_1,\ldots , \sigma_r$ with the same initial state $s$, we define $(\sigma_1;\ldots ;\sigma_r)$ to be a sequence with the requests of the $r$ inputs concatenated and such that the initial state $s$ arrives together with the first request of each $\sigma_i$. For example, if $\sigma_1=(s,x_1,\ldots , x_n)$ and $\sigma_2=(s,x_1',\ldots , x_{n'}')$, then $(\sigma_1;\sigma_2)=(\init, \{s,x_1\},x_2,\ldots , x_n,\{s,x_1'\},x_2',\ldots , x_{n'}')$. Note that if $\sigma=\sigma_1\ldots \sigma_r$ and $\sigma^*=(\sigma_1;\ldots ;\sigma_r)$, then $\sigma$ is an actual $\P$-input whereas $\sigma^*$ technically is not (since $\{s, x_1\}$ is not a valid $\P$-request). When defining the repeated version of an online problem, we need to use $\sigma^*$ in order to make sure that the algorithm knows when one phase ends and another begins (this may not be possible for the algorithm to deduce in $\sigma$).

\begin{definition}
\label{m:def:prs}
\label{def:prs}
Let $\P$ be an online problem, let $S$ be the set of initial states for $\P$, and let $I$ ($I_s$) be the set of all possible request sequences for $P$ (with initial state $s$). Define $\Prs$ to be the online problem with inputs
\mbox{$I^*=\{\sigma^*=(\sigma_1;\sigma_2;\ldots ; \sigma_r)\mid r\geq 1, s\in S, \sigma_i\in I_s\}.$}
An algorithm for $\Prs$ must produce an output $\gamma^*=(\gamma_1,\ldots, \gamma_r)$ where $\gamma_i=(y_1,\ldots , y_{n_i})$ is a valid sequence of answers for the $\P$-input $\sigma_i=(s,x_1,\ldots , x_{n_i})\in I_s$. The score of the output $\gamma^*$ is $\score (\gamma^*,\sigma^*)=\sum_{i=1}^r \score_{\raisebox{-1pt}{$\scriptstyle P$}} (\gamma_i,\sigma_i)$, where $\score_{\raisebox{-1pt}{$\scriptstyle P$}}(\gamma_i,\sigma_i)$ is the score of the $\P$-output $\gamma_i$ with respect to the $\P$-input $\sigma_i$. The optimal offline algorithm for $\Prs$ is denoted $\OPT^*_{\Sigma}$.

\noindent $\P^*_{\lor}$ is defined similarly, except that  $\score(\gamma^*,\sigma^*)=\max\{\score_{\raisebox{-1pt}{$\scriptstyle P$}} (\gamma_1,\sigma_1),\ldots , \score_{\raisebox{-1pt}{$\scriptstyle P$}}(\gamma_r,\sigma_r)\}$. 
\end{definition}
In order to better understand the definition of $\Prs$, it is useful to imagine that after serving the last request of round $i-1$ but before serving the first request of round $i$, the current state is changed to the initial state $s$ (note that the algorithm knows when this happens since in $(\sigma_1;\ldots ;\sigma_r)$, the first request of each $\sigma_i$ is special). It is, however, important to keep in mind that even though a reset occurs when round $i$ begins, the previous rounds are not forgotten. The algorithm has perfect recall of all requests and answers in all previous rounds. Without this recall, advice complexity lower bounds for $\Prs$ would be easy to obtain but would be of little use.

If $\P$ is the $k$-server problem, then an initial state $s$ is a placement of the $k$ servers in the metric space. Thus, in $\Prs$, after serving the last request of round $i-1$ and before serving the first request of round $i$, the $k$ servers automatically return to their initial position specified by $s$. Note that when $\P$ is the $k$-server problem, we can concatenate $\P$-inputs $\sigma_1 ,\sigma_2,\ldots , \sigma_r$ into one long $\P$-input $\sigma=\sigma_1\sigma_2\ldots \sigma_r$. The only difference between serving the $\P$-input $\sigma$ and serving the $\Prs$-input $\sigma^*=(\sigma_1;\ldots ; \sigma_r)$ is that for the $\P$-input $\sigma$, the $k$ servers are not returned to the initial state $s$ when a round ends. However, if the underlying metric space has finite diameter $\Delta$, this difference in the placement of servers when a new round begins can be handled by ensuring that $k\Delta$ is small compared to the total cost incurred during each round. In fact, it turns out that for many online problems, there is a natural reduction from $\Prs$ to $\P$ that essentially preserves all lower bounds. This is formalized in \cref{re}.

\subsection{Lower bounds for \texorpdfstring{$\Prs$}{P*}}
We will now show how to obtain advice complexity lower bounds for $\Prs$ by repeating a hard input distribution $r$ times and using our information theoretical direct product theorem.

\begin{lemma}
\label{techlemma}
Fix $r\geq 1$. Let $\P$ be a minimization problem. Let $p:I_s\rightarrow [0,1]$ be an input distribution, where $I_s$ is a finite set of $P$-inputs with the same initial state $s$ and length at most $n'$. Assume that for every deterministic $\P$-algorithm $\ALG$ without advice, it holds that $\E_{\sigma\sim p}[\ALG(\sigma)]\geq t$. Also, let $M$ be the largest cost that any $\P$-algorithm can incur on any input from $I_s$.

Then, there exists an $r$-round input distribution, $p^r$, over $\Prs$-inputs with at most $rn'$ requests in total, such that any deterministic $\Prs$-algorithm reading at most $b$ bits of advice (on inputs of length at most $rn'$) has an expected cost of at least $r(t-2M\sqrt{b/r})$. Furthermore, $\E_{\sigma^*\sim p^r}[\OPT_{\Sigma}^{\text{*}}(\sigma^*)]=r\E_{\sigma\sim p}[\OPT(\sigma)]$.
\end{lemma}
\begin{proof}

Define $p^r\colon\{(\sigma_1;\ldots ;\sigma_r)\mid \sigma_i\in I_s\}\rightarrow [0,1]$ to be the $\Prs$-input distribution which maps $\sigma^*=(\sigma_1;\sigma_2;\ldots;\sigma_r)$ into $p(\sigma_1)p(\sigma_2)\cdots p(\sigma_r)$. Thus, in each of the $r$ rounds, we draw independently a request sequence from $I_s$ according to $p$. From the definition of $\Prs$, we naturally obtain a cost-function, $\cost_i$, for each round $1\leq i\leq r$. Together with these cost-functions, the input distribution $p^r$ is an $r$-round input distribution for $\Prs$. Note that the $i$th round input distribution $p_i$ is simply $p_i=p$, and that a round of $p^r$ corresponds to a round of $P^*_{\Sigma}$.

Let $\ALG^*$ be a deterministic algorithm for $\Prs$ reading at most $b$ bits of advice on inputs of length at most $rn'$.
Fix $1\leq i\leq r$ and $w\in \mathcal{W}_i$. Let $d=D_{\KL}(\piw \| p_i)$. In order to apply \cref{maint1}, we need a lower bound on $\E [\cost_i(\ALG^*)\vert W_i=w]$ in terms of $d$. By Pinsker's inequality (\ref{pinsker}),
\begin{equation}
\|\piw - p_i\|_{1}\leq\sqrt{d\cdot \ln 4}.
\label{pintechlem}
\end{equation}
Thus, it suffices to bound $\E[\cost_i(\ALG^*)| W_i=w]$ in terms of the $L_1$-distance between $p_{i\vert w}$ and $p_i$. Let $h(d)=\sqrt{d\cdot\ln 4}$. We construct a $P$-algorithm, $\ALG_w$, without advice by hard-wiring $w$ (i.e., the advice and the requests in rounds $1$ to $i-1$) into the $\Prs$-algorithm $\ALG^*$. That is, for an input sequence $\sigma\in\supp (\piw)$, the new algorithm $\ALG_w$ simulates the computation of $\ALG^*$ on $\sigma$ when $W_i=w$ and $\ALG^*$ is given the requests $\sigma$ in round $i$. Note that this is possible since $w$ is fixed, and hence the output of $\ALG^*$ in round $i$ given $\sigma$ as input in this round is completely determined. For all other input sequences, $\ALG_w$ behaves arbitrarily (but does compute some valid output). It follows that $\ALG_w$ is well-defined for all input sequences in $\supp(p_{i})$. Thus, $\ALG_w$ defines a mapping $\sigma\mapsto \ALG_w(\sigma)$ on $\supp(p_{i})$ such that $\|\ALG_w\|_{\infty}=\max_{\sigma\in \supp(p_i)}\ab{\ALG_w(\sigma)}\leq M$ and such that if $\sigma\in\supp (\piw)\subseteq\supp(p_i)$, then $\ALG_w(\sigma)$ is equal to the cost incurred by $\ALG^*$ if $W_i=w$ and $\sigma$ is given as input in round $i$. It follows that

\begin{align*}
\E[\cost_i(\ALG^*)\vert W_i=w]&=\E_{\sigma\sim \piw}[\ALG_w(\sigma)]\\
&\geq \E_{\sigma\sim p}[\ALG_w(\sigma)]-M\cdot h(d) \tag*{\{By (\ref{pintechlem}) and (\ref{etv})\}}\\
&\geq t-M\cdot h(d).
\end{align*}
Define $f(d)=t-M\cdot \sqrt{\ln 4\cdot d}$. Since $f$ is convex and decreasing, it follows from our direct product theorem (\cref{m:mainthm}{}) that (remember that $p_i=\pea$):
\begin{equation}
\label{techlemmaconc}
\E_{\sigma^*\sim p^r}[\ALG^*(\sigma^*)]=\E[\cost(\ALG^*)]\geq rf(b/r)=r\left(t-M\sqrt{\frac{b \cdot \ln 4}{r}}\right) \geq r\left(t-2M\sqrt{\frac{b}{r}}\right).
\end{equation}
Here, we used that $\sqrt{\ln 4}<2$. Finally, we observe that
\begin{align*}
\E_{\sigma^*\sim p^r}[\OPT^*_{\Sigma}(\sigma^*)]&=\E_{\sigma^*\sim p^r}\left[\sum_{i=1}^r\OPT(\sigma_i)\right]=\sum_{i=1}^r\E_{\sigma^*\sim p^r}[\OPT(\sigma_i)]=\sum_{i=1}^r\E_{\sigma_i\sim p_i}[\OPT(\sigma_i)]\\
&=\sum_{i=1}^r\E_{\sigma_i\sim p}[\OPT(\sigma_i)]=r\E_{\sigma\sim p}[\OPT(\sigma)].
\end{align*}
\end{proof}

\subsection{Definition of repeatable online problems}

\cref{techlemma} shows how to obtain lower bounds for $\Prs$. Of course, we are usually not very interested in the problem $\Prs$ itself. However, for several online problems, there is an obvious reduction from $\Prs$ to $\P$ since it is possible for an adversary to simulate a restart in $\P$. This idea is formalized by the notion of a repeatable online problem.

\begin{definition}
\label{re}
Let $\P$ be an online minimization problem such that for every fixed $\P$-input, there is only a finite number of valid outputs. Furthermore, let $k_1,k_2,k_3\geq 0$. We say that $\P$ is \emph{$\Sigma$-repeatable with parameters $(k_1,k_2,k_3)$} if there exists a mapping $g:I^*\rightarrow I$ with the following properties:
\begin{enumerate}[label=\textcolor{darkgray}{\sffamily\bfseries\mathversion{bold} $\Sigma$\arabic*.},leftmargin=1.75\parindent]
\item For every $\sigma^*\in I^*$, 
\begin{equation}
\ab{g(\sigma^*)}\leq \ab{\sigma^*}+k_1r,
\label{re1}
\end{equation}
where $r$ is the number of rounds in $\sigma^*$.
\item For every deterministic $\P$-algorithm $\ALG$, there is a deterministic $\Prs$-algorithm $\ALG^*$ such that for every $\sigma^*\in I^*$, 
\begin{equation}
\ALG^*(\sigma^*)\leq \ALG(g(\sigma^*))+k_2r,
\label{re2}
\end{equation}
where $r$ is the number of rounds in $\sigma^*$. 
\item For every $\sigma^*\in I^*$,
\begin{equation}
\OPT^*_{\Sigma}(\sigma^*)\geq \OPT(g(\sigma^*))-k_3r,
\label{re3}
\end{equation}
where $r$ is the number of rounds in $\sigma^*$.
\end{enumerate}
\end{definition}
If a problem $\P$ is repeatable with $k_2=k_3=0$, we say that $\P$ is \emph{strictly repeatable}. The definition of $\lor$-repeatable is identical to that of $\Sigma$-repeatable, except that $\Prs$ and $\OPT^*_\Sigma$ are replaced by $\P^*_{\lor}$ and $\OPT^*_{\lor}$. The following simple lemma is useful when working with repeatable online problems.

\begin{lemma}
Let $\P$ be a $\Sigma$-repeatable problem with parameters $(k_1,k_2, k_3)$, let $b$ be a constant, and let $\ALG$ be a $\P$-algorithm with advice. Let $J^*\subseteq I^*$ be a subset of $\Prs$-inputs. Suppose that for every $\sigma\in J^*$, the algorithm $\ALG$ reads at most $b$ bits of advice on the $\P$-input $g(\sigma^*)$. Then there exists a $\Prs$-algorithm $\ALG^*$ such that $\ALG^*$ reads exactly $b$ bits of advice on every input $\sigma^*\in J^*$ and such that $\ALG(g(\sigma^*))\geq \ALG^*(\sigma^*)-k_2r$ for every input $\sigma^*\in J^*$ where $r$ is the number of rounds in $\sigma^*$.
\label{lem:repadv}
\end{lemma}
\begin{proof}
By definition, there exists $2^b$ deterministic algorithms, $\ALG_1,\ldots , \ALG_{2^b}$, without advice such that $\ALG(g(\sigma^*))=\min_i\ALG_i(g(\sigma^*))$ for every input $\sigma^*\in J^*$. For each $1\leq i\leq 2^b$, it follows from (\ref{re2}) in \cref{re} that there exists a deterministic $\Prs$-algorithm $\ALG_i^*$ without advice such that $\ALG_i(g(\sigma^*))\geq \ALG_i^*(\sigma^*)-k_2r$ for every $\sigma^*$. We are now ready to define $\ALG^*$ and the corresponding advice oracle: The algorithm $\ALG^*$ always reads $b$ bits of advice before the first request arrives. This is possible since $b$ is a constant. For an input $\sigma^*\in J^*$, the oracle computes an index $i$ such that $\ALG(g(\sigma^*))=\ALG_i(g(\sigma^*))$. This index $i$ is written onto the advice tape. The algorithm $\ALG^*$ learns $i$ from the advice tape and serves $\sigma^*$ using $\ALG^*_i$. 
It follows that $\ALG(g(\sigma^*))\geq \ALG^*(\sigma^*)-k_2r$ for any input $\sigma^*\in J^*$. 
\end{proof}
We remark that \cref{lem:repadv} will only be used when the length of inputs in $J^*$ is bounded. This explains why the lemma is only stated for the case where $b$ is a constant. Also, note that for inputs outside of $J^*$, the algorithm $\ALG^*$ can behave arbitrarily.
\subsection{Proof of \texorpdfstring{\cref{maint1}}{Theorem 1}}
We are now ready to prove our main theorem for $\Sigma$-repeatable online problems, \cref{maint1}. For convenience, we restate the theorem before giving the proof.
\restateone*
\begin{proof}

Let $\P$ be a $\Sigma$-repeatable online problem with parameters $(k_1, k_2, k_3)$. The lower bound will be obtained via Yao's principle (\cref{Yao}). To this end, fix $\varepsilon'>0$ and $\alpha'$ and let $\ALG$ be an arbitrary deterministic $\P$-algorithm reading $b(n)\in o(n)$ bits of advice. We need to show that there exists a probability distribution $p$ such that $\E_{p}[\ALG(\sigma)]\geq (c-\varepsilon')\E_{p}[\OPT(\sigma)]+\alpha'$. Choose $\varepsilon=\varepsilon'$ and $\alpha=\alpha'+k_2+k_3(c-\varepsilon)+1$. The reason for choosing this particular value of $\alpha$ will become clear later on. For this choice of $\varepsilon$ and $\alpha$, let $\pea$ be a probability distribution for which it holds that $\E_{\pea}[\DET(\sigma)]\geq (c-\varepsilon)\E_{\pea}[\OPT(\sigma)]+\alpha$ for every deterministic algorithm $\DET$ without advice. Without loss of generality, we assume that all inputs in the support of $\pea$ have the same initial state (\cref{lem:wlogsames} in \cref{sec:omitted} justifies this assumption).

$\P$ is $\Sigma$-repeatable, and so for any fixed input in $\supp(\pea)$ there is only a finite number of possible outputs. Since $\supp(\pea)$ is itself finite, this means that there exists a real number $M$ such that, on inputs from $\supp(\pea)$, no valid $\P$-algorithm incurs a cost larger than $M$\footnote{Here, we use that according to \cref{def}, the cost of an output cannot be $\infty$ but must be a real number.}.
Thus, from $\pea$ we obtain for each $r\in\mathbb{N}$ via \cref{techlemma} a probability distribution $\pear$ such that for any $\Prs$-algorithm $\ALG^*$ which on inputs in $\supp(\pear)$ reads at most $b^*\in\mathbb{N}$ bits of advice, it holds that $\E_{\sigma^*\sim \pear}[\ALG^*(\sigma^*)]\geq r\left((c-\varepsilon)\E_{\pea}[\OPT(\sigma)]+\alpha-2M\sqrt{b^* /r}\right)$.


Since $\supp(\pea)$ is finite, there exists a constant $\lea$ such that every $\sigma\in\supp(\pea)$ contains at most $\lea$ requests. It follows that the number of requests in any input sequence in the support of $\pear$ is at most $\lea r$. By (\ref{re1}) in \cref{re}, this implies that the number of requests in $g(\sigma^*)$ is at most $\lea r+k_1r$ for every $\sigma^*\in\supp (\pear)$. Forget for a moment that $b$ is the advice complexity of $\ALG$, and just view $b$ as a function $b:\mathbb{N}\rightarrow\mathbb{N}$. By assumption, $b(n)\in o(n)$. Since $k_1$ and $\lea $ are constants (independent of $r$), this implies that $b(\lea r+k_1r)/r\rightarrow 0$ as $r\rightarrow\infty$. Choose $r$ large enough that $2M\sqrt{\frac{b(\lea r+k_1r)}{r}} \leq 1$.

We have now fixed the value of $r$. Let $b_{r}=b(\lea r+k_1r)$. Note that $b_{r}$ is a fixed integer (since $r$ is fixed) and that the algorithm $\ALG$ will read at most $b_{r}$ bits of advice on any input $g(\sigma^*)$ where $\sigma^*\in \supp(\pear)$. Using \cref{lem:repadv}, we convert the $\P$-algorithm $\ALG$ into a $\Prs$-algorithm $\ALG^*$ such that $\ALG^*$ reads $b_{r}$ bits of advice on every input $\sigma^*\in \supp(\pear)$ and such that $\ALG(g(\sigma^*))\geq \ALG^*(\sigma^*)-k_2r$ for every $\sigma^*\in\supp(\pear)$. The proof is completed by the following calculation:

\begin{align*}
\E_{\sigma^*\sim \pear}[\ALG(g(\sigma^*))]&\geq \E_{\sigma^*\sim \pear}[\ALG^*(\sigma^*)]-k_2r\\
&\geq r\left((c-\varepsilon)\E_{\sigma\sim\pea}[\OPT(\sigma)]+\alpha-2M\sqrt{\frac{b_{r}}{r}}\right)-k_2r\tag*{\{\cref{techlemma}\}}\\
&\geq r\Big((c-\varepsilon)\E_{\sigma\sim\pea}[\OPT(\sigma)]+\alpha-1\Big)-k_2r\tag*{\{by choice of $r$\}}\\
&=(c-\varepsilon)r\E_{\sigma\sim\pea}[\OPT(\sigma)]-r(1+k_2-\alpha)\\
&=(c-\varepsilon)\E_{\sigma^*\sim \pear}[\OPT^*_{\Sigma}(\sigma^*)]-r(1+k_2-\alpha)\tag*{\{\cref{techlemma}\}}\\
&\geq (c-\varepsilon)\E_{\sigma^*\sim \pear}[\OPT(g(\sigma^*))]-r(1+k_2+k_3(c-\varepsilon)-\alpha)\tag*{\{by (\ref{re3})\}}\\
&\geq (c-\varepsilon')\E_{\sigma^*\sim \pear}[\OPT(g(\sigma^*))]+\alpha'.\tag*{\{by choice of $\alpha$\}}
\end{align*}
\end{proof}

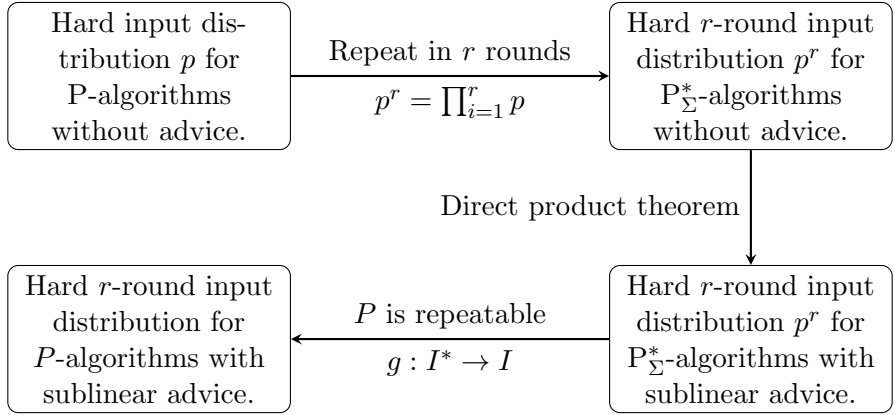
\begin{figure}
\begin{center}
\begin{tikzpicture}[node distance=2cm]

\node (start) [startstop, text width=3.5cm] {Hard input distribution $p$ for $\P$-algorithms without advice.};

\node (repeat) [startstop, text width=3.5cm, right of=start, xshift=6cm] {Hard $r$-round input distribution $p^r$~for $\Prs$-algorithms without advice.};

\node (repeatadv) [startstop, text width=3.5cm, below of=repeat, yshift=-1.5cm] {Hard $r$-round input distribution $p^r$ for $\Prs$-algorithms with sublinear advice.};

\node (startadv) [startstop, text width=3.5cm, below of=start, yshift=-1.5cm] {Hard $r$-round input distribution for $P$-algorithms with sublinear advice.};

\draw [arrow] (start) -- node[anchor=north] {$p^r=\prod _{i=1}^r p$} node[anchor=south] {Repeat in $r$ rounds}  (repeat);
\draw [arrow] (repeat) -- node[anchor=west] {}node[anchor=east]{Direct product theorem} (repeatadv);
\draw [arrow] (repeatadv) -- node[anchor=north] {$g:I^*\rightarrow I$}node[anchor=south] {$P$ is repeatable} (startadv);

\end{tikzpicture}
\end{center}
\caption{Schematic overview of our technique for obtaining advice complexity lower bounds for repeatable online problems. One starts with a hard input distribution $p$ for $P$-algorithms without advice, then repeats it in $r$ rounds to get a hard $r$-round input distribution $p^r$ for $\Prs$. Using a direct product theorem, we get that $p^r$ is also a hard input distribution for algorithms with sublinear advice. These two steps are collected in \cref{techlemma}. Finally, since $\P$ is repeatable, we can obtain a lower bound on $\P$-algorithms with sublinear advice using the reduction $g$.}
\end{figure}
\FloatBarrier

\subsection{Generalized task systems are repeatable}
We seek a simple set of conditions which are sufficient to show that an online problem is $\Sigma$-repeatable. It is not hard to show that any problem which can be modeled as a task system \cite{BLSmetric} is $\Sigma$-repeatable. However, it turns out that one can relax the definition of a task system and still maintain repeatability. Therefore, we introduce \emph{generalized task systems (GTS)}. Essentially, a GTS is a task system with (almost) no restrictions imposed on the distance function. \cref{GTSar} shows that every online problem which can be modeled as a GTS satisfying a certain finiteness condition is $\Sigma$-repeatable. 
\begin{definition}
\label{GTSdefinition}
A \emph{generalized task system (GTS)} is a minimization problem defined by a triple $(\mathcal{S}, \mathcal{T}, d)$ where $\mathcal{S}$ is a set of $N$ \emph{states}, $\mathcal{T}$ is a set of allowable tasks and $d:\mathcal{S}\times \mathcal{S}\rightarrow \mathbb{R}_{\geq 0}$ is an arbitrary function. A \emph{task} is a mapping $t:\mathcal{S}\rightarrow \mathbb{R}_{\geq 0}\cup\{\infty\}$ satisfying that there exists at least one state $s\in\mathcal{S}$ such that $t(s)\neq \infty$.

An input $\sigma=(s_0, t_1, t_2,\ldots , t_n)$ consists of an \emph{initial state} $s_0\in \mathcal{S}$ and $n$ tasks, $t_1,\ldots , t_n,$ where $t_i\in\mathcal{T}$ for each $1\leq i\leq n$. An online algorithm must compute as output a set of states $y=(s_1,\ldots , s_n)$ such that $s_i$ is computed from $(s_0,t_1, t_2,\ldots , t_i)$ and such that $t_i(s_i)\neq \infty$. When the algorithm outputs a state $s_i$ in round $i$, it incurs a cost of  $d(s_{i-1}, s_i)+t_i(s_i)$. We say that $d(s_{i-1},s_i)$ is the \emph{transition cost} and that $t_i(s_i)$ is the \emph{processing cost} of serving the task $t_i$. The total cost of the output $y$ is the sum of the costs incurred in each round.
\end{definition}

If the function $d$ is clear from the context, we will sometimes write $(\gs, \gt)$ instead of $(\gs, \gt , d)$. If for a task $t$ and a state $s$ it holds that $d(s,s)+t(s)=0$, then we say that $s$ is a \emph{haven} against the task $t$. Note that if $d(s,s)>0$, then $s$ cannot be a haven against any task.

Clearly, any task system and any metrical task system (MTS) is also a generalized task system (GTS).
Recall that for any MTS with $N$ states, there is a $O(\text{polylog~} N)$-competitive randomized algorithm~\cite{DBLP:journals/jcss/FakcharoenpholRT04}. It is easy to see that no such result can be achieved for an arbitrary GTS. However, if an online problem $\P$ can be modeled as a GTS, this gives us a lot of information about the (game-theoretical) structure of $\P$: It gives us a notion of a state, it shows that the cost of an output can be decomposed into the costs of each answer, and other similar properties. It will often be important for us to have upper (and lower) bounds on the cost that can be incurred while serving a single task.

\begin{definition}
\label{finiteltsdef}
Let $(\mathcal{S}, \mathcal{T}, d)$ be a GTS. The \emph{max-cost} of $(\gs, \gt, d)$ is the infimum over all real numbers $\Delta$ which for every $t\in\mathcal{T}$ satisfy the following condition:
\begin{align}
\text{For all $s,s'\in\mathcal{S}$, if $t(s')\neq \infty$ then $d(s,s')+t(s')\leq \Delta$.}\label{st1}
\end{align}
The \emph{min-cost} of $(\gs, \gt, d)$ is the supremum over all real numbers $\delta$ which for every $t\in\mathcal{T}$ satisfy the following conditions:
\begin{align}
&\text{For all $s,s'\in\mathcal{S}$, if $s\neq s'$ then $d(s,s')\geq \delta$.}\label{st2}\\
&\text{For all $s\in\mathcal{S}$, either $d(s,s)+t(s)=0$ or $d(s,s)+t(s)\geq \delta$.}\label{st3}
\end{align}
\end{definition}

If $(\gs, \gt)$ has max-cost $\Delta$ then it follows from (\ref{st1}) that the cost incurred by any algorithm in a single round is at most $\Delta$. Similarly, if $(\gs, \gt)$ has min-cost $\delta$, then (\ref{st2}) ensures that if an online algorithm switches to a different state, then it incurs a cost of at least $\delta$. The condition (\ref{st3}) ensures that if an online algorithm stays in the same state in some round, then it either incurs no cost at all or it incurs a cost of at least $\delta$. The min-cost will play an important role later on in \cref{sec:compon}. For now, we only need to consider generalized task systems with bounded max-cost.

\begin{lemma}
Let $\P$ be a minimization problem which can be modeled as a generalized task system with finite max-cost $\Delta$. Then $\P$ is $\Sigma$-repeatable with parameters $(0, \Delta, \Delta)$.
\label{GTSar}
\end{lemma}
\begin{proof}
Let $\P$ be the online problem defined by the generalized task system $(\gs, \gt, d)$. Assume that $\P$ has finite max-cost $\Delta$. We will show that $\P$ is $\Sigma$-repeatable with parameters $(0,\Delta,\Delta)$. Let $I$ (resp. $I^*$) be the set of all input sequences for $\P$ (resp. $\P^*$). An input $\sigma^*=(\sigma_1;\ldots ; \sigma_r)\in I^*$ consists of $r$ rounds. For $1\leq i\leq r$, $\sigma_i=(s_0, t_1^i,t_2^i,\ldots , t_{n_i}^i)$ is a sequence of tasks belonging to $\gt$.

  The mapping $g:I^*\rightarrow I$ maps $\sigma^*=(\sigma_1;\sigma_2;\ldots ; \sigma_r)$ into the request sequence 
\begin{equation}
\sigma=\sigma_1\sigma_2\cdots\sigma_r=\left(s_0,t_1^1,t_2^1,\ldots , t_{n_1}^1, t_1^2,\ldots , t_{n_2}^2,\ldots , t_{1}^r,\ldots , t_{n_r}^r\right).
\end{equation}
Note that $\sigma$ is obtained from $\sigma^*$ simply by revealing the same tasks in the same order but not performing the reset to the initial state $s_0$ at the beginning of each round $i>1$. It is clear that (\ref{re1}) of \cref{re} is satisfied with $k_1=0$ since $\ab{g(\sigma^*)}=\ab{\sigma^*}$. Let $\ALG$ be a deterministic $\P$-algorithm (without advice). Define $\ALG^*$ to be the following $\P^*$-algorithm: $\ALG^*$ serves the $j$th task of $\sigma^*$ by moving to the same state as $\ALG$ does when serving the $j$th task of $g(\sigma^*)$. We will now prove that so defined $\ALG^*$ satisfies (\ref{re2}). When a reset occurs at the end of a round, $\ALG^*$ will be moved from its current state to the initial state of the next round (at no cost). This means that when serving the very first task after a reset, $\ALG^*$ might incur a larger cost than $\ALG$ does. However, by assumption, the cost incurred by $\ALG^*$ for serving the first request of the next round is at most $\Delta$. For the remaining tasks in the round, $\ALG^*$ incurs exactly the same cost as $\ALG$. Thus, $\ALG^*(\sigma^*)\leq \ALG(g(\sigma^*))+k_2r$ where $k_2=\Delta$ is a constant independent of $\ALG^*$ and $r$. 

The proof of (\ref{re3}) is similar. $\OPT$ can serve the tasks of $g(\sigma^*)$ in the same way as $\OPT^*_{\Sigma}$ serves the corresponding tasks of $\sigma^*$. The only problem is that when a reset occurs, $\OPT^*_{\Sigma}$ is moved to the initial state of the new round for free, which might allow $\OPT^*_{\Sigma}$ to serve the very first task after the reset at a lower cost than $\OPT$. However, this difference in cost for serving the first request of a round can be at most $\Delta$. Thus, $\OPT^*_{\Sigma}(\sigma^*)\geq\OPT(g(\sigma^*))-k_3r$ where $k_3=\Delta$.
\end{proof}

\section{Compact online problems}
\label{sec:compon}
It is natural to ask when a lower bound on randomized online algorithm without advice is witnessed by a family of input distributions satisfying the conditions of \cref{maint1} (we formalize this question in \cref{compdef}). It is well-known that there exist problems with an infinite number of requests and answers in each round for which an optimal lower bounds on randomized algorithms cannot be obtained via Yao's principle (see e.g. \cite{BE98b}). However (as we will see in \cref{subsec:lts}), for several important online problems, one can show that a lower bound on randomized algorithms must always be witnessed by a family of input distributions which are compatible with \cref{maint1}.

\subsection{Definition of compact online problems and proof of \texorpdfstring{\cref{mmtt}}{Theorem 2}}
\label{sec:defcomp}
\begin{definition}
\label{m:compdef}
\label{compdef}
Let $\P$ be a minimization problem and let $c>1$ be a constant such that the expected competitive ratio of every randomized $\P$-algorithm is at least $c$. We say that $\P$ is \emph{compact} if for every $\varepsilon>0$ and every $\alpha\geq 0$, there exists an input distribution $\pea$ with finite support such that if $\DET$ is a deterministic online algorithm (without advice), then $\E_{\pea}[\DET(\sigma)]\geq (c-\varepsilon)\cdot \E_{\pea}[\OPT(\sigma)]+\alpha$.
\end{definition}



Informally, an online problem is compact if we can use inputs of bounded length to prove lower bounds on the competitive ratio which are essentially tight. The following two lemmas and \cref{mmtt} follows trivially by combining \cref{maint1} with \cref{compdef}. Nevertheless, we believe that \cref{mmtt} is an important conceptual implication of \cref{maint1}.

\begin{lemma}
\label{comprep1}
Let $\P$ be a compact and $\Sigma$-repeatable online problem, and let $c>1$ be a constant. Assume that every randomized algorithm without advice has a competitive ratio of at least $c$. Then the competitive ratio of every randomized algorithm reading $o(n)$ bits of advice is at least $c$.
\end{lemma}
\begin{proof}
The problem $\P$ is compact, and so for every $\varepsilon>0$ and every $\alpha\geq 0$, there exists a probability distribution $\pea$ with finite support such that $\E_{\pea}[\DET(\sigma)]\geq (c-\varepsilon)E_{\pea}[\OPT(\sigma)]+\alpha$ for every deterministic algorithm $\DET$. Since $\P$ is $\Sigma$-repeatable, it follows from \cref{maint1} that every randomized algorithm with sublinear advice complexity has a competitive ratio of at least $c$.
\end{proof}

\begin{lemma}
\label{comprep2}
Let $\P$ be a compact and $\Sigma$-repeatable online problem and let $c\geq 1$ be a constant. If there exists a (possibly randomized) $c$-competitive algorithm reading $o(n)$ bits of advice then, for all $\varepsilon>0$, there exists a randomized $(c+\varepsilon)$-competitive algorithm without advice.
\end{lemma}
\begin{proof}
Let $\ALG$ be a $c$-competitive $\P$-algorithm reading $o(n)$ bits of advice. Assume by way of contradiction that for some $\varepsilon>0$, there does not exist a randomized $(c+\varepsilon)$-competitive algorithm without advice. By \cref{comprep1}, this implies that the competitive ratio of $\ALG$ must be at least $c+\varepsilon$, a contradiction. 
\end{proof}

\setcounter{theorem}{1}
\begin{theorem}
Let $\P$ be a compact and $\Sigma$-repeatable minimization problem with at most $\dri$ inputs of length $n$, and let $c$ be a constant not depending on $n$. The following are equivalent:
\begin{enumerate}
\item For every $\varepsilon>0$, there exists a randomized $(c+\varepsilon)$-competitive $\P$-algorithm without advice.
\item For every $\varepsilon>0$, there exists a deterministic $(c+\varepsilon)$-competitive $\P$-algorithm with advice~complexity $o(n)$.
\end{enumerate}
\end{theorem}
\setcounter{theorem}{5}
\begin{proof}
Suppose that for every $\varepsilon>0$, there exists a randomized $(c+\varepsilon)$-competitive $\P$-algorithm without advice. Since $\P$ is a minimization problem with at most $\dri$ possible inputs of length $n$, it follows from the derandomization result of B\"ockenhauer et al.\ \cite{Ak-server} (or \cref{thm:derandmin}) that for every $\varepsilon>0$, there exists a $(c+2\varepsilon)$-competitive deterministic $\P$-algorithm with advice complexity $O(\log n)$. Since $\varepsilon$ was arbitrary, this proves the first implication of the theorem.

Suppose that for every $\varepsilon>0$, there exists a deterministic $(c+\varepsilon)$-competitive $\P$-algorithm with advice complexity $o(n)$. Then, by \cref{comprep2}, for every $\varepsilon>0$ there exists a randomized $(c+2\varepsilon)$-competitive $\P$-algorithm without advice. Since $\varepsilon$ was arbitrary, this proves the remaining implication of the theorem.
\end{proof}

\cref{mmtt} is of course only interesting in those cases where one can prove that a problem is compact without actually determining the best possible competitive ratio that a randomized algorithm can achieve. We now sketch how to prove that several important $\Sigma$-repeatable online problems are compact. Interestingly, our proof will rely on the ``upper bound part'' of Yao's principle~\cite{Yao} which is (much) less frequently used than the lower bound part. 

Fix a $\Sigma$-repeatable problem $\P$ such that for every $n$, the number of inputs of length at most $n$ is finite\footnote{Since $\P$ is $\Sigma$-repeatable (\cref{re}), this assumption automatically implies that there is only finitely many different $\P$-algorithms for inputs of length at most $n$.}. Let $c>1$ be a constant such that the expected competitive ratio of every randomized $\P$-algorithm is at least $c$. What does it mean for $\P$ to \emph{not} be compact? It means that there exists an $\varepsilon>0$ and $\alpha\geq 0$ such that the following holds: For every $n'\in\mathbb{N}$ and for every probability distribution $p$ over $\P$-inputs of length at most $n'$, there exists a deterministic algorithm $\DET$ such that $\E_{\sigma\sim p}[\DET(\sigma)]< (c-\varepsilon)\E_{\sigma\sim p}[\OPT(\sigma)]+\alpha$. Recall that, by assumption, there is only a finite number of inputs and outputs for $\P$ of length at most $n'$. This makes it possible to view the problem $\P$ restricted to inputs of length at most $n'$ as a \emph{finite} two-player zero-sum game between an algorithm and adversary. Thus (see \cref{yaofinite} for the full proof), by Yao's principle, we get that there exists a randomized $\P$-algorithm $\ra_{n'}$ such that $\E[\ra_{n'}(\sigma)]< (c-\varepsilon)\OPT(\sigma)+\alpha$ for every $\P$-input of length at most $n'$. Now, if we can somehow show that it is possible to use the algorithms $\ra_1,\ra_2,\ldots$ to obtain a single algorithm $\ra$ which is better than $c$-competitive (on all possible inputs), then the problem at hand must, by contradiction, be compact. In what follows, we will use this strategy to show that all problems which can be modeled as a task system satisfying certain requirements are compact.

\subsection{Lazy task systems are compact}
\label{subsec:lts}
Recall that we introduced generalized task systems in \cref{GTSdefinition} in order to make it easier to show that an online problem is $\Sigma$-repeatable. We will now define a certain subclass of generalized task systems known as \emph{lazy task systems (LTS)}. In a lazy task system, the distance function must satisfy the triangle inequality and it must separate states. \cref{LTZcompact} shows that these conditions on the distance function together with the finiteness of the state space imply that an online problem is compact (and $\Sigma$-repeatable). In \cref{app:modellazy}, we show how to model several well-known online problems as lazy task systems. We note that a metrical task system is always a lazy task system. 
\begin{definition}
\label{def:lts}
A \emph{lazy task system (LTS)} is a generalized task system $(\gs, \gt, d)$ where the function $d$ satisfies the following conditions for all $s,s',s''\in \gs$:
\begin{align}
&\text{$d(s,s')>0$ if $s\neq s'$.}\label{ltsreq1}\\
&\text{$d(s,s')\leq d(s,s'')+d(s'',s')$.}\label{ltsreq2}
\end{align}
We say that a LTS-algorithm $\ALG$ is \emph{lazy} if $\ALG$ does not move to another state if its current state is a haven for the current task.\end{definition}
The reason for the name \emph{lazy} task system is that given a LTS-algorithm $\ALG$, there exists a lazy LTS-algorithm $\ALG'$ such that $\ALG'(\sigma)\leq \ALG(\sigma)$ for every input $\sigma$. This follows since the function $d$ satisfies the triangle inequality \cref{ltsreq2}.
We will need this fact when proving that a finite LTS is compact.

The following lemma states that a finite LTS always has a strictly positive min-cost (and a bounded max-cost). The lemma follows immediately from the finiteness of the LTS combined with inequality (\ref{ltsreq1}).
\begin{lemma}
Let $(\mathcal{S}, \mathcal{T})$ be a LTS where $\gs$ and $\gt$ are both finite sets. There exists $\delta,\Delta\in\mathbb{R}$ such that $0<\delta\leq \Delta$ and such that $(\gs, \gt)$ has min-cost $\delta$ and max-cost $\Delta$.
\label{finiteltslem}
\end{lemma}
We now turn to proving that finite lazy task systems are compact. As previously mentioned, we will use the same proof technique as Ambühl~\cite{ambuhl} and Mömke~\cite{momke}. We begin by introducing the necessary terminology.

\begin{definition}
Let $x\in\mathbb{N}$ and let $A$ be a set of deterministic algorithms $\{\ALG_1,\ALG_2,\ldots\}$ for a LTS $(\gs, \gt)$. Let $\sigma=(s_0,t_1,\ldots, t_n)$ be a $(\gs, \gt)$-input. An algorithm is \emph{caught} after serving $t_i$ if it has incurred a cost of at least $x$ for serving $t_1,\ldots , t_{i}$. After serving a request, an algorithm which has not yet been caught can choose to \emph{surrender}. An algorithm which is not caught nor has surrendered is said to be \emph{free}.
A request is \emph{bad} (with respect to $A$ and $x$) if at least one free algorithm of $A$ incurs a non-zero cost for serving the request.
\end{definition}

We remark that eventually we will only consider algorithms that never surrenders. However, allowing an algorithm to surrender makes it easier to state a sufficiently strong induction hypothesis in the proof of the following lemma.

\begin{lemma}
\label{lcaught}
For every $x\in\posr$ and $k\in\mathbb{N}$, there exists some $N_{k,x}\in\mathbb{N}$ such that for every finite LTS $(\mathcal{S}, \mathcal{T})$ with min-cost $\delta$ and $\ab{\gs}=k$, and every set $A=\{\ALG_1, \ALG_2,\ldots\}$ of lazy deterministic algorithms (each of which is initially placed in some state), no input sequence can contain more than $\delta^{-k}N_{k,x}$ bad requests with respect to $A$ and $x$.
\end{lemma}
\begin{proof}

Since the LTS $(\gs, \gt)$ is finite, its min-cost $\delta$ is strictly positive according to \cref{finiteltslem}. We start by assuming that $\delta=1$.

Fix $x\in\posr$. The proof is by induction on $k$. Setting $N_{1,x}=\lceil x\rceil$ shows that the lemma is true for $k=1$. Indeed, since there is one unique state, all free algorithms incurs a cost of at least $1$ for each bad task (since we are assuming that $\delta=1$). Thus, after $N_{1,x}=\lceil x\rceil$ bad tasks, all algorithms which have not surrendered will have incurred a cost of at least $\lceil x\rceil \geq x$.

Assume the lemma is true for some $k\in\mathbb{N}$, and let $(\gs, \gt)$ be a LTS with $\ab{\gs}=k+1$. Fix a set of algorithms $A$. Let $\sigma$ be an input sequence and partition $\sigma$ into phases $\sigma_1,\sigma_2,\ldots ,\sigma_l$ (depending on $A$) as follows: The first phase $\sigma_1$ starts with the first task of $\sigma$. A phase ends when there have been $N_{k,x}+1$ bad tasks since the beginning of the phase. Thus, each complete phase contains exactly $N_{k,x}+1$ bad tasks (the last task is always bad) and an arbitrary number of good tasks. Note that $\sigma=\sigma_1\sigma_2\ldots \sigma_l$, where the final phase $\sigma_l$ may be incomplete. Let $1\leq j<l$. Assume that $\sigma_j=(t_1,\ldots , t_m)$. We want to show that the following claim is true.

\emph{Claim: If an algorithm in the set $A$ is free at the beginning of a complete phase $\sigma_j$, then that algorithm must either surrender or incur a cost of at least $1$ while serving $\sigma_j$.}

 Note that the last task $t_m$ of the phase $\sigma_j$ must be bad since $\sigma_j$ is a complete phase. Let $\tau=( t_1,\ldots , t_{m-1})$ be all tasks of $\sigma_j$ except the final bad task $t_m$. In order to prove the claim, assume that $\ALG'\in A$ is an algorithm which does not incur any cost for serving $\tau$ and which is still free after serving $t_{m-1}$ (if no such algorithm exists, the claim is trivially true). We want to show that $\ALG'$ must necessarily incur a cost of $1$ for serving $t_m$. Note that $\ALG'$ must have been in the same state, $s'$, while serving all of the tasks in $\tau$ (because of inequality \cref{ltsreq1}) and that $s'$ must have been a haven (\cref{GTSdefinition}) for all tasks in $\tau$. In particular, $d(s',s')=0$. In order to show that $\ALG'$ incurs a cost of at least $1$ for serving $t_{m}$, we will show that all free algorithms in $A$ must be in state $s'$ after serving $t_{m-1}$. If there exists a task $t\in\tau$ such that $t(s)=\infty$ for all $s\neq s'$, then this is obviously true. Suppose therefore that for every task $t\in\tau$, at least one state in $\gs\setminus\{s'\}$ is available.

Let $M\subseteq A$ be the set of free algorithms which were not in state $s'$ at the beginning of the phase $\sigma_j$. Let $M_{s'}\subseteq M$ be those algorithms in $M$ which at some point during the phase served a task in $\tau$ by moving to $s'$, and therefore (since all algorithms in $A$ are lazy) stayed in $s'$ while serving $\tau$.  We will now make use of the induction hypothesis. Suppose that we modify all algorithms in $M$ such that if $\ALG\in M$ moves to $s'$ when serving some task, the modified algorithm, $\widetilde{\ALG}$, instead moves to some arbitrary available state in $\mathcal{S}\setminus\{s'\}$ and then surrenders. Let $\widetilde{M}$ be the set of modified algorithms. The algorithms in $\widetilde{M}$ are well-defined and valid algorithms for the LTS $(\gs\setminus\{s'\}, \widetilde{\gt})$, where $\widetilde{\gt}$ consists of the same tasks as $\gt$ but with the state $s'$ removed. Consider now the task sequence $\widetilde{\tau}$ for $(\gs\setminus\{s'\}, \widetilde{\gt})$ induced by $\tau$. If $t_i\in \tau$ is bad with respect to $A$, this is because $t_i$ is bad for some free algorithm $\ALG\in M$ which has not yet moved into state $s'$ while serving $\tau$. Thus, the corresponding task $\widetilde{t_i}$ in $\widetilde{\tau}$ must be bad for the corresponding (free) algorithm $\widetilde{\ALG}$ in $\widetilde{M}$. It follows that $\widetilde{\tau}$ must contain at least $N_{k,x}$ tasks which are bad with respect to $\widetilde{M}$. By the induction hypothesis, after the $N_{k,x}$ bad tasks of $\widetilde{\tau}$ have been served, all of the modified algorithms in $\widetilde{M}$ have either surrendered or have incurred a cost of at least $x$ and therefore been caught. This means that for the $(\gs, \gt)$-algorithms in $M$, after the first $N_{k,x}$ bad tasks of $\tau$, all of the algorithms in $M$ have either moved into state $s'$, surrendered, or been caught. Recall that by definition, all algorithms in $A\setminus M$ have been in state $s'$ since the beginning of the current phase. We conclude that when $t_m$ arrives, all free algorithms must be in state $s'$. Therefore, in order for the final task $t_m$ of the phase to be bad, that task must force all algorithms in state $s'$ to incur a cost of at least $1$. Since all other algorithms have either surrendered or been caught, this means that all algorithms which are free after serving $\sigma_j$ must have incurred a cost of at least $1$ during the phase $\sigma_j$. This proves the claim.

It follows that every input sequence contains at most $\lceil x\rceil$ complete phases (since after this number of complete phases, all algorithms in $A$ have incurred a cost of at least $x$ or have surrendered). Hence, an input sequence can contain at most $N_{k+1,x}:=\lceil x\rceil\cdot (N_{k,x}+1)$ bad tasks with respect to $A$ and $x$. Since this recursively defined bound on $N_{k,x}$ depends only on $k$ and $x$ (and not the set of algorithms $A$ nor the LTS $(\gs, \gt)$), this finishes the proof for $\delta=1$. The general case follows by multiplying by $\delta^{-1}$ in each of the $k$ steps of the induction (and hence gives a factor of $\delta^{-k}$).
\end{proof}

We are now ready to prove an important \emph{Reset Lemma} for lazy task systems: Suppose that for every $n'$, we have a good algorithm for inputs of length at most $n'$. Then, we can split the input into epochs of bounded length. When a new epoch begins, the good online algorithm wipes its own memory and serves the epoch as if the input had only consisted of this epoch. This results in a good algorithm for inputs of arbitrary lengths. While this sound similar to the concept of a repeatable online problem, note that for repeatable online problems, it is the \emph{adversary} that can simulate a reset of the game. On the other hand, for LTS, it is the algorithm that can decide to essentially reset the game between itself and the adversary. The only difficulty in proving the reset lemma lies in the fact that an input sequence could contain arbitrarily long subsequences for which $\OPT$ pays nothing. However, using \cref{lcaught}, we can show that there is a limit as to how long the adversary can benefit from such subsequences.

\begin{lemma}[Reset Lemma]
\label{resetlemma}
Let $(\mathcal{S}, \mathcal{T})$ be a LTS where both $\mathcal{S}$ and $\mathcal{T}$ are finite. Let $c\geq 1$ be a constant and let $\alpha:\mathbb{N}\rightarrow\posr$ be a function of $n$ such that $\alpha=o(n)$. Assume that for every $n'\in\mathbb{N}$, there exists a randomized algorithm $\ra_{n'}$ (depending on $n'$) such that $\E[\ra_{n'}(\sigma)]\leq c\OPT(\sigma)+\alpha(n')$ for every input sequence $\sigma$ where $\ab{\sigma}\leq n'$. Then for every constant $\varepsilon>0$, there exists a single randomized algorithm $\ra$ such that $\ra$ is $(c+\varepsilon)$-competitive (for inputs of arbitrary length).
\end{lemma}
In this section (where our goal is to prove \cref{mmtt}), we will only use the Reset Lemma when $\alpha$ is just an additive constant, that is, $\alpha=O(1)$. However, we prove a more general version of the Reset Lemma where the additive term $\alpha$ is allowed to be any function of order $o(n)$. This will be needed in \cref{sec:constructmain} where we use online learning to give a constructive version of \cref{mmtt}. Note that the lemma is obviously false if $\alpha$ was allowed to be of order $O(n)$ instead of $o(n)$, since no algorithm would ever incur a cost larger than $\alpha(n)=\Delta\cdot n=O(n)$ on inputs of length at most $n$, where $\Delta$ is the max-cost of the LTS.
\begin{proof}[Proof (of \cref{resetlemma})]
Let $\varepsilon>0$. The goal is to design an algorithm $\ra$ for inputs $\sigma$ of arbitrary length. To this end, we first fix some constants (i.e., numbers which may depends on $c$, $\alpha$, $\varepsilon$, and $(\gs, \gt)$, but not on the input $\sigma$).
Let $k=\ab{\gs}$, and let $\delta$ be the min-cost and $\Delta$ the max-cost of $(\gs, \gt)$. By \cref{finiteltslem}, we get that $0<\delta\leq \Delta<\infty$. Furthermore, for reasons that will become clear later on, let $x=\varepsilon^{-1}k\Delta$ and choose $n'\in\mathbb{N}$ large enough so that \begin{equation}\label{nle}n'\geq \lceil\varepsilon^{-1}\delta^{-1}(\alpha(n')+(c+\varepsilon)\Delta)\rceil\cdot (\delta^{-k}N_{k,x}+k).\end{equation} This is possible since, by assumption, $\alpha(n)\in o(n)$ while $\varepsilon, c, \delta, k,$ and $N_{k,x}$ are all constants with respect to the input length $n$. Thus, $\lceil\varepsilon^{-1}\delta^{-1}(\alpha(n)+(c+\varepsilon)\Delta)\rceil\cdot (\delta^{-k}N_{k,x}+k)\in o(n)$, from which the existence of an $n'$ satisfying (\ref{nle}) follows. Please note that $n'$ is simply some fixed integer, depending on the constants $\varepsilon, c, \delta,\Delta, k,N_{k,x}$ and the function $\alpha$. Thus, $n'$ and $\alpha(n')$ are both independent of the input $\sigma$ and its length.

By assumption, there exists a randomized algorithm $\ra_{n'}$ such that $\E[\ra_{n'}(\sigma)]\leq c\cdot \OPT(\sigma)+\alpha(n')$ for every input $\sigma$ of length $\ab{\sigma}\leq n'$. Recall that formally, a randomized algorithm is a probability distribution over a (possibly infinite number of) deterministic algorithms. However, since the number of possible inputs and outputs of length at most $n'$ is finite (we assumed that $\gs$ and $\gt$ were both finite), it follows that there only exist a finite number of possible deterministic algorithms for inputs of length at most $n'$. In particular, there must be a finite number of deterministic algorithms $\ALG_1,\ldots , \ALG_m$ such that for inputs of length at most $n'$, the algorithm $\ra_{n'}$ can be viewed as a probability distribution over these $m$ deterministic algorithms. For $1\leq i\leq m$, let $p(i)$ be the probability assigned to the deterministic algorithm $\ALG_i$ by $\ra_{n'}$. We will assume that the algorithms $\ALG_1,\ldots , \ALG_m$ are lazy. If not, we can simply make them lazy and thereby obtain a new randomized algorithm with exactly the same performance guarantee as $\ra$.

Consider an input sequence $\sigma$ of arbitrary length $n$. We partition $\sigma$ into phases as follows: The first phase begins with the first request. A phase ends when for every state $s\in\gs$, it holds that $s$ was not a haven for at least one task in the current phase. At any time step, we let $M$ be the set of states which have so far been havens against every task in the current phase. Note that both the phases and the set $M$ can be computed online. Also, note that immediately before a phase starts, every state belongs to $M$. During a phase, the number of states in $M$ will decrease. When the last task of the phase arrives, $M$ will become  empty. Note that any algorithm (including $\OPT$) must incur a cost of at least the min-cost $\delta$ while serving the tasks of a complete phase.

We will now modify the algorithm $\ra_{n'}$. The modified algorithm, $\ra'$, will still only be used on inputs of length at most $n'$ but will serve as an important subroutine in the final algorithm for inputs of arbitrary length. In order to define $\ra'$, we first obtain a set of deterministic algorithms $\{\ALG_1',\ldots , \ALG_m'\}$ by modifying the algorithms $\{\ALG_1,\ldots , \ALG_m\}$.
For each $i$, the modified algorithm, $\ALG_i'$, is defined as follows: At the beginning of each phase, $\ALG_i'$ behaves exactly as $\ALG_i$ (recall that an online algorithm can figure out itself when an old phase ends and a new phase begins). However, assume that a task $t$ arrives such that after serving $t$, the algorithm $\ALG_i$ will have incurred a total cost of at least $x$ in the current phase. In this case, $\ALG_i'$ serves $t$ in the same way that $\ALG_i$ does and then goes \emph{\OPT-chasing}. An algorithm which is \OPT-chasing will move to some (arbitrarily chosen) state $s$ currently in $M$. The algorithm will remain in state $s$ until a task against which $s$ is not a haven arrives. When this happens, the \OPT-chasing algorithm $\ALG_i'$ moves to another state currently in $M$. This continues until $M$ is empty, that is, until the phase ends. When the last task $t$ of the phase arrives, the algorithm $\ALG_i'$ will stop its \OPT-chasing and will serve $t$ by moving to the same state as $\ALG_i$ does when serving $t$.

Let $\widehat{\sigma}$ be a phase of $\sigma$. We claim that $\ALG_i'(\widehat{\sigma})\leq (1+\varepsilon)\ALG_i(\widehat{\sigma})$. Let $\widehat{\sigma}=\widehat{\sigma}_1\widehat{\sigma}_2$ where $\widehat{\sigma}_1$ are the tasks served by $\ALG_i'$ before it went \OPT-chasing and $\widehat{\sigma}_2$ are the tasks served by $\ALG_i'$ while \OPT-chasing. By definition, $\ALG_i'(\widehat{\sigma}_1)=\ALG_i(\widehat{\sigma}_1)$. Also, $\ALG_i'(\widehat{\sigma}_2)\leq k\Delta$. Since $\ALG_i'$ only goes \OPT-chasing if $\ALG_i(\widehat{\sigma}_1)\geq x=\varepsilon^{-1}k\Delta$, we get that $\ALG_i'(\widehat{\sigma})\leq \ALG_i(\widehat{\sigma}_1)+k\Delta\leq (1+\varepsilon)\ALG_i(\widehat{\sigma}_1)\leq (1+\varepsilon)\ALG_i(\widehat{\sigma})$.

We are now ready to define $\ra'$. The algorithm $\ra'$ randomly selects an algorithm, $\ALG '$, from $\{\ALG_1',\ldots , \ALG_m'\}$ such that $\Pr[\ALG'=\ALG_i']=p(i)$. Note that $\ALG '$ is a random variable. Also, $\ra'$ simulates (deterministically) all of the algorithms $\ALG_1',\ldots , \ALG_m'$ while serving the input. If a task $t$ arrives for which at least one of the $m$ algorithms will incur a non-zero cost (i.e., $t$ is bad), then $\ra'$ passes on the task $t$ to $\ALG '$ and serves $t$ by moving to the same state as $\ALG '$ does. $\ra'$ also updates its deterministic simulation of the $m$ algorithms. On the other hand, if a task $t$ arrives for which none of the $m$ algorithms would incur a non-zero cost (i.e., every algorithm is in a state which is a haven against $t$), then the task $t$ is \emph{ignored}. That is, $\ra'$ does not pass on the task $t$ to $\ALG '$ and it does not update its deterministic simulation of the $m$ algorithms. Instead, $\ra'$ simply serves $t$ without incurring any cost (recall that $\ra'$ must currently be in a state which is a haven against $t$) and without the help of $\ALG '$. Thus, all further requests will be handled exactly the same as if $t$ had never arrived. 

We will now bound the number of requests that can be passed on to $\ALG '$ in a single phase. By \cref{lcaught}, after $\ra'$ has passed on $\delta^{-k}N_{k,x}$ requests in a single phase to $\ALG '$, all of the $m$ algorithms $\ALG_1',\ldots , \ALG_m'$ must have incurred a cost of at least $x$. Thus, they must have gone $\OPT$-chasing. In particular, they are all currently occupying some state that belongs to $M$. Thus, from now on, every request not ignored by $\ra'$ will decrease the size of $M$ by one. It follows that after passing on at most $k$ further requests, $M$ is empty and the phase has ended. We conclude that at most $\delta^{-k}N_{k,x}+k$ requests in a single phase can be passed on to $\ALG '$.

Finally, we are ready to define the algorithm $\ra$. To this end, we divide the input sequence into \emph{epochs}. Each epoch will consist of a number of phases. The first epoch starts when the first task of $\sigma$ arrives. An epoch ends when it contains $\lceil\varepsilon^{-1}\delta^{-1}(\alpha(n')+(c+\varepsilon)\Delta)\rceil$ complete phases. Since the beginning and end of each phase can be computed online, so can the beginning and end of each epoch.  

Suppose a new epoch $\tau$ has just begun. The algorithm $\ra$ will be in some state $s$ (the state in which $\ra$ served the last task of the previous epoch, or, if $\tau$ is the very first epoch, $s$ will be the initial state of the input sequence $\sigma$). The algorithm $\ra$ will use a \emph{new instantiation} of $\ra'$ (which will initially be in state $s$) to serve $\tau$. By a new instantiation, we mean that an epoch should be served using an instantiation of $\ra'$ which is independent (and unaware) of all previous epochs of $\sigma$ and how they were served. Formally, $\ra$ will simulate $\ra'$ on the $(\gs, \gt)$-input with initial state $s$ and tasks $\tau$, and $\ra$ will serve all tasks of $\tau$ exactly as they are served by $\ra'$. When $\tau$ ends, $\ra$ throws away the current instantiation of $\ra'$ and serves the next epoch using a new instantiation of $\ra'$.

We denote by $\OPT(\tau)$ the cost incurred by $\OPT$ while serving an epoch $\tau$ when given $\sigma$ as input. Note that the expected cost incurred by $\ra$ while serving an epoch $\tau$, which we denote by $\E[\ra(\tau)]$, is exactly $\E[\ra'(\tau)]$. We will give an upper bound on $\E[\ra'(\tau)]$ in terms of $\OPT(\tau)$. Since at most $N_{k,x}+k$ tasks of each phase are passed on to $\ALG '$ by $\ra'$, this means that if we use $\ra'$ to serve the epoch $\tau$, then at most $\lceil\varepsilon^{-1}\delta^{-1}(\alpha(n')+(c+\varepsilon)\Delta)\rceil\cdot (\delta^{-k}N_{k,x}+k)$ tasks are being passed on to $\ALG '$ by the definition of an epoch. By the choice of $n'$, this means that at most $n'$ tasks are being passed on to $\ALG '$ in a single epoch. Note that by the triangle inequality, the optimal cost of serving an epoch must be at least the optimal cost of serving those requests in the epoch that were passed on to $\ALG'$. Thus, by the performance guarantee of $\ra_{n'}$ for inputs of length at most $n'$, we get that $\E[\ra'(\tau)]\leq (c+\varepsilon) \OPT_s(\tau)+\alpha(n')$ where $\OPT_s$ is an optimal offline algorithm which initially is in the same state $s$ as $\ra'$ was when $\tau$ began. With the exception of the very first epoch, it could be the case that $\OPT$ (the optimal offline algorithm for the entire sequence $\sigma$) begins the epoch $\tau$ in a state different from $s$. However, we have that $\OPT_s(\tau)\leq \OPT(\tau)+\Delta$. Thus, 
\begin{align}
\label{rtaubound}
\E[\ra(\tau)]=\E[\ra'(\tau)]\leq (c+\varepsilon)\OPT_s(\tau)+\alpha(n') \leq (c+\varepsilon)\OPT(\tau)+\alpha (n')+(c+\varepsilon)\Delta. 
\end{align}

Let $\tau_1,\ldots, \tau_l$ be all of the epochs (so that $\sigma=\tau_1\ldots \tau_l$). For $1\leq i <l$, $\tau_i$ is a complete epoch. Since $\OPT$ incurs a cost of at least $\delta$ for each complete phase, it follows that 
\begin{align}
\notag\OPT(\tau_i)&\geq \delta \lceil\varepsilon^{-1}\delta^{-1}(\alpha(n')+(c+\varepsilon)\Delta) \rceil \\
&\geq \varepsilon^{-1}(\alpha(n')+(c+\varepsilon)\Delta).\label{OPTtaubound}
\end{align}
Combining (\ref{rtaubound}) and (\ref{OPTtaubound}), we get that $\ra(\tau_i)\leq (c+2\varepsilon)\OPT(\tau_i)$ for $1\leq i<l$. Note that the last epoch, $\tau_l$, may be incomplete and so we cannot use this estimate for $\ra(\tau_l)$. However, from the bound on the complete epochs and the competitiveness of $\ra'$, it follows that

\begin{align*}
\E[\ra(\sigma)]&=\sum_{i=1}^{l} \E[\ra(\tau_i)] = \left(\sum_{i=1}^{l-1}\E[\ra(\tau_i)]\right) + \E[\ra(\tau_l)]\\
&\leq \left(\sum_{i=1}^{l-1}(c+2\varepsilon)\OPT(\tau_i)\right) + (c+\varepsilon)\OPT(\tau_l)+\alpha(n')+(c+\varepsilon)\Delta\\
&\leq \left(\sum_{i=1}^{l}(c+2\varepsilon)\OPT(\tau_i)\right) + \alpha(n')+(c+\varepsilon)\Delta\\
&= (c+2\varepsilon)\OPT(\sigma) + \alpha(n')+(c+\varepsilon)\Delta.
\end{align*}
Since $\alpha(n')+(c+\varepsilon)\Delta=O(1)$ does not depend on the input $\sigma$, this shows that $\ra$ is $(c+2\varepsilon)$-competitive for input sequences of arbitrary length. Since $\varepsilon$ was arbitrary this proves the theorem.

\end{proof}

We say that an online problem is \emph{truly finite} if there are only a finite number of possible inputs and outputs. Note that we do not just require that for some fixed input length there should be a finite number of inputs (and outputs), or that there are only finitely many types of requests. It is the total number of inputs and outputs which must be finite. In particular, a necessary condition for an online problem to be finite is that there is a fixed upper bound on the number of requests. Also, all valid algorithms for a finite online problem are (non-strictly) $1$-competitive (simply choose the additive constant to be sufficiently large). The following lemma follows directly from (the difficult direction of) Yao's principle. For completeness, we include the proof.

\begin{lemma}
\label{yaofinite}
Let $\P$ be a truly finite minimization problem and let $c>1$, $\varepsilon>0$ and $\alpha\geq 0$. Assume that, for every probability distribution $p$ over inputs to $\P$, there exists a deterministic algorithm $\ALG$ such that $\E_{\sigma\sim p}[\ALG(\sigma)] < (c-\varepsilon)\E_{\sigma\sim p}[\OPT(\sigma)]+\alpha$. Then, there exists a randomized algorithm $\ra$ such that $\E[\ra(\sigma)] < (c-\varepsilon)\OPT(\sigma)+\alpha$ for every input sequence $\sigma$.
\end{lemma}
\begin{proof}

Let $\sigma_1,\sigma_2,\ldots , \sigma_m$ be an enumeration of all $\P$-inputs and let $\ALG_1,\ALG_2,\ldots , \ALG_{m'}$ be all of the valid deterministic $\P$-algorithms (there are only a finite number of such inputs and algorithms since $\P$ is truly finite). It is well-known that competitive analysis of an online problem may be viewed as a two-player zero-sum game between an online player and an adversary. Consider the function $f(i,j)=(c-\varepsilon)\OPT(\sigma_j)+\alpha-\ALG_i(\sigma_j)$. The payoff function for the online player is $f(i,j)$ while the payoff function for the adversary is $-f(i,j)$. Thus, the online player is trying to select $i$ so as to maximize $f(i,j)$ while the adversary is trying to select $j$ so as to minimize $f(i,j)$.

Since we are interested in randomized algorithms, we allow the players to use mixed strategies (a probability distribution over deterministic strategies). We denote by $\mathbf{y}=(y_1,\ldots , y_{m'})$ a mixed strategy of the online player (where $y_i$ is the probability assigned to $\ALG_i$) and by $\mathbf{x}=(x_1,\ldots , x_m)$ a mixed strategy of the adversary. Let $F(\mathbf{y},\mathbf{x})=\E_{i\sim \mathbf{y}}[\E_{j\sim \mathbf{x}}[f(i,j)]]$ be the expected payoff (for the online player). By a slight abuse of notation, we will write $F(\ALG_i, \mathbf{x})$ as a shorthand for $F(\mathbf{y},\mathbf{x})$ where $\mathbf{y}$ is the probability distribution which assigns a probability of $1$ to $\ALG_i$. Similarly for $F(\mathbf{y},\sigma_j)$.

Since $\P$ is a finite two-person zero-sum game, it follows from von Neumann's minimax theorem that
\begin{equation}
\max_{\mathbf{y}}\min_{\mathbf{x}} F(\mathbf{y},\mathbf{x})=\min_{\mathbf{x}}\max_{\mathbf{y}} F(\mathbf{y},\mathbf{x}).
\end{equation}
As a corollary (sometimes known as Loomis' lemma), it follows that there exists an optimal mixed strategy, $\mathbf{y}^{\star}$, for the online player and an optimal mixed strategy, $\mathbf{x}^{\star}$, for the adversary such that 
\begin{equation}
\label{loomislemma}
\max_i F(\ALG_i,\mathbf{x}^{\star})=\min_j F(\mathbf{y}^{\star}, \sigma_j).
\end{equation}
We claim that $\max_i F(\ALG_i, \mathbf{x}^{\star})>0$. If not, then $\max_i F(\ALG_i, \mathbf{x}^{\star})\leq 0$ which implies that $\E_{j\sim \mathbf{x}^{\star}}[\ALG_i(\sigma_j)]\geq (c-\varepsilon)\E_{j\sim \mathbf{x}^{\star}}[\OPT(\sigma_j)]+\alpha$ for every deterministic algorithm $\ALG_i$, but this contradicts the assumption of the theorem. This shows that the claim is true and thus, by \cref{loomislemma}, we have that $\min_j F(\mathbf{y}^{\star}, \sigma_j)=\max_i F(\ALG_i, \mathbf{x}^{\star})>0$. Equivalently, $\E_{i\sim \mathbf{y}^{\star}}[\ALG_i(\sigma)] < (c-\varepsilon)\OPT(\sigma)+\alpha$ for every input $\sigma$. Thus, the mixed strategy $\mathbf{y}^{\star}$ proves the existence of a randomized algorithm with the desired performance guarantee.

\end{proof}

\begin{theorem}
Let $(\mathcal{S}, \mathcal{T})$ be a LTS such that both $\mathcal{S}$ and $\mathcal{T}$ are finite. Then $(\mathcal{S}, \mathcal{T})$ is compact.
\label{LTZcompact}
\end{theorem}
\begin{proof}
Assume by way of contradiction that $(\gs, \gt)$ is a finite task system which is not compact. Let $c>1$ be a constant such that the expected competitive ratio of every randomized $(\gs, \gt)$-algorithm is at least $c$ (if no such $c$ exists, then we are done). Since $(\gs, \gt)$ is not compact, there exists an $\varepsilon>0$ and an $\alpha\geq 0$ such that the following holds: For every $n'\in\mathbb{N}$ and every probability distribution $p$ over inputs of length at most $n'$, there exists a deterministic algorithm $\ALG$ such that
\begin{equation}
\E_{p}[\ALG(\sigma)]< (c-\varepsilon)\E_{p}[\OPT(\sigma)]+\alpha.
\label{finitestcomp}
\end{equation}
Recall that $\gs$ and $\gt$ are both finite by assumption. This means that for every $n'\in\mathbb{N}$, we obtain a truly finite online problem by considering the set of $(\gs,\gt)$-inputs of length at most $n'$. Thus, it follows from \cref{yaofinite} and (\ref{finitestcomp}) that for every $n'\in\mathbb{N}$, there exists a randomized algorithm $\ra_{n'}$ such that $\E[\ra_{n'}(\sigma)]<(c-\varepsilon)\OPT(\sigma)+\alpha$ for every $(\gs,\gt)$-input $\sigma$ of length at most ${n'}$. But then \cref{resetlemma} implies that there exists a single randomized $(\gs, \gt)$-algorithm which is $(c-\varepsilon /2)$-competitive for inputs of arbitrary length. This gives the desired contradiction, since the competitive ratio of every randomized $(\gs, \gt)$-algorithm was assumed to be at least $c$.
\end{proof}

\section{Applications of main results}
\label{sec:appmaint1}
We are finally ready to apply \cref{maint1,mmtt}. In particular, we will prove the lower bounds for $\Sigma$-repeatable problems stated in \cref{maintable} and, when relevant, discuss the equivalence of randomization and sublinear advice implied by \cref{mmtt}.

If a problem is $\Sigma$-repeatable but not known to be compact, we have to verify that the lower bound on randomized algorithms without advice is compatible with \cref{maint1}. This is often straightforward, even if the lower bound is not explicitly proved using Yao's principle. Indeed, most lower bound proofs starts by fixing (an upper bound on) the number of requests $n$. For this fixed $n$, a \emph{finite} number of request sequences of length at most $n$ are constructed in an adversarial manner. It is then shown that the lower bound follows by letting $n$ tend to infinity. By Yao's principle (see \cref{yaofinite}), such a lower bound can always be converted into a family of input distributions compatible with \cref{maint1}.

\subsection{Paging}
The paging problem with a cache of size $k$ (and a set of $N$ pages) can be modeled as a finite LTS and is therefore compact and $\Sigma$-repeatable. Thus, the well-known~(see \cite{Paging91, BLSmetric}) lower bound of $H_k=\Omega(\log k)$ on the competitive ratio of randomized online algorithms without advice carries over to online algorithms with sublinear advice. Previously, the best known lower bound for algorithms with sublinear advice was $5/4$ \cite{A1}.

It is shown in \cite{A1} that with $\lceil \log k \rceil$ bits of advice, a deterministic algorithm can achieve a competitive ratio of $3\log k+O(1)=O(\log k)$, and with $n$ bits of advice, a deterministic algorithm can be strictly $1$-competitive. Combining these results with our lower bound, we see that the asymptotic advice complexity of paging is now fully understood: $O(1)$ bits of advice are enough for a deterministic algorithm to be $O(\log k)$-competitive (a randomized algorithm can be $O(\log k)$-competitive without advice). However, $o(n)$ bits of advice are not enough to be better than $\Omega(\log k)$-competitive, while $n$ bits of advice suffice for a deterministic algorithm to be strictly $1$-competitive. 

\cref{pagingprec} gives a lower bound on the exact number of advice bits needed to achieve some competitive ratio $1\leq c< H_k$ for paging. 
\subsection{\texorpdfstring{$k$}{k}-server}
\label{sec:appmaint1:k}
For every $k$ and every finite metric space $(M,d)$, the $k$-server problem on $(M,d)$ can be modeled as a finite LTS and is therefore compact and $\Sigma$-repeatable. Since paging corresponds to the $k$-server problem on the uniform metric, the lower bound of $H_k=\Omega(\log k)$ for paging algorithms with advice complexity $o(n)$ also applies to the $k$-server problem. Previously, the best known lower bound for the competitive ratio of $k$-server algorithms with sublinear advice was $3/2$~\cite{Smula}. This lower bound was shown to hold even for two servers on a line, and improved an earlier lower bound of $5/4$ obtained in \cite{A1} and \cite{Asushmita}.

For the $2$-server problem on the line, a lower bound of $e/(e-1)\approx 1.58$ for randomized algorithms without advice is well-known~\cite{KMMO90}. For the general $2$-server problem, an even better lower bound of $1+e^{-1/2}\approx 1.60$ for randomized algorithms was obtained in \cite{2serverlower}. Both of these lower bounds are obtained using finite metric spaces. Thus, combining these results with \cref{maint1} improves on the previous best lower bound of $3/2$ for the $2$-server problem (on a line and in general).

In \cite{Bartalram}, Bartal et al.\ show that the competitive ratio of a randomized $k$-server algorithm on \emph{any} metric space (with more than $k$ points) is at least $\Omega(\log k / \log^2 \log k)$. By \cref{maint1}, the same lower bound holds for algorithms with advice complexity $o(n)$. No advice complexity lower bounds for the $k$-server problem on arbitrary metric spaces were previously known.

For arbitrary metric spaces, the best known $k$-server algorithm is the deterministic work-function algorithm which is $(2k-1)$-competitive~\cite{WFA}. In a recent breakthrough, Bansal et al.\ gave a randomized $O(\log^2 k\log^3 m \log\log m)$-competitive $k$-server algorithm \cite{polylogkserver}, where $m$ is the number of points in the underlying metric space. In particular, this algorithm has a better competitive ratio than the work-function algorithm if $m$ is subexponential in $k$. It has been conjectured that it is in fact possible to achieve a competitive ratio of $O(\log k)$ for any metric space (see e.g.\ \cite{DBLP:journals/csr/Koutsoupias09}). It is known that it is possible to achieve a competitive ratio of $O(\log k)$ using $2n$ bits of advice \cite{Ak-server} (see also \cite{Ak-serverRR}). By \cref{mmtt}, achieving the same competitive ratio with $o(n)$ bits of advice would confirm the randomized $k$-server conjecture (for finite metric spaces).

\subsection{Metrical task systems}\label{sec:appmaint1:mts} Any finite MTS with $N$ states can, obviously, be modeled as a finite LTS and is therefore compact and $\Sigma$-repeatable. As with the $k$-server problem, the best known lower bound for randomized algorithms is the $\Omega(\log N)$ lower bound on the uniform metric space. Together with \cref{maint1}, this implies that a MTS-algorithm with advice complexity $o(n)$ must have a competitive ratio of at least $\Omega(\log N)$. However, this advice complexity lower bound was already obtained (for another metric space) in \cite{A2} by a reduction from the generalized matching pennies problem. Thus, we do not obtain an improved advice complexity lower bound for the MTS problem on a worst-case metric space. We remark that the metric space and the type of tasks used to obtain the $\Omega(\log N)$ lower bound for algorithms with $o(n)$ bits of advice in \cite{A2} is of such a nature that the lower bound does not imply a similar lower bound for e.g.\ $k$-server or paging.

As with the $k$-server problem, combining \cref{maint1} with the results of Bartal et al.\ \cite{Bartalram} yields a lower bound of $\Omega(\log N/ \log^2\log N)$ for algorithms with sublinear advice on arbitrary metric spaces. No advice complexity lower bounds for the MTS problem on arbitrary metric spaces were previously known.

Contrary to the $k$-server problem, almost tight lower and upper bounds on the competitive ratio of randomized algorithms without advice are known. In particular, there exists a randomized $O(\log ^2 N\log\log N)$-competitive MTS-algorithm for arbitrary metric spaces~\cite{DBLP:journals/jcss/FakcharoenpholRT04}. However, there is still a small gap of $\Theta(\log N \log\log N)$ between the best known upper and lower bound for randomized algorithms without advice. \cref{mmtt} shows that finding the best possible competitive ratio of deterministic algorithms with advice complexity $o(n)$ would close this gap for finite metric spaces. We remark that it is possible to achieve a competitive ratio of $O(\log N)$ using $n$ bits of advice~\cite{A2}.
\subsection{List update}
The list update problem on a finite list can be modeled as a finite LTS and is therefore compact and $\Sigma$-repeatable. In \cite{AListupdate}, a lower bound of $15/14$ for list update algorithms with sublinear advice was shown by a reduction from binary string guessing. We improve this lower bound to $3/2$ by combining the fact that list update is $\Sigma$-repeatable with a lower bound of $3/2$ for randomized algorithms without advice due to Teia~\cite{Teia93}. 

The best known randomized list update algorithm is the $1.6$-competitive COMB algorithm~\cite{COMBlu}. By \cref{mmtt}, achieving a competitive ratio of $1.5\leq c<1.6$ with $o(n)$ bits of advice would improve on this upper bound.
\subsection{Bipartite matching}
\label{app:bipmath}
We consider the (unweighted) bipartite matching problem~\cite{optbip}. The initial state of the problem is a number $l\in\mathbb{N}$. A sequence of requests is defined in terms of a bipartite graph $G=(L,R,E)$ (where $L$ and $R$ are the parts of the bipartition of the vertex set and $E$ is the set of edges) satisfying $\ab{L}=l$. The vertices of $L$ are called the \emph{offline} vertices. A request is a vertex $v\in R$ together with all edges between $L$ and $r$. As soon as a vertex arrives, it must irrevocably be matched to an incident (and unmatched) vertex in $L$, if possible. The goal is to match as many vertices as possible. Note that the number of requests, $n$, equals the size of $R$. It is usually assumed that $n=\ab{R}=\Theta(l)$. 

Note that the size $l$ of $L$ is part of the initial state of the problem. In particular, the additive constant in the definition of the competitive ratio cannot depend on $l$ since $l$ is part of the input. This is necessary to avoid trivializing the problem, since $\OPT(\sigma)\leq l$ for all inputs $\sigma$.


Let $\P$ be the bipartite matching problem. We will show that $\P$ is $\Sigma$-repeatable. Given an input $\sigma^*=(\sigma_1;\ldots ; \sigma_r)$ for $\Prs$ where all rounds have $l$ offline vertices, we define the input $g(\sigma^*)$ for $\P$ as follows: The input $g(\sigma^*)$ will have $r\cdot l$ offline vertices\footnote{Since the number of offline vertices is part of the input, the reduction $g$ from $\Prs$ to $\P$ is allowed to change it.}. We partition these vertices into $r$ groups such that the $i$th group contains $l$ vertices. The requests of $\sigma_1$ will be revealed using the first group of offline vertices, then the requests of $\sigma_2$ will be revealed using the second group of offline vertices, and so on. It is easy to see that this reduction shows that bipartite matching is $\Sigma$-repeatable with parameters $(0,0,0)$.

Previously, it was known that $\lceil \log (n!)\rceil =\Omega(n\log n)$ bits of advice were necessary to be strictly $1$-competitive~\cite{Abipfinal}. A classical result of Karp et al.\ says that there exists a randomized $e/(e-1)$-competitive algorithm, and that no randomized algorithm (without advice) can achieve a better competitive ratio~\cite{optbip}. We get from this and \cref{maint1} that no randomized algorithm with advice complexity $o(n)$ can achieve a competitive ratio better than $e/(e-1)$, and that this is tight. Shortly after our work, it was shown by D{\"{u}}rr et al. (using a streaming algorithm from \cite{DBLP:journals/algorithmica/EggertKMS1}) that for every constant $\varepsilon>0$, there exists a $(1+\varepsilon)$-competitive algorithm using $O(n)$ bits of advice~\cite{DBLP:journals/corr/DurrKR16}.

\subsection{Reordering buffer management}
It is easy to show that the reordering buffer management problem with a buffer of size $k$ is $\Sigma$-repeatable (see \cref{lem:rbmrep}). In \cite{STOCrbm}, Adamaszek el al.\ give a lower bound of $\Omega(\log \log k)$ for randomized online algorithms without advice. Since reordering buffer management is $\Sigma$-repeatable (and since the lower bound in \cite{STOCrbm} satisfies the conditions of \cref{maint1}), we get the same lower bound for algorithms with advice complexity $o(n)$. In \cite{FOCSrbm}, a randomized $O(\log \log k)$-competitive algorithm is obtained. Using derandomization, we see that there exists a deterministic $O(\log \log k)$-competitive algorithm with $O(\log n)$ bits of advice. Finally, it is shown in \cite{Arbm} that for every $\varepsilon>0$, there is a deterministic $(1+\varepsilon)$-competitive algorithm reading $O(n)$ bits of advice and that $\Omega(n\log k)$ bits of advice is necessary (and obviously sufficient) to be $1$-competitive. Thus, the asymptotic number of advice bits needed to achieve a given asymptotic competitive ratio is now fully determined for randomized algorithms.
\subsection{Dynamic binary search trees}
\label{app:dbs}
We will consider the dynamic binary search tree problem from the perspective of competitive analysis. Several models exists, most of which are equivalent up to constant factors (see e.g.\ \cite{DBLP:journals/siamcomp/Wilber89}). The model that we use is as follows: The goal is to create a binary search tree (BST) which supports searches for keys from a static set $[N]$. An input is a sequence of keys $\sigma=(x_1,\ldots , x_n)$ such that $x_i\in [N]$ for each $1\leq i\leq n$ (so that searches are always successful). A BST-algorithm maintains a pointer. The pointer can be moved to its left child, its right child, or its parent, all at unit cost. Also, the BST-algorithm can perform a rotation on the current node and its parent at unit cost. When $x_i$ is revealed, the BST-algorithm must make a sequence of unit-cost operations such that the node containing the key $x_i$ is at some point visited. We assume that the pointer starts in the node where it ended after the search for the previous key $x_{i-1}$ in the input sequence. Finally, we assume that initializing the pointer (at the last node visited) is a unit-cost operation which must always be performed at the beginning of a search. This assumption ensures that the cost of finding a key is at least the number of nodes visited.

Demaine et al.\ has shown that there exist $O(\log\log N)$-competitive binary search trees~\cite{DBLP:journals/siamcomp/DemaineHIP07}. The well-known dynamic optimality conjecture claims that Sleater and Tarjan's splay trees \cite{DBLP:journals/jacm/SleatorT85} are in fact $O(1)$-competitive. We show in \cref{app:BSTlazy} in \cref{app:modellazy} that the dynamic binary search tree problem (as it is defined above) can be modeled as a lazy task system. This implies that if there exists a BST with advice complexity $o(n)$ where $n$ is the number of searches (and not the number of keys), then there also exists a randomized $O(1)$-competitive BST without advice. We remark that it has previously been shown to be possible to (effectively) convert a BST with advice complexity $O(1)$ into a deterministic BST with essentially the same performance guarantee~\cite{DBLP:journals/siamcomp/DemaineHIP07}.

\subsection{Sleep states management}
The $2$-sleep states management problem ($2$-SSM, see \cite{ssm,Assm}) can, obviously, be modeled as a generalized task system, and so it is $\Sigma$-repeatable. As observed in \cite{ssm}, the $2$-SSM problem is essentially a repeated ski-rental problem. Thus, it follows from \cite{KMMO90} that no randomized algorithm without advice can be better than $e/(e-1)$-competitive. Furthermore, this lower bound satisfies the conditions of \cref{maint1}. Thus, we get a lower bound of $e/(e-1)$ for algorithms with advice complexity $o(n)$ for $2$-SSM. The previous best lower bound for algorithms with sublinear advice was $7/6$ and was obtained in \cite{Assm}. We note that \cite{KMMO90} also gives a randomized algorithm with a competitive ratio of $e/(e-1)$. Furthermore, in \cite{Assm}, it is shown that it is possible for a deterministic algorithm to achieve a competitive ratio arbitrarily close to $e/(e-1)$ using $O(1)$ bits of advice. Since $n$ bits of advice trivially suffice to be optimal for $2$-SSM, the asymptotic advice complexity of $2$-SSM is now fully determined.

\subsection{One-dimensional unit clustering}
We consider the one-dimensional online unit clustering problem (see \cite{DBLP:journals/mst/ChanZ09,UnitClustering}). We claim that this problem is $\Sigma$-repeatable with parameters $(0,0,0)$. Let $\sigma^*=(\sigma_1;\ldots ; \sigma_r)$ be an input for the repeated version of the problem. For $1\leq i\leq r$, let $x^i_{\min}$ be the left-most point of $\sigma_i$, and let $x^i_{\max}$ be the right-most point of $\sigma_i$. Thus, for every point $x\in \sigma_i$, it holds that $x\in[x^i_{\min}; x^i_{\max}]$. Let $L_i=\ab{x^{i}_{\max}-x^{i}_{\min}}$. Furthermore, let $a_1=0, b_1=L_1$ and define recursively $a_i=b_{i-1}+2$ and $b_i=a_i+L_{i}$ for $2\leq i\leq r$.  We construct an input $\sigma$ for the online unit clustering problem as follows: The points of $\sigma_i$ are shifted so that they all fall within the interval $[a_i, b_i]$. This is possible since $b_i-a_i=L_i$. Furthermore, for $i\neq i'$, a cluster cannot cover points from both $[a_i, b_i]$ and $[a_{i'}; b_{i'}]$. It follows easily that this reduction satisfies the conditions of \cref{re}.

It is shown in \cite{UnitClustering} that every randomized algorithm is at least $3/2$-competitive. This is done by showing that for every $\varepsilon>0$, there exists a hard input distribution with finite support such that no deterministic algorithm can be better than $(3/2-\varepsilon)$-competitive against this distribution. Note that the number of possible requests for the unit clustering problem is uncountable. However, the lower bound of $3/2$ is proved using input distributions with finite support. Thus, combining the lower bound of \cite{UnitClustering} with \cref{maint1}, it follows that an algorithm with $o(n)$ bits of advice must be at least $3/2$-competitive. We remark that the best known online algorithm for one-dimensional unit clustering is a $5/3$-competitive deterministic algorithm~\cite{DBLP:journals/tcs/EhmsenL13}.

\subsection{Max-SAT}
The online Max-SAT problem (see \cite{Max-SAT}) is $\Sigma$-repeatable with parameters $(0,0,0)$. Indeed, given an input $\sigma^*=(\sigma_1;\ldots ; \sigma_r)$ for the repeated version of the problem, we can construct an input $\sigma$ of Max-SAT by using a fresh set of variables for each round $\sigma_i$ (that is, for $i\neq i'$, the set of variables used in round $i$ is disjoint from the set of variables used in round $i'$). It is easy to see that this reduction satisfies the conditions of \cref{re}. 

It is shown in \cite{Max-SAT} that no randomized algorithm can achieve a competitive ratio better than $3/2$. This is done using a family of hard input distributions which are compatible with \cref{maint1}. Thus, we get that an algorithm with advice complexity $o(n)$ must have a competitive ratio of at least $3/2$. We remark that a $3/2$-competitive randomized algorithm (without advice) is given in \cite{Max-SAT}.

\subsection{Metric matching}
\label{app:metmatch}
In the metric matching problem (see \cite{DBLP:journals/jal/KalyanasundaramP93,DBLP:journals/tcs/KhullerMV94}), we are given as initial state a metric space $(M,d)$ and $k$ servers $S\subseteq M$ with a pre-specified location\footnote{Note that $S$ is a multisubset of $M$ since more than one server could be placed at the same point in $M$.}. A request $r_i$ is a point $r_i\in M$ in the metric space. When a request $r_i$ arrives, an online algorithm must irrevocably match the request to a server $s\in S$ which has not previously been matched to another request. The algorithm incurs a cost of $d(r_i,s)$ for matching the request $r_i$ to the server $s$. Note in particular that the number of requests is (at most) $k$.

We will show that the metric matching problem is $\Sigma$-repeatable with parameters $(0,0,0)$. To this end, we will use that the underlying metric space and the pre-specified servers are parts of the input (and not fixed as parameters of the problem). As with bipartite matching, the problem would trivialize if the additive constant was allowed to depend on $k$ and the metric space (since the cost of $\OPT$ is never more that $k\Delta$ where $\Delta$ is the diameter). Note that this is very different from the $k$-server problem.

 Given an input $\sigma^*=(\sigma_1;\ldots ; \sigma_r)$ for the repeated version of the problem, we can construct an input $\sigma$ of metric matching as follows: Denote by $(M',d')$ and $S'\subseteq M'$ the metric space and set of servers used as initial state for $\sigma^*$ (recall that all of $\sigma_1,\ldots ,\sigma_r$ by definition are required to have the same initial state). We define a new metric space $(M,d)$ which contains $r$ subspaces, $(M_1, d_1),\ldots , (M_r,d_r)$, each of which is a copy of $(M',d')$, in such a way that the distance between two points in two different subspaces is at least $B$ for some sufficiently large $B\in\mathbb{R}$. The requests of $\sigma$ are simply the requests of $\sigma_1,\ldots , \sigma_r$ such that the requests of $\sigma_i$ are points in $(M_i,d_i)$. By choosing $B$ to be sufficiently large, we can ensure that if an algorithm ever matches a request in $\sigma_i$ to a server in a subspace other than $(M_i,d_i)$, then that algorithm incurs a cost larger than any algorithm which only matches requests to servers in their corresponding subspace. With such a choice of $B$, it is easy to see that the conditions of \cref{re} are satisfied.

For general metric spaces, there exists an $O(\log^2 k)$-competitive randomized algorithm~\cite{DBLP:journals/algorithmica/BansalBGN14}. Also, it is known that any randomized algorithm must be $\Omega(\log k)$-competitive~\cite{DBLP:conf/soda/MeyersonNP06}. Note that $\Omega(\log k)$ is not constant with respect to the input length. However, together with \cref{maint1}, it does imply that no algorithm with sublinear advice can be $O(1)$-competitive.

From the discussion at the end of \cref{sec:defcomp}, it is easy to see that the metric matching problem is compact (on finite metric spaces). Indeed, in this problem, there are $k$ servers placed in a metric space. Each server can be matched to at most one request (and vice versa). But this means that for every fixed metric space and fixed set of $k$ servers, the length of the input never exceeds $k$. This allow (reusing the notation from the discussion) us to construct a single algorithm $\ra$ which on inputs with $k$ servers run the appropriate algorithm $\ra_k$. 

An interesting special case is the line metric. An $O(\log k)$-competitive randomized algorithm is known~\cite{DBLP:conf/icalp/GuptaL12}, but even for deterministic algorithms, the best lower bound is only 9.001~\cite{DBLP:journals/tcs/FuchsHK05}. Therefore, it has been conjectured that there exists an $O(1)$-competitive randomized algorithm (see e.g.\ \cite{DBLP:conf/dagstuhl/KalyanasundaramP96}). If instead of the real line we consider the problem on finite metric discretizations of the line, then the metric matching problem becomes compact. Thus, it follows from \cref{mmtt} that this conjecture is equivalent to the conjecture that there exists a deterministic $O(1)$-competitive algorithm with sublinear advice for the line metric.

\refstepcounter{subsection}
\subsection*{\thesubsection \quad Examples of non-repeatable online problems}

Some examples of classical online problems which are provably not $\Sigma$-repeatable includes bin packing (see however \cref{subsec:betterbp}) and makespan minimization (on identical machines). For both of these problems, it is has been shown that sublinear advice allows an algorithm to achieve a competitive ratio which is provably better than the competitive ratio of any randomized online algorithm without advice~\cite{Abp,Amakespan}. Since the lower bounds for algorithms without advice satisfy the conditions of \cref{maint1}, this proves that the problems are not $\Sigma$-repeatable.

\section{Constructive version of \texorpdfstring{\cref{mmtt}}{Theorem 2} via online learning}
\label{sec:constructmain}
As mentioned, the proof of \cref{mmtt} is non-constructive. It does not show how to obtain the randomized algorithm given an algorithm with sublinear advice. In this section, we show how to easily obtain a constructive version of an important special case of \cref{mmtt} by combining a variant of the Weighted Majority Algorithm from online learning with the Reset Lemma (\cref{resetlemma}). Recall that \cref{mmtt} holds for all compact and $\Sigma$-repeatable problems. However, the most interesting class of problems known to be compact are those that can be formulated as a (finite) LTS. We will show how to obtain an algorithmic version of \cref{mmtt} for lazy task systems. 
\begin{theorem}
\label{conmainthm}
Let $\P$ be a minimization problem which can be modeled as a finite lazy task system, let $c\geq 1$ be a constant and let $\ALG$ be a deterministic $c$-competitive $\P$-algorithm with advice complexity $o(n)$. Then, for every $\varepsilon>0$, there exists a randomized $(c+\varepsilon)$-competitive algorithm without advice.
\end{theorem}
\begin{proof}
By assumption, $\P$ can be modeled as a finite LTS $(\gs, \gt)$. Since $(\gs, \gt)$ is finite, it has a finite max-cost $\Delta\in\mathbb{R}$. Let $\varepsilon_1>0$ be arbitrary and let $b$ be the advice complexity of the $c$-competitive $\P$-algorithm $\ALG$. We want to apply the Reset Lemma. To this end, let $n'\in\mathbb{N}$. Note that for inputs of length at most $n'$, the algorithm $\ALG$ can be converted into $m=2^{b(n')}$ deterministic algorithms, $\ALG_1,\ldots , \ALG_{m}$, without advice. Following Blum and Burch~\cite{BBmts} (see also \cref{sec:related}), we use the algorithm Hedge~\cite{WMAfs} to combine the $m$ algorithms into a single randomized algorithm, $\ra_{n'}$, such that for every input $\sigma$ of length at most $n'$, it holds that (cf. \cref{bbeq}):
\begin{align*}
\E[\ra_{n'}(\sigma)]&=(1+2\varepsilon_1)\cdot\min_{1\leq i\leq m}\ALG_i(\sigma)+\left(\frac{7}{6}+\frac{1}{\varepsilon_1} \right)\Delta\ln m\\
&=(1+2\varepsilon_1)\cdot \ALG(\sigma) + \left(\frac{7}{6}+\frac{1}{\varepsilon_1} \right)\Delta\; b(n')\ln (2).\\
&\leq (1+2\varepsilon_1)\cdot (c\cdot \OPT(\sigma)+\alpha_0)+\left(\frac{7}{6}+\frac{1}{\varepsilon_1} \right)\Delta\; b(n')\ln (2)\\
&=(1+2\varepsilon_1)\cdot (c\cdot\OPT(\sigma))+(1+2\varepsilon_1)\alpha_0+\left(\frac{7}{6}+\frac{1}{\varepsilon_1} \right)\Delta\; b(n')\ln (2).
\end{align*} 
Define $\alpha:\mathbb{N}\rightarrow\posr$ as follows:
\begin{equation} 
\alpha(n)=(1+2\varepsilon_1)\alpha_0+\left(\frac{7}{6}+\frac{1}{\varepsilon_1} \right)\Delta\; b(n)\ln (2).
\end{equation}
By assumption, $b(n)\in o(n)$. Furthermore, $\varepsilon_1, \alpha_0$, and $\Delta$ are constants (with respect to $n$). Thus, $\alpha(n)\in o(n)$. 

We now have that for each $n'\in\mathbb{N}$, there exists a randomized algorithm $\ra_{n'}$ such that $\E[\ra_{n'}(\sigma)]\leq (1+2\varepsilon_1)c\cdot \OPT(\sigma)+\alpha(n')$ for every input $\sigma$ of length $\ab{\sigma}\leq n'$. Thus, since $\alpha(n)\in o(n)$, the Reset Lemma (\cref{resetlemma}) shows that for every $\varepsilon_2>0$, there exists a randomized algorithm $\ra$ without advice such that $\ra$ is $((1+2\varepsilon_1)c+\varepsilon_2)$-competitive on inputs of arbitrary length. Recall that $c$ is a constant. Choosing $\varepsilon_1=\varepsilon/(4c)$ and $\varepsilon_2=\varepsilon/2$ yields a randomized $(c+\varepsilon)$-competitive algorithm without advice.

\end{proof}

\section{Repeated matrix games}
One of the most successful techniques for proving advice complexity lower bounds has been via reductions from the string guessing problem \cite{A2, sg}. In this section, we show that string guessing is a special case of a more general online problem where the algorithm and the adversary repeatedly play a zero-sum game. Applications to bin packing and paging are given.

\label{sec:repmat}
We first fix the notation. Let $q\geq 2$ and let $[q]=\{1,2,\ldots , q\}$. Let $A\in\posr^{q\times q}$ be a quadratic $q\times q$ matrix with non-negative real entries. Such a matrix defines a finite two-player zero-sum game: Player I, the ``row-player'', has $q$ pure strategies, one for each row of $A$. Player II, the ``column-player'', has $q$ pure strategies, one for each column of $A$. Simultaneously, Player I chooses a row $x\in [q]$ and Player II chooses a column $y\in [q]$. The payoff function for player I is $A(x,y)$ and the payoff function for player II is $-A(x,y)$. If player I uses the mixed strategy\footnote{A mixed strategy is a probability distribution over rows (or columns).} $\mu:[q]\rightarrow[0,1]$ and player II uses the mixed strategy $\nu:[q]\rightarrow [0,1]$ then the expected payoff of player I is defined as $$A(\mu,\nu)=\mathop{\E}_{\substack{x\sim \mathbf{\mu}\\ y\sim\mathbf{\nu}}}\left[ A(x,y)\right]=\sum_{x=1}^q \sum_{y=1}^q A(x, y)\mu(x)\nu(y).$$ Since the game defined by $A$ is a finite two-player zero-sum game, it follows from von Neumann's minimax theorem that it has a \emph{value} $V\geq 0$ such that
\begin{equation}
V=\max_{\mu}\min_{\nu} A(\mu,\nu) = \min_{\nu}\max_{\mu} A(\mu,\nu).
\label{nmpre}
\end{equation}
Here, the $\max$ is over all mixed strategies $\mu$ for player I, and the $\min$ is over all mixed strategies $\nu$ for player II. The value $V$ (along with optimal strategies) is efficiently computable using linear programming. As a corollary to (\ref{nmpre}), sometimes known as Loomis' lemma, it follows that there exists an optimal mixed strategy, $\mu^{\star}$, for player I and an optimal mixed strategy, $\nu^{\star}$, for player II such that $V=\max_{x\in[q]} A(x,\nu^{\star})=\min_{y\in[q]} F(\mu^{\star}, y)$.

The game defined above is sometimes called a \emph{one-shot} game. Given such a one-shot game, one may define a \emph{repeated game} where the two players repeatedly play the one-shot game. We will now define an online version of this repeated game.

\begin{definition}
Let $q\in\mathbb{N}$. A \emph{repeated matrix game (RMG)} is an online problem defined by a cost matrix $A\in\posr^{q\times q}$. The input $\sigma=( n, x_1, \ldots , x_n )$ consists of an integer $n$ and a sequence of $n$ characters from $[q]$. Round $1\leq i\leq n$ of the game proceeds according to the following rules:
\begin{enumerate}
\item If $i=1$, the algorithm learns $n$. If $i>1$, the algorithm learns the input character $x_{i-1}$.
\item The algorithm answers $y_i=f_i(n,x_1,\ldots , x_{i-1})\in [q]$, where $f_i$ is a function defined by the algorithm.
\end{enumerate}
The cost of the output $\gamma=y_1\ldots y_n$ computed by the algorithm is $\sum_{i=1}^nA(x_i, y_i)$.
\end{definition}
Note that in a repeated matrix game, the adversary is the row-player and the algorithm is the column-player. We will not consider the competitive ratio of a RMG-algorithm but instead the total cost incurred by the algorithm. Also, following previous work~\cite{A2, sg, AOC} on string guessing problems, we slightly abuse notation by using $n$ to denote the number of input characters and not the number of requests, which is $n+1$ (in the standard model of online problems, it is not possible to avoid this extra round since the $i$th input character must be revealed strictly later than when the algorithm has to output its $i$th character). Also, note that RMG does not have any initial state ($n$ is not the initial state but a request that must be answered by the algorithm).

For a RMG with cost matrix $A$, there exists a deterministic algorithm which incurs a cost of at most $\min_{y\in [q]}\max_{x\in [q]} A(x,y)$ in each round by playing according to an optimal pure strategy for the column player. Clearly, no deterministic algorithm (without advice) can do better than this. On the other hand, a randomized algorithm can ensure an expected cost of at most $V$ in each round by playing according to an optimal mixed strategy for the column player. A small amount of advice suffices to derandomize this strategy.
\begin{theorem}
\label{RMGalg}
Let $V$ be the value of the (one-shot) two-player zero-sum game defined by the matrix $A\in\posr^{q\times q}$. Then, there exists a randomized algorithm $\ra$ without advice for the RMG with cost matrix $A$ such that $\E[\ra(\sigma)]\leq Vn$ for every input $\sigma$. Furthermore, there exists a deterministic algorithm, $\ALG$, reading $\lceil\log q\rceil$ bits of advice such that $\ALG(\sigma)\leq Vn$ for every input $\sigma$.
\end{theorem}
\begin{proof}
Consider the one-shot game defined by $A$. From the minimax theorem, we know that there exists an optimal mixed strategy $\nu^*:[q]\rightarrow [0,1]$ for the column-player such that for each $x\in[q]$, it holds that $A(x,\nu^*)=\E_{y\sim \nu^*}[A(x,y)]=\sum_{y=1}^qA(x,y)\nu^*(y)\leq V$. We will now define the algorithm $\ra$ for the RMG with cost matrix $A$. Before the first request arrives, the algorithm $\ra$ selects a character, $Y\in[q]$, at random according to $\nu^*$ (note that $Y$ is a random variable). In each round, $\ra$ answers $Y$. Let $1\leq i\leq n$ and let $x_i\in[q]$ be the input character in round $i$. We know that the expected cost $A(x_i, \nu^*)$ incurred by $\ra$ in round $i$ is at most $V$. Thus, by linearity of expectation, $E[\ra(\sigma)]\leq Vn$ for every input $\sigma$.

We can convert $\ra$ into a deterministic algorithm $\ALG$ with advice. The oracle looks at the input $\sigma$ and determines the best choice, $y$, of the random variable $Y$ for this particular input (that is, the oracle finds the character $y$ minimizing $\sum_{i=1}^nA(x_i,y)$). It writes $y$ onto the advice tape using $\lceil \log q\rceil $ bits of advice. The algorithm $\ALG$ learns $y$ from the advice and answers $y$ in each round. It follows that $\ALG(\sigma)\leq \E[\ra(\sigma)]\leq Vn$ for every input $\sigma$.
\end{proof}
\subsection{Proof of \texorpdfstring{\cref{m:RMGmainthm}}{Theorem 3}}
We will show that linear advice is required to ensure a cost of at most $(V-\varepsilon)n$ if $\varepsilon>0$ is constant. This result follows from \cref{techlemma} and the fact that all RMGs are $\Sigma$-repeatable. However, it is much more natural to prove this result directly using our direct product theorem (\cref{mainthm}), since the result is in fact an easier version of \cref{techlemma} where each round consists of exactly one request. Also, there is a slight technical problem with using the fact that a RMG is $\Sigma$-repeatable if one wants a non-asymptotic lower bound: A RMG-input must consist of at least two requests (since the number of requests is $n+1$, where $n$ is the number of input characters). By our definitions of $\Sigma$-repeatable online problems and our statement of \cref{techlemma}, this will result in a lower bound which is unnecessarily a factor of two lower than it should be. Therefore, we give a direct proof of \cref{RMGmainthm} instead of using \cref{techlemma}. For convenience, we restate \cref{RMGmainthm} before giving the proof.

\restateRMGmainthm*
\begin{proof}[Proof of \cref{RMGmainthm}]
Since the value of the game defined by $A$ is $V$, we know that there exists a probability distribution $\mu^*$ over characters from $[q]$ such that for every $y\in [q]$, it holds that $A(\mu^*, y)\geq V$. 

Let $n\in\mathbb{N}$. Define $p^n\colon\{(n, x_1, \ldots , x_n)\vert x_i\in [q]\}\rightarrow [0,1]$ to be the RMG-input distribution which maps $(n, x_1, \ldots , x_n)$ into $\mu^*(x_1)\mu^*(x_2)\cdots \mu^*(x_n)$. Thus, in each of the $n$ rounds, we draw independently an input character from $[q]$ according to $\mu^*$. The $i$th round cost function simply assigns a cost of $A(x_i, y_i)$ to the algorithm where $x_i$ is the input character and $y_i$ is the answer produced by the algorithm in round $i$. Together with these cost-functions, the input distribution $p^n$ is an $n$-round input distribution. Note that the $i$th round input distribution $p_i$ is simply $p_i=\mu^*$.

Let $\ALG$ be a deterministic RMG-algorithm reading at most $b$ bits of advice on inputs of length at most $n$.
Fix $1\leq i\leq n$ and $w\in \mathcal{W}_i$. Let $d=D_{\KL}(\piw \| p_i)$. In order to apply \cref{maint1}, we need a lower bound on $\E [\cost_i(\ALG)\vert W_i=w]$ in terms of $d$. By Pinsker's inequality (\ref{pinsker}),
\begin{equation}
\|\piw - p_i\|_{1}\leq\sqrt{d\cdot \ln 4}.
\label{pintechlemRMG}
\end{equation}
Let $h(d)=\sqrt{d\cdot\ln 4}$. For every $w\in\mathcal{W}_i$ where $\Pr[W_i=w]>0$, the algorithm $\ALG$ outputs a fixed character $y_w\in[q]$ in round $i$ and incurs an expected cost of $\E_{x\sim \piw}[A(x,y_w)]$. Obviously, the cost incurred by $\ALG$ in any round is never larger than $\|A\|_\infty:=\max_{x,y}A(x,y)$. Thus,
\begin{align*}
\E[\cost_i(\ALG)\vert W_i=w]&=\E_{x\sim \piw}[A(x,y_w)]\\
&\geq \E_{x\sim p_i}[A(x, y_w)]-\|A\|_\infty \cdot h(d) \tag*{\{By (\ref{pintechlemRMG}) and (\ref{etv})\}}\\
&= \E_{x\sim \mu^*}[A(x, y_w)]-\|A\|_\infty \cdot h(d)\\
&\geq V-\|A\|_\infty\cdot h(d).
\end{align*}Define $f(d)=V-\|A\|_\infty\cdot h(d)$. Since $f$ is convex and decreasing, \cref{mainthm} yields 
\begin{equation}
\label{techlemmaconc}
\E_{\sigma\sim p^n}[\ALG(\sigma)]=\E[\cost(\ALG)]\geq nf(b/n)=n\left(V-\|A\|_\infty\sqrt{\frac{b \cdot \ln 4}{n}}\right).
\end{equation}
It follows from Yao's principle that if a randomized algorithm on every input of length $n$ reads at most $b$ bits of advice and incurs an expected cost of at most $(V-\varepsilon)n$, then
\begin{align*}
(V-\varepsilon)n\geq n\left(V-\|A\|_\infty\sqrt{\frac{b \cdot \ln 4}{n}}\right).
\end{align*}
A straightforward calculation completes the proof of the theorem.

\end{proof}
Since $\varepsilon\leq V\leq \|A\|_\infty$, we have $\frac{\varepsilon^2}{2\ln (2)\|A\|^2_\infty}\leq \frac{1}{2\ln(2)}<0.722$. Thus, the lower bound (\ref{RMGmainthmeq}) is never larger than $0.722n$ and will often be much lower. In what follows, we will obtain better bounds than (\ref{RMGmainthmeq}) for some specific cost matrices.
\subsection{Guessing strings and matching pennies}
For alphabets of size $q=2$, string guessing with known history~\cite{A2, sg} is the repeated matrix game with cost matrix $$A_2=\begin{pmatrix} 0 & 1 \\ 1 & 0 \end{pmatrix}.$$ More generally, the $q$-SGKH problem~\cite{A2, sg} is the repeated matrix game with cost matrix $A_q=J_q-I_q$ where $J_q$ is the all-ones matrix and $I_q$ the identity matrix, both of size $q\times q$. That is,
$$A_q= \begin{pmatrix} 0 &  1 & 1  & \ldots & 1\\
1  &  0 & 1 & \ldots & 1\\
1 & 1 & 0 &  \ldots & 1\\
\vdots & \vdots & \vdots &\ddots & \vdots \\
1  &   1 & 1       &\ldots & 0\end{pmatrix}.$$

Emek et al.\ introduced the generalized matching pennies problem~\cite{A2} which is equivalent to the string guessing problem, except that a dummy cost of $1/\tau$ is added to every entry of $A_q$ for some integer $\tau$. Using the high-entropy technique, they obtained lower bounds valid for alphabets of size $q\geq 4$ and certain ranges of the advice read and the cost incurred. B{\"o}ckenhauer et al.\ later obtained tight (up to an additive lower order term of order $O(\log n)$) lower and upper bounds on the deterministic advice complexity of the string guessing problem \cite{sg}. Since their introduction, a large number of advice complexity lower bounds have been obtained by reductions from the string guessing or generalized matching pennies problem. This includes lower bounds for bin packing~\cite{Abp,DBLP:conf/wads/AngelopoulosDKR15}, list update~\cite{AListupdate}, metrical task systems~\cite{A2}, reordering buffer management~\cite{Arbm}, set cover~\cite{Asetcover}, and several other online problems.

It is easy to see that the value $V_q$ of the game defined by the matrix $A_q$ is $V_q=(q-1)/q$, and that the optimal mixed strategy for both players is to select a character uniformly at random. Thus, it follows immediately from \cref{RMGmainthm} that $\Omega(n)$ bits of advice are needed for an algorithm to ensure a cost of at most $( \frac{q-1}{q}-\varepsilon)n$. Note how \cref{RMGmainthm} makes it a trivial task to obtain this result. For several applications, the lower bound of $\Omega(n)$ is all that one needs. However, if one seeks a better lower bound for larger alphabets, or if the exact coefficient of the higher-order term is important, a better lower bound is needed. In this section, we will show how to reprove the tight lower bounds due to B{\"o}ckenhauer et al.~\cite{sg} using our direct product theorem. We believe that our technique yields a very intuitive proof of the lower bound. In particular, it explains nicely why the $q$-ary entropy function appears in the tight lower bound. Also, we slightly strengthen the lower bound of \cite{sg} by showing that it also applies to randomized algorithms with advice (this also follows from our derandomization results in \cref{sec:derand}).

\begin{theorem}
\label{thmreprove}
Let $q\geq 2$ and let $0<\alpha <(q-1)/q$. A randomized $q$-SGKH algorithm which on every input of length $n$ has an expected cost of at most $\alpha n$ must read at least $$b\geq K_{1/q}(1-\alpha)n=(1-h_q(\alpha))n\log_2 q$$ bits of advice.
\end{theorem}
\begin{proof}
Let $n\in\mathbb{N}$. For $1\leq i\leq n$, the input character $x_i$ is chosen uniformly at random from $[q]$. This gives rise to an $n$-round input distribution where each round consists of exactly one request. The $i$th round cost-function $\cost_i$ simply assigns a cost of $0$ to an algorithm in round $i$ if the algorithm guesses correctly and a cost of $1$ otherwise.

Fix a deterministic algorithm $\ALG$ which reads at most $b$ bits of advice on inputs of length $n$. Let $1\leq i\leq n$ and let $w\in\mathcal{W}_i$ be any possible history in round $i$. Let $d=D_{KL}(\piw \|p_i)$.
Then, by \cref{Klem}, for every $x\in [q]$ it holds that $\piw(x)\leq K_{1/q,r}^{-1}(d)$. Thus,
$$\E[\cost_i(\ALG)|W_i=w] \geq 1-K_{1/q,r}^{-1}(d),$$ since no matter which character that $\ALG$ chooses in round $i$, the probability of this character being wrong (and $\ALG$ therefore incurring a cost of $1$) is at least $1-K_{1/q,r}^{-1}(d)$. Let $f(d)=1-K_{1/q,r}^{-1}(d)$. Note that $f$ is convex and decreasing. By \cref{mainthm}, we therefore get that $\E[\cost(\ALG)]\geq nf(n^{-1}b)$. We will now calculate how large $b$ needs to be in order to ensure that the expected cost is at most $\alpha n$.
\begin{align*}
&\alpha n\geq nf(b/n)\\
&\Rightarrow \alpha\geq 1-K_{1/q, r}^{-1}(b/n)\\
&\Rightarrow K^{-1}_{1/q, r}(b/n)\geq 1-\alpha\\
&\Rightarrow b/n\geq K_{1/q}(1-\alpha)\tag*{\{$K_{1/q}$ is increasing on $[1/q,1]$\}}\\
&\Rightarrow b\geq n K_{1/q}(1-\alpha)=n(1-h_q(\alpha))\log(q).
\end{align*}
Thus, any deterministic algorithm needs to read at least $b\geq K_{1/q}(1-\alpha)n$ bits of advice in order to ensure an expected cost of at most $\alpha n$ against this $n$-round input distribution. By Yao's principle, we get that any randomized algorithm must read at least $K_{1/q}(1-\alpha)n$ bits of advice in order to guarantee that it always incurs an expected cost of at most $\alpha n$.
\end{proof}
As mentioned, it is shown in \cite{sg} that \cref{thmreprove} is tight up to (additive) lower order terms. Note that \cref{thmreprove} shows that we need $\Omega(n\log q)$ bits to ensure a cost of at most $(\frac{q-1}{q}-\varepsilon)n$, whereas \cref{RMGmainthm} only gives a lower bound of $\Omega(n)$.

\subsection{Anti-string guessing}
The \emph{anti-string guessing problem (Anti-$q$-SG)} over an alphabet of size $q$ is the repeated matrix game with cost-matrix $I_q$ (the identity matrix of size $q\times q$). Thus, the goal of the algorithm is to output a character \emph{different} from the input character. The value of the one-shot game defined by $I_q$ is $1/q$. For both players, the optimal strategy is to select a character uniformly at random. \cref{RMGmainthm} therefore implies that an online algorithm needs $\Omega(n)$ bits of advice in order to achieve a cost of at most $(1/q-\varepsilon)n$ on every input of length $n$. As was the case for the original string guessing problem, \cref{RMGmainthm} makes it trivial to obtain this result. Note that for $q=2$, string guessing and anti-string guessing are equivalent. However, one interesting aspect of anti-string guessing is that it becomes easier as $q$ grows. 

In this section, we will show that a direct application of our information theoretical direct product theorem yields lower bounds for anti-string guessing which are not just asymptotically tight, but are tight up to an additive $O(\log n)$ term.

\begin{theorem}
\label{thm:antilower}
Let $q\geq 2$ and let $0<\alpha <1/q$. A randomized Anti-$q$-SG algorithm which on every input of length $n$ incurs an expected cost of at most $\alpha n$ must read at least $$b\geq K_{1/q}(\alpha)n=(1-h_q(1-\alpha))n\log_2 q$$ bits of advice.
\end{theorem}
\begin{proof}
Let $n\in\mathbb{N}$. For $1\leq i\leq n$, the input character $x_i$ is chosen uniformly at random from $[q]$. This gives rise to an $n$-round input distribution where each round consists of exactly one request. The $i$th round cost-function $\cost_i$ simply assigns a cost of $0$ to an algorithm in round $i$ if the algorithm outputs something else than the input character and a cost of $1$ if the algorithm outputs the input character.

Fix a deterministic algorithm $\ALG$ which reads at most $b$ bits of advice on inputs of length $n$. Let $1\leq i\leq n$ and let $w\in\mathcal{W}_i$ be any possible history in round $i$. Let $d=D_{KL}(\piw\|p_i)$. Then, by \cref{Klem}, the smallest probability assigned to a single character by $\piw$ is at least $K_{1/q,l}^{-1}(d)$.
Thus, $\E[\cost_i(\ALG)|W_i=w]\geq K_{1/q,l}^{-1}(d)$, since no matter which character that $\ALG$ chooses in round $i$, the probability of this character being equal to the input character is at least $K_{1/q,l}^{-1}(d)$. Let $f(d)=K_{1/q,l}^{-1}(d)$. Note that $f$ is convex and decreasing. By \cref{mainthm}, we have that $\E[\cost(\ALG)]\geq nf(n^{-1}b)$. We will now calculate how large $b$ needs to be in order to ensure that the expected cost is at most $\alpha n$:
\begin{align*}
&\alpha n\geq nf(b/n)\\
&\Rightarrow\alpha\geq K_{1/q,l}^{-1}(b/n)\\
&\Rightarrow K_{1/q}(\alpha)\leq b/n\tag*{\{$K_{1/q}$ is decreasing on $[0, 1/q]$\}}\\
&\Rightarrow b\geq n K_{1/q}(\alpha)=n(1-h_q(1-\alpha))\log(q).
\end{align*}
Thus, any deterministic algorithm needs to read at least $b\geq K_{1/q}(\alpha)n$ bits of advice in order to ensure an expected cost of at most $\alpha n$ against this $n$-round input distribution. By Yao's principle, we get that any randomized algorithm must read at least $K_{1/q}(\alpha)n$ bits of advice in order to guarantee that its expected cost is never more than $\alpha n$.

\end{proof}
\paragraph{Upper bound for anti-string guessing.} It follows from \cref{RMGalg} that it is possible to ensure an expected cost of at most $n/q$ for Anti-q-SG by simply answering randomly in each round (or by using $\lceil \log q\rceil$ bits of advice to communicate the character appearing least frequently in the input). We will now show how a simple probabilistic argument yields an upper bound on the advice needed to ensure a cost less than $n/q$. This upper bound shows that \cref{thm:antilower} is essentially tight.

\begin{theorem}
\label{thm:upperanti}
Let $q\geq 2$ and let $0<\alpha<1/q$. There exists a deterministic algorithm for Anti-$q$-SG which on every input of length $n$ incurs a cost of at most $\alpha n$ and which reads
\begin{align*}
b=K_{1/q}(\alpha)n+O(\log n+\log\log q)=(1-h_q(1-\alpha))n\log_2 q+O(\log n+\log\log q)
\end{align*}
bits of advice.
\end{theorem}
\begin{proof}
See \cref{app:antilowerup}.
\end{proof}

\subsubsection{Application to paging}
We have already shown using \cref{maint1} that $\Omega(n)$ bits of advice are needed to be better than $H_k$-competitive for paging. However, as usual, lower bounds obtained from \cref{maint1} are only interesting from an asymptotic point of view. In this section, we show how a reduction from anti-string guessing to paging gives the same asymptotic lower bound but with (much) better constants than those hidden in the $\Omega(n)$ lower bound from \cref{maint1}. We remark that in the reduction from Anti-$q$-SG to paging, we need the fact that the lower bound obtained in \cref{thm:antilower} is proved using the uniform input distribution.

\begin{theorem}
\label{pagingprec}
Let $k\in\mathbb{N}$ and let $1<c< H_k$. For every $\varepsilon>0$, there exists an $n_\varepsilon$ such that a strictly $c$-competitive paging algorithm with a cache of size $k$ must read at least $$b\geq K_{\mbox{$\frac{1}{k+1}$}}\left(\frac{c}{(k+1)H_k-\varepsilon}+\frac{c}{n}\right)n$$ bits of advice on inputs of length $n\geq n_\varepsilon$.
\end{theorem}
\begin{proof}
We will prove the lower bound by a reduction from Anti-$(k+1)$-SG to paging with a cache of size $k$ and a universe of $k+1$ different pages, $\{s_1,\ldots , s_{k+1}\}$. To this end, let $\ALG$ be an arbitrary $c$-competitive deterministic paging algorithm with advice complexity $b(n)$ and an additive constant of $\alpha$. Define $\ALG'$ to be the following algorithm for Anti-$(k+1)$-SG: For an input string $x=x_1\ldots x_n\in [k+1]^n$, the algorithm $\ALG'$ simulates the paging algorithm $\ALG$ on the input $\sigma_x=(s_{x_1}, s_{x_2},\ldots , s_{x_n})$ with an initial cache of, say, $\{s_1,\ldots , s_k\}$. When $\ALG'$ has to output an answer $y_i\in [k+1]$ in round $i$, it computes the page, $s_j$, which is outside of the cache of $\ALG$ when the $i$th request of $\sigma_x$ arrives. $\ALG'$ then answers $y_i=j$. Clearly, $\ALG'$ incurs a cost of $1$ in round $i$ if and only if $\ALG$ makes a page-fault on the $i$th request. Furthermore, $\ALG'$ only needs to read $b(n)$ bits of advice for simulating $\ALG$. 


Now, we let $p:[k+1]^n\rightarrow [0,1]$ be the uniform distribution over all strings of length $n$ (with an alphabet of size $k+1$). We know from \cref{thm:antilower} that if $\E_{x\sim p}[\ALG'(x)]\leq \beta n$ for every input $x$ of length $n$, then $b(n)\geq K_{1/q}(\beta)n$. Furthermore, it is known (see e.g.\ \cite{BE98b}) that for every $\varepsilon>0$, there exists an $n_\varepsilon$ such that $\E_{x\sim p}[\OPT(\sigma_x)]\leq \frac{n}{(k+1)H_k-\varepsilon}+1$ whenever $n\geq n_{\varepsilon}$. For the rest of the proof, assume that $n\geq n_\varepsilon$. Using that $\ALG$ is strictly $c$-competitive, we have
\begin{align*}
\label{eq:pagred}
\E_{x\sim p}[\ALG'(x)]&= \E_{x\sim p}[\ALG(\sigma_x)]\\
&\leq c\E_{x\sim p}[\OPT(\sigma_x)]\\
&\leq \frac{cn}{(k+1)H_k-\varepsilon}+c. 
\end{align*}
By construction, the number of advice bits read by $\ALG'$ is $b(n)$. It follows from \cref{thm:antilower} that $$b(n) \geq K_{\frac{1}{k+1}}\left(\frac{c}{(k+1)H_k-\varepsilon}+\frac{c}{n}\right)n. $$
\end{proof}
\subsection{Weighted binary string guessing}
\label{appendix:binpackproof}
Let $0<s\leq t$. The \emph{weighted binary string guessing problem} with weights $(s,t)$, denoted BSG-$(s,t)$ for short, is the repeated matrix game with cost-matrix $$A=\begin{pmatrix} 0 & s \\ t & 0 \end{pmatrix}.$$
For simplicity, we will use $\{0,1\}$ as the alphabet instead of $\{1,2\}$ when working with weighted binary string guessing. It is easy to show that the value of the one-shot game defined by the matrix $A$ is $V=st/(s+t)$, and that the optimal strategy for an algorithm is to output $0$ with probability $s/(s+t)$ and $1$ with probability $t/(s+t)$. Let $\varepsilon>0$. From \cref{RMGmainthm}, we get that an online algorithm which on inputs of length $n$ is guaranteed to incur a cost of at most $\left(st/(s+t)-\varepsilon\right) n$ must read at least $\Omega(n)$ bits of advice. On the other hand, according to \cref{RMGalg}, the algorithm that answers $0$ with probability $s/(s+t)$ and $1$ with probability $t/(s+t)$ incurs an expected cost of at most $(st/(s+t)) n$ on every input of length $n$. Note that for $s=t=1$, these results are in agreement with the results for (unweighted) binary string guessing.

\begin{theorem}
\label{thm:wsgb}
Let $0<s\leq t$ and let $0<\alpha <\frac{st}{s+t}$. A randomized BSG-$(s,t)$ algorithm which on every input of length $n$ incurs an expected cost of at most $\alpha n$ must read at least $$b\geq K_{s/(s+t)}\left(\frac{\alpha}{t}\right)n$$ bits of advice.
\end{theorem}

\begin{proof}
Let $t>1$. We start by proving the lower bound for BSG-$(1,t)$. To this end, we use the following $n$-round input distribution: In each round, the correct answer is $0$ with probability $t/(t+1)$ and $1$ with probability $1/(t+1)$. Fix a deterministic algorithm $\ALG$ which reads at most $b$ bits of advice on on inputs of length $n$. Let $1\leq i\leq n$ and let $w\in\mathcal{W}_i$ be any possible history in round $i$. Recall that $p_i(0|w)=\Pr[X_i=0|W_i=w]$ and similar for $p_i(1|w)$. Note that if $\ALG$ answers $0$ in round $i$, then its expected cost is $p_i(1|w)\cdot t$, and if $\ALG$ answers $1$, then its expected cost is $p_i(0|w)$. It follows that
\begin{equation}
\E[\cost_i(\ALG)|W_i=w]\geq\min\left\{p_i(0| w), p_i(1| w)\cdot t\right\}.
\end{equation}
Let $d=D_{KL}(p_i(x|w)\|p_i)$. It follows from (\ref{Klem1}) of \cref{Klem} that $p_i(0|w)\geq K_{p(0),l}^{-1}(d)$ and $p_i(1|w)\geq K^{-1}_{p(1),l}(d)$. A lengthy but straightforward calculation reveals that for all $x\in[0,1]$, it holds that $K_{p(1)}(x/t)\leq K_{p(0)}(x)$. Using this inequality with $x=K_{p(0)}^{-1}(d)$ yields $K_{p(1),l}^{-1}(d)\cdot t\leq K_{p(0)}^{-1}(d)$. Thus, $\min\{p_i(0\vert w), p_i(1\vert w)\cdot t\}\geq K_{p(1),l}^{-1}(d)\cdot t$ and therefore
\begin{equation}
\E[\cost_i(\ALG)|W_i=w]\geq K_{p(1),l}^{-1}(d)\cdot t.
\end{equation}
Let $f(d)=K_{p(1),l}^{-1}(d)\cdot t$. Since $f$ is convex and decreasing, it follows from \cref{mainthm} that $\E[\cost(\ALG)]\geq n f(b/n)$. We will now calculate how large $b$ needs to be in order to ensure that the expected cost is at most $\alpha n$ for $0<\alpha< t/(1+t)$:
\begin{align}
\notag &\alpha n\geq nf(b/n)\\
\notag&\Rightarrow \alpha \geq K_{p(1),l}^{-1}(b/n)t\\
\notag&\Rightarrow K_{p(1)} (\alpha / t)\leq b/n\\
&\Rightarrow b\geq nK_{1/(1+t)}(\alpha / t). \label{thm16:calc}
\end{align}

\emph{The general case. }We are now ready to prove the lower bound for the general case. Let $0<\alpha<\frac{st}{s+t}$. Note that in particular, $\alpha<t$. If $s=t$, then BSG-$(s,t)$ is equivalent (after rescaling the costs) to $2$-SGKH and we are done. Assume therefore (without loss of generality) that $s<t$. Let $\ALG$ be a BSG-$(s,t)$ algorithm with advice complexity $b(n)$. Let $\ALG'$ be the BSG-$(1,t/s)$ algorithm which works exactly as $\ALG$. By definition, for every input string $x\in\{0,1\}^n$, it holds that $\ALG(x)=s\cdot \ALG'(x)$. It follows that if $\ALG$ incurs an expected cost of at most $\alpha n$ on inputs of length $n$, then $\ALG'$ never incurs an expected cost larger than $(\alpha / s)n$ for any input string of length $n$. It then follows from (\ref{thm16:calc}) that 
\begin{align*}
b(n)\geq K_{s/(s+t)}\left(\frac{\alpha}{t}\right)n
\end{align*}
\end{proof}
\subsubsection{Proof of \texorpdfstring{\cref{thm:betterbp}}{Theorem 4}: Bin packing}
\label{subsec:betterbp}
We will show how our lower bound for BSG-$(s,t)$ gives rise to an improved advice complexity lower bound for the classical bin packing problem. This illustrates that even for a problem which is not itself repeatable, it may still be possible to use our techniques to obtain improved lower bounds via e.g.\ a reduction from a repeated matrix game.

 Boyar et al.\ gave a reduction from binary string guessing ($2$-SGKH) to bin packing and used this to obtain a lower bound of $9/8$ for bin packing algorithms with sublinear advice~\cite{Abp}. Later, Angelopoulos et al.\ slightly modified the reduction and improved the lower bound to $7/6$~\cite{DBLP:conf/wads/AngelopoulosDKR15}. We will show that using the same reduction as in \cite{DBLP:conf/wads/AngelopoulosDKR15}, but reducing from \emph{weighted} string guessing further improves the lower bound to $4-2\sqrt{2}$. Note that $7/6=1.166\overline{6}$ while $4-2\sqrt{2}>1.1715$. We start by describing the reduction from \cite{DBLP:conf/wads/AngelopoulosDKR15}.
\paragraph{Reduction from string guessing to bin packing (from \cite{DBLP:conf/wads/AngelopoulosDKR15, Abp}).} Let $x=x_1\ldots x_n\in\{0,1\}^n$ be a binary string and let $\ab{x}_1$ be the hamming weight of $x$. We define a bin packing instance, $\sigma_x$, as follows: $\sigma_x$ consists of three phases. Phase 1 consists of $\ab{x}_1$ items all of size $1/2+\varepsilon$. Phase 2 consists of $n$ items with sizes in $(1/3, 1/2-\varepsilon)$. Initially, we set $l_0=1/3$ and $h_0=1/2-\varepsilon$ and $m_0=\frac{h_0+l_0}{2}$. For $1\leq i\leq n-1$, if $x_i=0$ then let $l_i=m_{i-1}$ and $h_i=h_{i-1}$, and if $x_i=1$ then let $l_i=l_{i-1}$ and $h_i=m_{i-1}$. In both cases, let $m_i=(l_i+h_i)/2$. The $i$th item of phase 2 will have size $m_i$. If $x_i=0$, we say that the $i$th item is \emph{small} and if $x_i=1$, we say that it is \emph{large}. Note that exactly $\ab{x}_1$ items will be large and $n-\ab{x}_1$ items will be small. Let $I_s$ be the set of small items and let $I_l$ be the set of large items. Phase 3, the final phase of $\sigma_x$, consists of $n-\ab{x}_1$ items: For each small item $a\in I_s$, phase 3 contains an item of size $1-a$ (here, we follow the standard bin packing convention of using the same symbol for both the item and the size of the item).

Note that $\ab{\sigma_x}=2n$ and $\OPT(\sigma_x)=n$. Indeed, an optimal offline algorithm will pack all large items on top of a phase 1 item and will open a new bin for all small items. It will then pack all phase 3 items on top of the corresponding small item.

Given a bin packing algorithm, $\ALG$, we define a string guessing algorithm, $\ALG'$, as follows: On input $x\in\{0,1\}^n$, the algorithm $\ALG'$ will simulate $\ALG$ on $\sigma_x$. In order to do this, $\ALG'$ is provided with the advice, $\varphi_x$, read by $\ALG$ on input $\sigma_x$. However, note that in order for $\ALG'$ to know how long phase 1 of $\sigma_x$ is, $\ALG'$ will need to know $\ab{x}_1$. Thus, $\ab{x}_1$ is written onto the advice tape of $\ALG'$ in a self-delimiting way using $2\lceil \log n\rceil +1=O(\log n)$ bits. Since $\ALG'$ knows $x_1\ldots x_{i-1}$ when it has to output its guess, $y_i$, in round $i$ (and since it knows $\ab{x}_1$ and $\varphi_x$ from its advice), it is clear that $\ALG'$ can determine into which bin $\ALG$ packs the $i$th item, $a_i$, of phase 2 before deciding on its answer $y_i$. If $a_i$ is packed on top of a phase 1 item, then $\ALG'$ will answer $y_i=1$. If $a_i$ is used to open a new bin, then $\ALG'$ will answer $y_i=0$. If $a_i$ is packed on top of another phase 2 item, then $\ALG'$ will answer $y_i=1$. Intuitively, $\ALG'$ answers $0$ if the packing of $a_i$ indicates that $\ALG$ thinks $a_i$ is a small item (and $\ALG'$ answers $1$ if $\ALG$ seems to think that $a_i$ is large).

We say that a string guessing algorithm makes a \emph{$0$-mistake} if it answers $1$ when the correct answer is $0$, and we define \emph{$1$-mistake} symmetrically. Let $e_0$ be the number of $0$-mistakes made by $\ALG'$ and let $e_1$ be the number of $1$-mistakes. Angelopoulos et al.\ proves in \cite{DBLP:conf/wads/AngelopoulosDKR15} the following relation between the number and type of mistakes made by $\ALG'$ on the input string $x$ and the number of additional (when compared to $\OPT$) bins opened by $\ALG$ when given $\sigma_x$ as input:
\begin{align}
e_0&\leq 2(\ALG(\sigma_x)-n),\label{e0mistakes}\\
e_1&\leq (\ALG(\sigma_x)-n).\label{e1mistakes}
\end{align}
They then observe that this implies $e_0+e_1=\ALG'(x)\leq 3(\ALG(\sigma_x)-n)$ and use the hardness result for $2$-SGKH to obtain a lower bound for bin packing with sublinear advice. Note that by considering the unweighted sum $e_0+e_1$, the asymmetry present in the reduction between $0$-mistakes and $1$-mistakes is not being taken into account. We will now show that using our hardness result for weighted string guessing, we can use this asymmetry between $0$-mistakes and $1$-mistakes in the reduction to obtain a better lower bound for bin packing. For convenience, we restate \cref{thm:betterbp} before giving the proof.

\restatebetterbp*
\begin{proof}[Proof of \cref{thm:betterbp}]
Fix $t\geq 1$. The proof is by a reduction from BSG-$(1,t)$. We will show that if 
\begin{equation}
c<\frac{t}{(1+t)(2+t)}+1
\label{bpasump}
\end{equation}
 and $c$ is constant, then every $c$-competitive randomized bin packing algorithm must read at least $\Omega(n)$ bits of advice.

Fix $n\in\mathbb{N}$. Let $\ALG$ be a deterministic bin packing algorithm reading $b(2n)$ bits of advice on inputs of length $2n$. Using the reduction of Angelopoulos et al.\ \cite{DBLP:conf/wads/AngelopoulosDKR15} (which is described in details before the statement of \cref{thm:betterbp}), we convert $\ALG$ into a BSG-$(1,t)$ algorithm, $\ALG'$, which on input $x\in\{0,1\}^n$ works by simulating $\ALG$ on the bin packing instance $\sigma_x$ of length $2n$. From (\ref{e0mistakes}) and (\ref{e1mistakes}), it follows that $\ALG'$ will make at most $2(\ALG(\sigma)-n)$ $0$-mistakes and at most $(\ALG(\sigma)-n)$ $1$-mistakes. Thus, $\ALG'(x)=e_0+te_1\leq (2+t)(\ALG(\sigma_x)-n)$.

Let $\ra$ be a randomized $c$-competitive bin packing algorithm reading $b(2n)$ bits of advice on inputs of length $2n$. We can convert $\ra$ into a randomized BSG-$(1,t)$ algorithm, $\ra'$, by applying the reduction described above to all algorithms in the support of the probability distribution over deterministic algorithms specified by $\ra$. Doing so, we get that $\E[\ra'(x)]\leq (2+t) (\E[\ra(\sigma_x)]-n)$. Since $\ra$ is $c$-competitive, this implies that $ \E[\ra'(x)]\leq (2+t)(cn+\alpha-n)=(2+t)(c-1)n+\alpha(2+t)$ for some additive constant $\alpha$ (here, we used that $\OPT(\sigma_x)=n$).

Since $c$ is a constant satisfying (\ref{bpasump}), an easy calculation yields $\E[\ra'(x)]\leq(t/(1+t)-\varepsilon)n+\alpha(2+t)$ for some constant $\varepsilon>0$. Since $\alpha(2+t)$ is constant, this implies that for all $n$ large enough, we have $\E[\ra'(x)]\leq (t/(1+t)-\frac{\varepsilon}{2})n$. Recall that $\ab{\sigma_x}=2n$ and that $\ra'$ therefore reads $b(2n)+2\lceil \log n\rceil +1$ bits of advice when given $x$ as input (because of how the reduction from string guessing to bin packing works). From \cref{thm:wsgb} (or simply \cref{RMGmainthm}), we get that $b(2n)+2\lceil \log n\rceil +1=\Omega(n)$. This implies that $b(n)=\Omega(n)$. The proof is finished by setting $t=\sqrt{2}$ for which (\ref{bpasump}) attains its maximum value of $4-2\sqrt{2}$.

\end{proof}

\section{Superlinear advice lower bound for graph coloring}
Up until now, we have almost exclusively shown linear advice complexity lower bounds. In this section, we will show to use our direct product theorems to obtain the following superlinear lower bound for online graph coloring (see \cite{DBLP:conf/dagstuhl/Kierstead96} for a survey on the problem): For every fixed $\varepsilon>0$, a randomized graph coloring algorithm with advice complexity $o(n\log n)$ must have a competitive ratio of at least $\Omega(n^{1-\varepsilon})$. Previously, it was only known that $\Omega(n\log n)$ bits of advice are were necessary to be $1$-competitive~\cite{Agraphc}. Note that $O(n\log n)$ bits of advice trivially suffice to be optimal. Thus, our lower bound for graph coloring shows that there is a drastic phase-transition in the advice complexity of the problem. 

It is easy to see that for every $c=n^{1-o(1)}$, there exists a $c$-competitive graph coloring algorithm reading $o(n \log n)$ bits of advice. Indeed, let $k$ be the chromatic number of the input graph and suppose that $n/k\leq c$. Then we may just color the graph greedily, since even using $n$ colors is good enough to be $c$-competitive. On the other hand, if $n/k> c$, then $k<n/c$. But if $c=n^{1-o(1)}$, then this implies that $k=n^{o(1)}$. Using $b=n\lceil\log k\rceil =n\log (n^{o(1)})$ bits of advice (plus some small lower-order term to encode $k$ in a self-delimiting way) is therefore enough to be $c$-competitive. But if $b=n\log (n^{o(1)})$, then $b=o(n\log n)$. For example, one can obtain a competitive ratio of $O(n/\log^{O(1)}n)$ using $O(n\log (\log^{O(1)} n))=O(n\log\log n)$ bits of advice. Thus, our lower bound saying that $\Omega(n\log n)$ bits are needed to be $O(n^{1-\varepsilon})$-competitive for every constant $\varepsilon>0$ is essentially tight, in the sense that $O(n^{1-\varepsilon})$ cannot be increased to $n^{1-o(1)}$.

In order to prove a lower bound for algorithms with superlinear advice, we apply our martingale-theoretic direct product theorem (\cref{dpmartin}) to a hard input distribution for graph coloring due to Halld{\'o}rsson and Szegedy~\cite{MS}. They show that a randomized graph coloring algorithm without advice must have a competitive ratio of at least $\Omega(n/\log^2 n)$. By slightly modifying some of the parameters in their construction and combining it with \cref{dpmartin}, we obtain a lower bound of $\Omega(n\log n)$ on the advice complexity of algorithms with a competitive ratio of $O(n^{1-\varepsilon})$. We remark that one can equally well use \cref{mainthm} to obtain this lower bound, but we believe that \cref{dpmartin} is easier to apply in this case.

\label{sec:graphc}

We begin by formally defining the abstract guessing game which underlies the lower bound in \cite{MS}.
\begin{definition}[cf.\ \cite{MS}]
\label{def:hsgg}
The \HS (HSGG) is a minimization problem. At the beginning, two integers $k$ and $n$ where $k\leq n$ and $k$ is even are revealed to the algorithm. The problem consists of $n$ rounds. For $1\leq i \leq n$, round $i$ proceeds as follows:
\begin{enumerate}
\item A set $A_i\in [k]^{k/2}$ of \emph{available characters} is revealed to the algorithm.
\item The algorithm answers $y_i\in [n]$ subject to the following feasibility constraint: 
\begin{equation}
\text{For each $1\leq t<i$, if $y_t=y_i$, then $x_t$ must belong to $A_i$.}
\label{feas}
\end{equation}
\item The \emph{correct character} $x_i\in A_i$ is revealed to the algorithm.
\end{enumerate} 
The cost of the \emph{output} $y=y_1\ldots y_n\in [n]^n$ is the number of different characters from $[n]$ in $y$, i.e., the cost is $\ab{\{y_1,\ldots , y_n\}}$.
\end{definition}

Note that there will always be at least one character in $[n]$ satisfying the feasibility constraint (\ref{feas}). This follows since the game has $n$ rounds and since a character which have never been used before by the algorithm is always feasible. Furthermore, we remark that the correct character $x_i$ will not in general be a feasible answer for an online algorithm in round $i$. Namely, it could be that the algorithm in a previous round $t<i$ (wrongly) answered $y_t=x_i$. If $x_t\notin A_i$, then answering $x_i$ is not possible in round $i$. However, the cost of an optimal solution to an instance of HSSG is always at most $k$. Indeed, since \OPT knows the correct characters $x_1\ldots x_n\in [k]^n$, \OPT may simply answer $y_i=x_i$ in each round. To see that this is possible, note that if $y_t=y_i$ for some $1\leq t<i$, then $y_t=x_i$ and hence $x_t=y_t=x_i\in A_i$. Thus, $x_i$ will be a valid answer in round $i$ since \OPT has not previously made any mistakes.

While $k$ is an upper bound on $\OPT$, it is not a tight bound (even if all $k$ characters actually appears in $x_1\ldots x_n$). Consider the instance with $n=4,k=4$, correct characters $x_1x_2x_3x_4=1234$, and $A_1=A_2=\{1,2\}$, $A_3=A_4=\{3,4\}$. In this case, $y=1133$ is a feasible output of cost $2$.

Note that $O(n\log k)$ bits of advice allows for an optimal HSGG-algorithm.

\begin{theorem}
\label{HSGGadvice}
Let $c\in\mathbb{N}$, let $\delta>1/c$, and let $\ALG$ be a deterministic HSGG-algorithm such that on inputs where $n=k^c$, the cost of $\ALG$ is at most $O(n^{1-\delta})$. Then, the algorithm $\ALG$ must read at least $b=\Omega(n\log n)$ bits of advice.
\end{theorem}
\begin{proof}[Proof (cf.\ \cite{MS})]
Fix $c\in\mathbb{N}$. Let $k$ be an even square number and let $n=k^c$. We want to apply \cref{dpmartin}. To this end, we construct an $n$-round input distribution. Each round of the distribution will consist of exactly one request (and hence correspond to a single round of HSGG). In each of the $n=k^c$ rounds, the set of available characters $A_i\in[k]^{k/2}$ is selected uniformly at random (and independently of all previous choices made by the adversary). The correct character $x_i\in A_i$ is then selected uniformly at random~from~$A_i$.

Recall that when specifying an $n$-round input distribution, we also need to specify a cost function for each round. Following \cite{MS}, we will not use directly the cost function from the definition of HSGG. Instead, for each $1\leq i\leq n$, we define an auxiliary cost function, $\pair_i$, as follows: Let $y=y_1\ldots y_n$ be the output computed by some algorithm and let $x=x_1\ldots x_n$ be the correct characters. If $(y_i, x_i)\notin \{(y_t,x_i): 1\leq t<i\}$, i.e., round $i$ is the first time the algorithm answers $y_i$ while the correct character is $x_i$, then we say that $(y_i, x_i)$ is a \emph{new pair} (and that $y_i$ and $x_i$ are \emph{paired}). If $(x_i, y_i)$ is a new pair, then $\pair_i(y,x)=1$. Otherwise, if $y_i$ and $x_i$ were already paired previous to round $i$, then $\pair_i(y,x)=0$. It was observed in \cite{MS} that if $\pair(y,x)=\sum_{i=1}^n\pair_i(y,x)=m$, then the actual cost of the output $y$ is at least $\frac{m}{k/2}$ (since $\ab{A_i}=k/2$ for all $i$, an output character cannot be paired with more than $k/2$ correct characters).   

Let $\DET$ be a deterministic HSSG-algorithm without advice. We want to show that the probability that $\pair_i(\DET)$ is no more than $n/2$ is very small. For a round $i$ and an output character $j\in [n]$, let $H_i(j)=\{x_t \colon y_t=j\text{ and } t<i\}$ be those correct characters which prior to round $i$ were paired with $j$ (note that $H_i(j)$ depends on $\DET$). Define

\begin{equation}
\dist(A_i, \DET)=\min_{\mbox{\small $j\in [n]\colon H_i(j)\subseteq A_i$}} \ab{A_i \setminus H_i(j)}.
\end{equation}

We remark that $\dist(A_i,\DET)$ is a random variable. Suppose that $\DET$ outputs $y_i$ in round $i$. Then it must hold that $H_i(y_i)\subseteq A_i$ since otherwise $y_i$ was not a feasible answer in that round. Furthermore, since $x_i$ is selected uniformly at random from $A_i$, the probability that $(y_i, x_i)$ is a new pair is $\frac{\ab{A_i\setminus H_{i}(y_i)}}{\ab{A_i}}$. Thus, for every history $w\in\mathcal{W}_i$,

\begin{equation}
\Pr \left[\pair_i(\DET)=1\middle| W_i=w\right]\geq \frac{\E[\dist (A_i, \DET)| W_i=w]}{k/2}.
\end{equation}

The random variable $\dist (A_i, \DET)$ depends in a complicated way on the previous rounds of the guessing game. However, suppose we want an upper bound on $\Pr[\dist (A_i, \DET)\leq s]$ for some $s\in\mathbb{N}$. No matter how the previous rounds have proceeded, for each of the $n$ output characters $j\in [n]$, the set $H_i(j)$ is a subset of $[k]$ of size at most $k/2$. Furthermore, $H_i(j)$ is a subset of at most $\binom{k-\ab{H_i(j)}}{k/2-\ab{H_i(j)}}$ sets of size $k/2$. Thus, if $\ab{H_i(j)}\geq k/2-s$, then $H_i(j)$ is a subset of at most $\binom{k/2+s}{s}$ sets of size $k/2$. Using that there are $\binom{k}{k/2}$ choices for $A_i$ (each of which is equally likely) and the union bound, we get that
\begin{align*}
\Pr \left[\operatorname{dist} (A_i, \DET)\leq k/2-\sqrt{k} \middle| W_i=w\right]&\leq\frac{\binom{k-\sqrt{k}}{k/2-\sqrt{k}}\cdot n}{\binom{k}{k/2}}\\
&\leq \frac{2^{k-\sqrt{k}}\cdot k^c}{2^k \cdot k^{-1}}\\
&=2^{-\sqrt{k}}k^{c+1}.
\end{align*}
Here, we used that $\binom{k}{k/2}\geq 2^k / k$. We can now bound from below the expected value of $\dist(A_i, \DET)$. Set $s=k/2-\sqrt{k}$. Then,
\begin{align*}
\E\left[ \dist(A_i, \DET)\middle| W_i=w \right] &\geq s\cdot \Pr\left[\dist (A_i, \DET)\geq s \middle| W_i=w\right]\\
&\geq s\cdot \left(1-\Pr\left[\dist (A_i, \DET)\leq s \middle| W_i=w\right]\right)\\
&\geq \left(k/2-\sqrt{k}\right)\left(1-2^{-\sqrt{k}}k^{c+1}\right).
\end{align*}
Thus,
\begin{align*}
  \frac{\E\left[\dist(A_i, \DET)\middle| W_i=w\right]}{k/2}&\geq \frac{\left(k/2-\sqrt{k}\right)\left(1-2^{-\sqrt{k}}k^{c+1}\right)}{k/2}\\
&\geq \left(1-\frac{2}{\sqrt{k}}\right)\cdot\left(1-\frac{k^{c+1}}{2^{\sqrt{k}}}\right)\\
&\geq \left(1-\frac{2}{\sqrt{k}}\right)^2\text{ for $k\geq K$}.
\end{align*}
Here, $K$ is chosen such that $(1-2/\sqrt{k})\leq (1-2^{-\sqrt{k}}k^{c+1})$ for all $k\geq K$. This is always possible since $k^{c+1}\sqrt{k}=o(2^{\sqrt{k}})$. Set $p:=(1-2/\sqrt{k})^2$. For the rest of the proof, we assume $k\geq K$, and we assume $k\geq 49$ so that $p\geq 1/2$.

Just as $\dist(A_i, \DET)$, the random variable $\pair_i(\DET)$ is obviously not independent of $\pair_t(\DET)$ for $t<i$. However, the lower bound derived above holds for every possible history $w\in\mathcal{W}_i$ of the game at round $i$. Thus, $\pair_i(\DET)$ is a $\{0,1\}$-valued random variable with conditional expectation $\E[\pair_i(\DET) \vert W_i=w]=\Pr[\pair_i(\DET)=1 \vert W_i=w]\geq p$ for every $w\in\mathcal{W}_i$. From this, it follows that for every $w\in\mathcal{W}_i$,

\begin{align*}
\E\left[\left(\pair_i(\DET)-p\right)^2\middle| W_i=w\right]&=\Pr[\pair_i(\DET)=0 \vert W_i=w]\cdot (0-p)^2\\ 
&\phantom{{}={}}+\Pr[\pair_i(\DET)=1 \vert W_i=w]\cdot (1-p)^2\\
&=\left(1-\Pr[\pair_i(\DET)=1 \vert W_i=w]\right)\cdot p^2\\
&\phantom{{}={}}+\Pr[\pair_i(\DET)=1 \vert W_i=w]\cdot (1-p)^2\\
&\leq p(1-p).
\end{align*}
The last inequality above holds since we are assuming that $p\geq 1/2$. let $\ALG$ be a deterministic HSGG-algorithm with advice complexity $b(n)$ such that on inputs where $n=k^c$, the cost of $\ALG$ is at most $O(n^{1-\delta})$.  If we let $\varepsilon=p-1/2$, it follows from \cref{dpmartin} and our calculations above that
\begin{align*}
\Pr[\pair(\ALG)\leq 1/2 n]&=\Pr[\pair(\ALG)\leq (p-\varepsilon)n]\\
&\leq \exp_2\left(b(n)-K_{\mbox{\Large$\frac{\gamma}{1+\gamma}$}}\left(\frac{\alpha+\gamma}{1+\gamma}\right)n\right)\\
&=\exp_2\left(b(n)-K_{1-p}\left(\frac{1}{2}\right)n\right),
\end{align*}
where $\gamma=(p(1-p))/p^2=\frac{1-p}{p}$ and $\alpha=\varepsilon / p=1-\frac{1}{2p}$. Now,
\begin{align*}
K_{1-p}(1/2)&=1/2\cdot \log\left(\frac{1/2}{1-p}\right)+1/2\cdot \log\left(\frac{1/2}{p}\right)\\
&=1/2\cdot\log\left(\frac{1}{4p(1-p)}\right)\\
&=1/2\cdot\log\left(\frac{k^2}{16(\sqrt{k}-2)^2(\sqrt{k}-1)}\right)\\
&\geq 1/2 \cdot\log\left(\frac{k^2}{16 (\sqrt{k})^2(\sqrt{k})}\right)\\
&=1/2\cdot\log\left(\frac{\sqrt{k}}{16}\right).
\end{align*}
Thus,
\begin{equation}
\label{hsgg:maineqproof}
\Pr[\pair(\ALG)\leq n/2]\leq \exp_2\left(b(n)-\frac{1}{2}\log\left(\frac{\sqrt{k}}{16}\right)n\right).
\end{equation}

Recall that the cost of $\ALG$ is assumed to be at most $O(n^{1-\delta})$ for some fixed constant $\delta>1/c$. Assume by way of contradiction that $b(n)\in o(n \log n)$. Under this assumption, if $k$ (and thereby also $n=k^c$) is sufficiently large, then $b(n)<\frac{1}{2}\log\left(\frac{\sqrt{k}}{16}\right)n$ since $\frac{1}{2}\log\left(\frac{\sqrt{k}}{16}\right)n=\frac{1}{2}\log\left(\frac{n^{1/c}}{16}\right)n=\Omega(n\log n)$. Thus, for all sufficiently large $k$ and $n=k^c$, it follows from (\ref{hsgg:maineqproof}) that there exists an input, $\sigma$, such that $\pair(\ALG(\sigma))\geq n/2$ and, hence, \mbox{$\ALG(\sigma)\geq\frac{\pair(\ALG(\sigma))}{k/2}\geq\frac{n/2}{k/2}=n^{1-1/c}$}. This contradicts the claimed performance guarantee of $\ALG$ since $1-1/c>1-\delta$. We conclude that $b(n)\in \Omega(n\log n)$.

\end{proof}

The following lemma which provides a reduction from HSGG to graph coloring follows directly from \cite{MS}.
\begin{lemma}[\cite{MS}]
\label{lem:hsgggraph}
If there exists a deterministic graph-coloring algorithm with advice complexity $b(n)$ which for $n$-vertex $k$-colorable graphs uses at most $f(n,k)$ colors, then there exists a HSGG-algorithm with advice complexity $b(n)$ which always incurs a cost of at most $f(n,k)$.
\end{lemma}
Combining \cref{HSGGadvice} with \cref{lem:hsgggraph}, we obtain a proof of \cref{thm:graphcollow} (we restate the theorem here for convenience).

\restategraphcollow*
\begin{proof}[Proof of \cref{thm:graphcollow}]
Let $\ALG$ be a deterministic $c(n)$-competitive graph coloring algorithm where $c(n)\in O(n^{1-\varepsilon})$. By definition, $\ALG(\sigma)\leq c(n)\cdot \OPT(\sigma)+\alpha$ for some constant $\alpha$. In particular, this means that $\ALG$ can color $n$-vertex $k$-colorable graphs online using at most $f(n,k)=c(n)\cdot k+\alpha$ colors. By the reduction from HSGG to graph coloring (\cref{lem:hsgggraph}), we get a HSGG-algorithm, $\ALG'$, which incurs a cost of at most $f(n,k)$ and which has the same advice complexity as $\ALG$.

Choose $c\in\mathbb{N}$ such that $2/c<\varepsilon$ and consider HSGG-inputs where $n=k^c$. On these inputs, $\ALG'$ incurs a cost of at most $f(n, n^{1/c})=c(n)\cdot n^{1/c}+\alpha=O(n^{1-\varepsilon}n^{1/c})=O(n^{1-\delta})$ where $\delta:=\varepsilon-1/c$. Since $2/c<\varepsilon$, it follows that $\delta=\varepsilon-1/c>1/c$. By \cref{HSGGadvice}, $\ALG'$ must therefore use $\Omega(n\log n)$ bits of advice. Since the advice complexity of $\ALG'$ was the same as that of $\ALG$, this proves \cref{thm:graphcollow} for deterministic algorithms.

We still need to consider the case where the graph coloring algorithm is randomized. Note that for a graph on $n$ vertices, the number of requests in a single round can be at most $2^n$ (there are at most $2^n$ possible subsets of incident vertices for the vertex currently being revealed). Thus, the number of inputs of length $n$ for online graph coloring is at most $(2^n)^n=2^{n^2}=\dri$. It follows from \cref{thm:derandmin} in \cref{sec:derand} that a randomized $O(n^{1-\varepsilon})$-competitive graph coloring algorithm with advice complexity $b(n)$ can be converted to a deterministic $O(n^{1-\varepsilon})$-competitive graph coloring algorithm with advice complexity $b(n)+O(\log n)$. From the above, we get that $b(n)+O(\log n)=\Omega(n\log n)$ from which it follows that $b(n)=\Omega(n\log n)$.
\end{proof}

\section{Lower bounds for \texorpdfstring{$\lor$}{v}-repeatable problems}
\label{sec:maxrep}
So far, we have only considered $\Sigma$-repeatable (minimization) problems, where the cost of a solution in $\P^*$ is the sum of costs in each round. In this section, we consider $\lor$-repeatable (minimization) problems, where the cost of a solution in $\P^*$ is the \emph{maximum} of the costs in each round. \cref{lorrep} shows that lower bounds for $\lor$-repeatable problems can often be obtained rather easily. We remark that the proof of \cref{lorrep} does not rely on our previous direct product theorems.
\begin{theorem}
\label{lorrep}
Let $\P$ be a strictly $\lor$-repeatable problem and let $I=\{\sigma_1,\ldots , \sigma_m\}$ be a finite set of $\P$-inputs. Furthermore, let $t=\max_{\sigma\in I}\OPT(\sigma)$ and let $\varepsilon>0$ be a constant. Suppose that for every deterministic $\P$-algorithm without advice, $\ALG$, there exists some $1\leq i\leq m$ such that $\ALG(\sigma_i)\geq k$. Then, for every randomized $\P$-algorithm, $\ra$, reading $o(n)$ bits of advice, there exists a $\P$-input $\sigma$ such that $\E[\ra(\sigma)]\geq (1-\varepsilon)k$ and such that $\OPT(\sigma)\leq t$.
\end{theorem}
\begin{proof}
The proof is via Yao's principle. Thus, let $\ALG$ be an arbitrary deterministic $\P$-algorithm reading $b(n)\in o(n)$ bits of advice on inputs of length $n$. We will first construct a hard input distribution for $\Prm$, and then use that $\P$ is $\lor$-repeatable.

For each $r\in\mathbb{N}$, we define an $r$-round input distribution, $p^r$, for $\Prm$ as follows: In each round, the adversary selects uniformly at random an input from $I$. When a round ends, a reset occurs (thus, each round of the $r$-round input distribution corresponds to a round of $\Prm$). Let $n'=\max_{\sigma\in I}\ab{\sigma}$ be the length of the longest input in $I$. Note that an input in the support of $p^r$ has length at most $n'r$.

By assumption, $b:\mathbb{N}\rightarrow\mathbb{N}$ is a function such that $b(n)\in o(n)$. Since $n'$ is a constant, this means that $b(n'r)\in o(r)$. Thus, since $1-1/m<1$, it follows that $2^{b(n'r)}(1-1/m)^r\rightarrow 0$ as $r\rightarrow\infty$. Choose $r$ large enough that $2^{b(n'r)}(1-1/m)^r\leq \varepsilon$.

The number $r$ of rounds is now fixed. Following \cref{lem:repadv}, we will define a $\Prm$-algorithm $\ALG^*$ based on the $\P$-algorithm $\ALG$. The algorithm $\ALG^*$ will only be defined for inputs in $\supp(p^r)$. For other inputs, it can behave arbitrarily. By definition, for each input length $n$, there exists $2^{b(n)}$ deterministic $\P$-algorithms, $\ALG_1,\ldots , \ALG_{2^{b(n)}}$ such that $\ALG(\sigma)=\min_{j}\ALG_j(\sigma)$ for every $\P$-input $\sigma$ of length at most $n$. Since $\P$ is strictly $\lor$-repeatable, this implies that for every $n$, there exists $2^{b(n)}$ deterministic $\Prm$-algorithms, $\ALG_1^*,\ldots , \ALG_{2^{b(n)}}^*$ such that $\ALG_j^*(\sigma^*)\leq \ALG_j(g(\sigma^*))$ for every $\Prm$-input $\sigma^*$ of length at most $n$ (where $g$ is as in \cref{re}). Recall that an input in $\supp (p^r)$ has length at most $n'r$. We define $\ALG^*$ to be the following $\Prm$-algorithm: For an input $\sigma^*\in\supp (p^r)$, the algorithm $\ALG^*$ will read $b(n'r)$ bits of advice and use these to select the best of the $2^{b(n'r)}$ algorithms $\ALG_1^*,\ldots , \ALG_{2^{b(n'r)}}^*$ for this particular input $\sigma^*$. Note that since $r$ is fixed and $n'$ is a constant, $\ALG^*$ can compute $n'r$ itself (and thereby also $b(n'r)$). By definition, $\ALG^*(\sigma^*)\leq \ALG(g(\sigma^*))$ for every input $\sigma^*\in \supp (p^r)$.

Let $1\leq j\leq 2^{b(n'r)}$ and let $1\leq i\leq r$. We consider the computation of $\ALG^*_{j}$ in round $i$ of $p^r$. The computation of $\ALG_j^*$ in round $i$ may depend on the request revealed in previous rounds. However, recall that a reset occurred just before round $i$ began. Thus, for each fixed history $w\in\mathcal{W}_i$, the computation of the algorithm $\ALG_j^*$ in round $i$ given the history $w$ defines a deterministic $\P$-algorithm valid for all inputs in $I$. Thus, by assumption, there exists at least one input $\sigma\in I$ such that if $\sigma$ is the input in round $i$, then $\ALG_j^*$ incurs a cost of at least $k$ in round $i$. Since each of the $m$ inputs in $I$ have probability $1/m$ of being selected in round $i$, this implies that for \emph{every} $w\in\mathcal{W}_i$,
\begin{equation}
\Pr[\cost_i(\ALG_j^*) < k\vert W_i=w]\leq 1-1/m.
\label{eq:lorrepallhistory}
\end{equation}
Thus, even though the costs incurred by $\ALG_j^*$ in each of the $r$ rounds are not (necessarily) independent, we can use \cref{eq:lorrepallhistory} to bound from above the probability that the maximum cost over all $r$ rounds is less than $k$:
\begin{align*}
\Pr[\cost(\ALG_j^*)< k]&=\Pr[\forall 1\leq i\leq r\colon \cost_i(\ALG_j^*)<k]\\
&\leq (1-1/m)^r.
\end{align*}
Applying the union bound over all $2^{b(n'r)}$ algorithms $\ALG_1^*,\ldots , \ALG_{2^{b(n'r)}}$ gives
\begin{align}
\notag \Pr\left[\cost(\ALG^*)<k \right]&=\Pr\left[\exists 1\leq j\leq 2^{b(n'r)}\colon  \cost(\ALG_j^*)< k\right] \\
&\notag \leq 2^{b(n'r)}(1-1/m)^r\\
&\leq \varepsilon. \label{lorreppeq}
\end{align}
It follows that
\begin{align*}
\E [\cost(\ALG^*)]&\geq \Pr [\cost(\ALG^*)\geq k]\cdot k\\
&\geq (1-\varepsilon)\cdot k.
\end{align*}
This implies that $\E_{\sigma^*\sim p^r}[\ALG(g(\sigma^*))]\geq \E_{\sigma^*\sim p^r}[\ALG^*(\sigma^*)]\geq (1-\varepsilon)k$. 

Using Yao's principle, we conclude that for every randomized $\P$-algorithm $\ra$ with advice complexity $o(n)$, there exists an $r\in\mathbb{N}$ and a $\sigma^*\in\supp(p^r)$ such that $\E[\ra(g(\sigma^*))]\geq (1-\varepsilon)k$. Since $\OPT(g(\sigma^*))\leq \OPT^*_{\lor}(\sigma^*)\leq t$ for every $\sigma^*\in\supp(p^r)$, this finishes the proof.
\end{proof}
\subsection{Applications}
\label{sec:lorapp}
\paragraph{Application: Edge coloring}
It is easy to see that online edge coloring is strictly $\lor$-repeatable (see \cite{Bar-Noy} for a formal definition of the problem). In \cite{Bar-Noy}, it is shown that for every $\Delta\in\mathbb{N}$, there exists a finite set of inputs $I_\Delta$ with the following properties: For $\sigma\in I_\Delta$, the underlying graph of $\sigma$ is a forest (and hence $\Delta$-edge-colorable). Also, for every deterministic algorithm $\DET$ without advice, there exists some $\sigma\in I_\Delta$ such that $\DET(\sigma)\geq (2\Delta -1)$. We can use this construction together with \cref{lorrep}. Let $\varepsilon>0$ be a fixed constant. Suppose by way of contradiction that $\ALG$ is a $(2-\varepsilon)$-competitive randomized edge coloring algorithm with advice complexity $o(n)$. Then there exists some $\alpha$ such that $\E[\ALG(\sigma)]\leq(2-\varepsilon)\OPT(\sigma)+\alpha$ for all inputs $\sigma$. Choose $\Delta$ large enough so that $(2-\varepsilon)\Delta+\alpha < 2\Delta-1$. By \cref{lorrep}, there exists an input $\sigma\in I_\Delta$ such that $\E[\ALG(\sigma)]\geq 2\Delta-1 > (2-\varepsilon)\Delta+\alpha=(2-\varepsilon)\OPT(\sigma)+\alpha$. This is a contradiction. Since $\varepsilon$ was arbitrary, we conclude that the competitive ratio of a randomized edge-coloring algorithm with advice complexity $o(n)$ must be at least $2$. This result was already obtained in \cite{Aedge} where it shown directly how an adversary may repeatedly use the lower bound for algorithms without advice. \cref{lorrep} reveals that the proof of the lower bound in \cite{Aedge} is in fact just a special case of a much more general technique.
\paragraph{Application: Graph coloring}
We have already shown previously how to obtain good advice complexity lower bounds for graph coloring on general graphs (see \cref{sec:graphc}). We will now consider the $2$-vertex-coloring problem (i.e., the online graph coloring problem restricted to bipartite graphs).

It is easy to see that online $2$-vertex-coloring is strictly $\lor$-repeatable. Furthermore, it is shown in \cite{GyarfasL88} that for every $n$ which is a power of two, there exists a finite set of inputs, $I_n$, with the following properties: For $\sigma\in I_n$, the graph defined by the input $\sigma$ is a $2$-colorable graph on $n$ vertices. Furthermore, for every deterministic online graph coloring algorithm $\DET$ there exists a $\sigma\in I_n$ such that $\DET(\sigma)\geq \log_2 (n)$. 

We will now apply \cref{lorrep}. Let $c=O(1)$ be a constant. Assume by way of contradiction that $\ALG$ is a $c$-competitive randomized $2$-vertex-coloring algorithm and that the advice complexity of $\ALG$ is $o(n)$. This means that there exists an additive constant $\alpha$ such that $\E[\ALG(\sigma)]\leq c\cdot \OPT(\sigma)+\alpha$ for every input $\sigma$ where the underlying graph is bipartite. Choose $n$ large enough so that $2c + \alpha < \log_2 (n)$ (and such that $n$ is a power of two). By \cref{lorrep}, we get that for this choice of $n$, there exists an input $\sigma\in I_n$ such that $\E[\ALG(\sigma)]\geq \log_2(n)> 2c+\alpha=c\cdot\OPT(\sigma)+\alpha$. But this is a contradiction since $\ALG$ was assumed to be $c$-competitive with an additive constant of $\alpha$. We conclude that if a randomized algorithm with advice complexity $o(n)$ is $c$-competitive, then it must be the case that $c=\omega(1)$.

Previously, is was only known that $\Omega(n)$ bits of advice was needed to never use more than $3$ colors when coloring a bipartite graph online~\cite{Avbipartite}. Thus, it was known that an algorithm with advice complexity $o(n)$ would use $4$ colors on some graphs and therefore have a strict competitive ratio of at least $2$.

\paragraph{Application: $L(2,1)$-coloring on paths and cycles}
See \cite{AL21} for a formal definition of the problem. The input is a graph of maximum degree two. When a vertex is revealed, an online algorithm must irrevocably assign a non-negative integer to the vertex such that the integers assigned to the set of vertices is a valid $L(2,1)$-coloring. The cost of a solution is $\lambda$, where $\lambda$ is the smallest integer such that the set of colors used in the coloring is contained in $\{0,1,\ldots , \lambda\}$. It is easy to see that $\OPT(\sigma)\leq 4$ for every input $\sigma$.

In \cite{AL21}, the following results were obtained: It is shown that there exists a deterministic algorithm $\DET$ without advice such that $\DET(\sigma)\leq 6$ for every input $\sigma$. A matching lower bound is provided by showing that there exists a finite set of inputs $I$ such that, for each deterministic algorithm $\ALG$ without advice, $\ALG(\sigma)\geq 6$ for some $\sigma\in I$. Finally, it is also shown in \cite{AL21} that a deterministic algorithm $\ALG$ which can guarantee to solve the $L(2,1)$ problem on paths with a cost of at most $5$ must read at least $3.9402\cdot 10^{-10}n=\Omega(n)$ bits of advice. Since costs are integral and $\OPT(\sigma)$ is always at most $4$, this gives a (tight) lower bound of $3/2$ on the strict competitive ratio of algorithms with sublinear advice. Using derandomization (see \cref{sec:derand}), the authors conclude that no randomized algorithm can achieve a strict competitive ratio smaller than $3/2$.

It is easy to see that $L(2,1)$-coloring on paths is a strictly $\lor$-repeatable problem. Thus, combining the lower bound for deterministic algorithms without advice from \cite{AL21} with \cref{lorrep}, we get that for every $\varepsilon>0$ and every (possibly randomized) algorithm reading $o(n)$ bits of advice, there exists an input $\sigma$ such that $\E[\ALG(\sigma)]\geq (1-\varepsilon) 6$. Since costs are integral, this reproves the lower bound of $3/2$ on the strict competitive ratio of algorithms with $o(n)$ bits of advice. We believe that the proof is simpler than the existing proof. Furthermore, it highlights the fact that obtaining a lower bound for deterministic algorithms without advice suffices in order to obtain a lower bound for randomized algorithms.

\paragraph{Necessity of assumptions in \cref{lorrep}}
Note that in order to successfully apply \cref{lorrep}, it must be the case that $t=\max_{\sigma\in I}\OPT(\sigma)$ is small compared to $k$ (where $I$ and $k$ are as in the statement of \cref{lorrep}). One might wonder if it is possible to get rid of this assumption. In particular, one could hope to show that if algorithms without advice for a $\lor$-repeatable problem can be no better than $c$-competitive, then algorithms with advice complexity $o(n)$ also cannot be (much) better than $c$-competitive. Unfortunately, it turns out that this is false in general. For instance, the online multi-coloring problem on a path network (without deletions) provides a counter-example (see \cite{DBLP:conf/isaac/ChanCYZZ06} for a definition of the problem). This problem is strictly $\lor$-repeatable. It is known that no algorithm without advice (even if we allow randomization) can achieve a competitive ratio better than $4/3$~\cite{CFLpath}. On the other hand, it has been shown that there exists a strictly $1$-competitive algorithm which uses $O(\log n)$ bits of advice~\cite{CFLApath}. The advice given to the algorithm is the number of colors used in an optimal solution. 

The structure of the $4/3$-lower bound for randomized algorithms is as follows: For an input length $n$, one defines two inputs, $\sigma_1^n, \sigma_2^n$, and considers the uniform distribution $p$ over these two inputs. The inputs $\sigma_1^n$ and $\sigma_2^n$ are such that $\OPT(\sigma_1)=n/2$ while $\OPT(\sigma_2)=n/4$. This implies that $\E_{\sigma\sim p}[\OPT(\sigma)]=1/2\cdot (n/2+n/4)=(3/8)n$. On the other hand, it is shown that for every deterministic algorithm $\DET$ (without advice), it holds that $\E_{\sigma\sim p}[\DET(\sigma)]\geq n/2$, which yields a lower bound of $4/3$. However, this lower bound does not carry over to algorithms with advice complexity $o(n)$ by simply repeating it $r$ times as in the proof of \cref{lorrep}. The problem is that the expected cost of $\OPT$ will tend to $n/2$ as $r$ tends to infinity. Indeed, the probability that $\OPT$ incurs a cost of $n/2$ in any given round is $1/2$, and if the cost of $\OPT$ in even a single round is $n/2$ then the total cost of $\OPT$ will also be $n/2$. This example illustrates why we need to work with $t=\max_{\sigma\in I}\OPT(\sigma)$ in \cref{lorrep}.

\paragraph*{Acknowledgment.} The author would like to thank Joan Boyar, Magnus Find, Lene Favrholdt, Adi Ros\'{e}n, and the anonymous reviewers for helpful comments on this work and its presentation.

\phantomsection

\addcontentsline{toc}{section}{References}

\bibliographystyle{plainurl}
\bibliography{refs}

\begin{thebibliography}{10}

\bibitem{STOCrbm}
Anna Adamaszek, Artur Czumaj, Matthias Englert, and Harald R{\"{a}}cke.
\newblock Almost tight bounds for reordering buffer management.
\newblock In {\em STOC}, pages 607--616. {ACM}, 2011.
\newblock \href {http://dx.doi.org/10.1145/1993636.1993717}
  {\path{doi:10.1145/1993636.1993717}}.

\bibitem{Arbm}
Anna Adamaszek, Marc~P. Renault, Adi Ros{\'{e}}n, and Rob van Stee.
\newblock Reordering buffer management with advice.
\newblock In {\em WAOA}, pages 132--143. Springer, 2013.
\newblock \href {http://dx.doi.org/10.1007/978-3-319-08001-7_12}
  {\path{doi:10.1007/978-3-319-08001-7_12}}.

\bibitem{Amakespan}
Susanne Albers and Matthias Hellwig.
\newblock Online makespan minimization with parallel schedules.
\newblock In {\em SWAT}, pages 13--25, 2014.
\newblock \href {http://dx.doi.org/10.1007/978-3-319-08404-6_2}
  {\path{doi:10.1007/978-3-319-08404-6_2}}.

\bibitem{COMBlu}
Susanne Albers, Bernhard von Stengel, and Ralph Werchner.
\newblock A combined {BIT} and {TIMESTAMP} algorithm for the list update
  problem.
\newblock {\em Inf. Process. Lett.}, 56(3):135--139, 1995.
\newblock \href {http://dx.doi.org/10.1016/0020-0190(95)00142-Y}
  {\path{doi:10.1016/0020-0190(95)00142-Y}}.

\bibitem{ambuhl}
Christoph Amb{\"u}hl.
\newblock {\em On the list update problem}.
\newblock PhD thesis, ETH Z{\"u}rich, 2002.
\newblock \href {http://dx.doi.org/10.3929/ethz-a-004394113}
  {\path{doi:10.3929/ethz-a-004394113}}.

\bibitem{DBLP:conf/wads/AngelopoulosDKR15}
Spyros Angelopoulos, Christoph D{\"{u}}rr, Shahin Kamali, Marc~P. Renault, and
  Adi Ros{\'{e}}n.
\newblock Online bin packing with advice of small size.
\newblock In {\em WADS}, volume 9214 of {\em Lecture Notes in Computer
  Science}, pages 40--53. Springer, 2015.
\newblock \href {http://dx.doi.org/10.1007/978-3-319-21840-3_4}
  {\path{doi:10.1007/978-3-319-21840-3_4}}.

\bibitem{KuhnInf}
Robert~J. Aumann.
\newblock Mixed and behavior strategies in infinite extensive games.
\newblock In {\em Advances in {G}ame {T}heory}, pages 627--650. Princeton Univ.
  Press, 1964.

\bibitem{FOCSrbm}
Noa Avigdor{-}Elgrabli and Yuval Rabani.
\newblock An optimal randomized online algorithm for reordering buffer
  management.
\newblock In {\em {FOCS}}, pages 1--10. {IEEE} Computer Society, 2013.
\newblock \href {http://dx.doi.org/10.1109/FOCS.2013.9}
  {\path{doi:10.1109/FOCS.2013.9}}.

\bibitem{AzarBM93}
Yossi Azar, Andrei~Z. Broder, and Mark~S. Manasse.
\newblock On-line choice of on-line algorithms.
\newblock In Vijaya Ramachandran, editor, {\em SODA}, pages 432--440.
  {ACM/SIAM}, 1993.
\newblock URL: \url{http://dl.acm.org/citation.cfm?id=313559.313847}.

\bibitem{Max-SAT}
Yossi Azar, Iftah Gamzu, and Ran Roth.
\newblock Submodular max-sat.
\newblock In {\em ESA}, volume 6942 of {\em Lecture Notes in Computer Science},
  pages 323--334. Springer, 2011.
\newblock \href {http://dx.doi.org/10.1007/978-3-642-23719-5_28}
  {\path{doi:10.1007/978-3-642-23719-5_28}}.

\bibitem{AzuOrg}
Kazuoki {Azuma}.
\newblock {Weighted sums of certain dependent random variables.}
\newblock {\em {Tohoku Math. J. (2)}}, 19:357--367, 1967.

\bibitem{DBLP:journals/algorithmica/BansalBGN14}
Nikhil Bansal, Niv Buchbinder, Anupam Gupta, and Joseph Naor.
\newblock A randomized o(log2 k)-competitive algorithm for metric bipartite
  matching.
\newblock {\em Algorithmica}, 68(2):390--403, 2014.
\newblock \href {http://dx.doi.org/10.1007/s00453-012-9676-9}
  {\path{doi:10.1007/s00453-012-9676-9}}.

\bibitem{polylogkserver}
Nikhil Bansal, Niv Buchbinder, Aleksander Madry, and Joseph Naor.
\newblock A polylogarithmic-competitive algorithm for the \emph{k}-server
  problem.
\newblock {\em J. {ACM}}, 62(5):40, 2015.
\newblock Preliminary version in FOCS'11.
\newblock \href {http://dx.doi.org/10.1145/2783434}
  {\path{doi:10.1145/2783434}}.

\bibitem{Bar-Noy}
Amotz Bar-Noy, Rajeev Motwani, and Joseph~(Seffi) Naor.
\newblock The greedy algorithm is optimal for on-line edge coloring.
\newblock {\em Inf. Process. Lett.}, 44(5):251--253, 1992.
\newblock \href {http://dx.doi.org/10.1016/0020-0190(92)90209-E}
  {\path{doi:10.1016/0020-0190(92)90209-E}}.

\bibitem{Bartalram}
Yair Bartal, B{\'{e}}la Bollob{\'{a}}s, and Manor Mendel.
\newblock Ramsey-type theorems for metric spaces with applications to online
  problems.
\newblock {\em J. Comput. Syst. Sci.}, 72(5):890--921, 2006.
\newblock \href {http://dx.doi.org/10.1016/j.jcss.2005.05.008}
  {\path{doi:10.1016/j.jcss.2005.05.008}}.

\bibitem{ColorBean}
Dwight~R. Bean.
\newblock Effective coloration.
\newblock {\em J. Symb. Log.}, 41(2):469--480, 1976.
\newblock \href {http://dx.doi.org/10.1017/S0022481200051549}
  {\path{doi:10.1017/S0022481200051549}}.

\bibitem{BBKTW}
Shai Ben{-}David, Allan Borodin, Richard~M. Karp, G{\'{a}}bor Tardos, and Avi
  Wigderson.
\newblock On the power of randomization in on-line algorithms.
\newblock {\em Algorithmica}, 11(1):2--14, 1994.
\newblock Preliminary version in STOC'90.
\newblock \href {http://dx.doi.org/10.1007/BF01294260}
  {\path{doi:10.1007/BF01294260}}.

\bibitem{Avbipartite}
Maria~Paola Bianchi, Hans-Joachim B{\"o}ckenhauer, Juraj Hromkovi{\v{c}}, and
  Lucia Keller.
\newblock Online coloring of bipartite graphs with and without advice.
\newblock {\em Algorithmica}, 70(1):92--111, 2014.
\newblock \href {http://dx.doi.org/10.1007/s00453-013-9819-7}
  {\path{doi:10.1007/s00453-013-9819-7}}.

\bibitem{AL21}
Maria~Paola Bianchi, Hans{-}Joachim B{\"{o}}ckenhauer, Juraj Hromkovi{\v{c}},
  Sacha Krug, and Bj{\"{o}}rn Steffen.
\newblock On the advice complexity of the online {$L(2, 1)$}-coloring problem
  on paths and cycles.
\newblock {\em Theor. Comput. Sci.}, 554:22--39, 2014.
\newblock \href {http://dx.doi.org/10.1016/j.tcs.2014.06.027}
  {\path{doi:10.1016/j.tcs.2014.06.027}}.

\bibitem{BBmts}
Avrim Blum and Carl Burch.
\newblock On-line learning and the metrical task system problem.
\newblock {\em Machine Learning}, 39(1):35--58, 2000.
\newblock Preliminary version in COLT'97.
\newblock \href {http://dx.doi.org/10.1023/A:1007621832648}
  {\path{doi:10.1023/A:1007621832648}}.

\bibitem{Assm}
Hans{-}Joachim B{\"{o}}ckenhauer, Richard Dobson, Sacha Krug, and Kathleen
  Steinh{\"{o}}fel.
\newblock On energy-efficient computations with advice.
\newblock In {\em {COCOON}}, pages 747--758. Springer, 2015.
\newblock \href {http://dx.doi.org/10.1007/978-3-319-21398-9_58}
  {\path{doi:10.1007/978-3-319-21398-9_58}}.

\bibitem{aML}
Hans-Joachim B\"{o}ckenhauer, Sascha Geulen, Dennis Komm, and Walter Unger.
\newblock Constructing randomized online algorithms from algorithms with
  advice.
\newblock {\em ETH-Z\"{u}rich}, 2015.
\newblock \href {http://dx.doi.org/10.3929/ethz-a-010556417}
  {\path{doi:10.3929/ethz-a-010556417}}.

\bibitem{BockenhauerHK14}
Hans{-}Joachim B{\"{o}}ckenhauer, Juraj Hromkovi{\v{c}}, and Dennis Komm.
\newblock A technique to obtain hardness results for randomized online
  algorithms - {A} survey.
\newblock In {\em Computing with New Resources - Essays Dedicated to Jozef
  Gruska on the Occasion of His 80th Birthday}, volume 8808 of {\em LNCS},
  pages 264--276. Springer, 2014.
\newblock URL: \url{http://dx.doi.org/10.1007/978-3-319-13350-8_20}, \href
  {http://dx.doi.org/10.1007/978-3-319-13350-8_20}
  {\path{doi:10.1007/978-3-319-13350-8_20}}.

\bibitem{sg}
Hans-Joachim B{\"o}ckenhauer, Juraj Hromkovi{\v{c}}, Dennis Komm, Sacha Krug,
  Jasmin Smula, and Andreas Sprock.
\newblock The string guessing problem as a method to prove lower bounds on the
  advice complexity.
\newblock {\em Theor. Comput. Sci.}, 554:95--108, 2014.
\newblock Preliminary version in COCOON'13.
\newblock \href {http://dx.doi.org/10.1016/j.tcs.2014.06.006}
  {\path{doi:10.1016/j.tcs.2014.06.006}}.

\bibitem{Ak-server}
Hans-Joachim B{\"o}ckenhauer, Dennis Komm, Rastislav Kr\'alovi\v{c}, and
  Richard Kr\'alovi\v{c}.
\newblock On the advice complexity of the k-server problem.
\newblock In {\em ICALP (1)}, pages 207--218, 2011.
\newblock \href {http://dx.doi.org/10.1007/978-3-642-22006-7_18}
  {\path{doi:10.1007/978-3-642-22006-7_18}}.

\bibitem{A1}
Hans-Joachim B{\"o}ckenhauer, Dennis Komm, Rastislav Kr\'alovi\v{c}, Richard
  Kr\'alovi\v{c}, and Tobias M{\"o}mke.
\newblock On the advice complexity of online problems.
\newblock In {\em ISAAC}, pages 331--340, 2009.
\newblock \href {http://dx.doi.org/10.1007/978-3-642-10631-6_35}
  {\path{doi:10.1007/978-3-642-10631-6_35}}.

\bibitem{BE98b}
Allan Borodin and Ran El-Yaniv.
\newblock {\em Online Computation and Competitive Analysis}.
\newblock Cambridge University Press, 1998.

\bibitem{BLSmetric}
Allan Borodin, Nathan Linial, and Michael~E. Saks.
\newblock An optimal on-line algorithm for metrical task system.
\newblock {\em J. {ACM}}, 39(4):745--763, 1992.
\newblock Preliminary version in STOC'87.
\newblock \href {http://dx.doi.org/10.1145/146585.146588}
  {\path{doi:10.1145/146585.146588}}.

\bibitem{AOC}
Joan Boyar, Lene~M. Favrholdt, Christian Kudahl, and Jesper~W. Mikkelsen.
\newblock Advice complexity for a class of online problems.
\newblock In {\em STACS}, volume~30 of {\em LIPIcs}, pages 116--129. Schloss
  Dagstuhl, 2015.
\newblock \href {http://dx.doi.org/10.4230/LIPIcs.STACS.2015.116}
  {\path{doi:10.4230/LIPIcs.STACS.2015.116}}.

\bibitem{AListupdate}
Joan Boyar, Shahin Kamali, Kim~S. Larsen, and Alejandro L{\'o}pez-Ortiz.
\newblock On the list update problem with advice.
\newblock In {\em LATA}, pages 210--221, 2014.
\newblock \href {http://dx.doi.org/10.1007/978-3-319-04921-2_17}
  {\path{doi:10.1007/978-3-319-04921-2_17}}.

\bibitem{Abp}
Joan Boyar, Shahin Kamali, Kim~S. Larsen, and Alejandro L{\'o}pez-Ortiz.
\newblock Online bin packing with advice.
\newblock In {\em STACS}, pages 174--186, 2014.
\newblock Full paper to appear in {\em Algorithmica}.
\newblock \href {http://dx.doi.org/10.4230/LIPIcs.STACS.2014.174}
  {\path{doi:10.4230/LIPIcs.STACS.2014.174}}.

\bibitem{Burch}
Carl Burch.
\newblock {\em Machine learning in metrical task systems and other on-line
  problems}.
\newblock PhD thesis, Carnegie Mellon University, 2000.
\newblock http://cburch.com/pub/thesis.ps.gz.

\bibitem{DBLP:conf/isaac/ChanCYZZ06}
Joseph~Wun{-}Tat Chan, Francis Y.~L. Chin, Deshi Ye, Yong Zhang, and Hong Zhu.
\newblock Frequency allocation problems for linear cellular networks.
\newblock In {\em {ISAAC}}, volume 4288, pages 61--70. Springer, 2006.
\newblock \href {http://dx.doi.org/10.1007/11940128_8}
  {\path{doi:10.1007/11940128_8}}.

\bibitem{DBLP:journals/mst/ChanZ09}
Timothy~M. Chan and Hamid Zarrabi{-}Zadeh.
\newblock A randomized algorithm for online unit clustering.
\newblock {\em Theory Comput. Syst.}, 45(3):486--496, 2009.
\newblock \href {http://dx.doi.org/10.1007/s00224-007-9085-7}
  {\path{doi:10.1007/s00224-007-9085-7}}.

\bibitem{DBLP:conf/soda/ChattopadhyayEEP12}
Arkadev Chattopadhyay, Jeff Edmonds, Faith Ellen, and Toniann Pitassi.
\newblock A little advice can be very helpful.
\newblock In {\em {SODA}}, pages 615--625. {SIAM}, 2012.

\bibitem{CFLpath}
Marie~G. Christ, Lene~M. Favrholdt, and Kim~S. Larsen.
\newblock Online multi-coloring on the path revisited.
\newblock {\em Acta Inf.}, 50(5-6):343--357, 2013.
\newblock \href {http://dx.doi.org/10.1007/s00236-013-0184-4}
  {\path{doi:10.1007/s00236-013-0184-4}}.

\bibitem{CFLApath}
Marie~G. Christ, Lene~M. Favrholdt, and Kim~S. Larsen.
\newblock Online multi-coloring with advice.
\newblock {\em Theor. Comput. Sci.}, 596:79--91, 2015.
\newblock \href {http://dx.doi.org/10.1016/j.tcs.2015.06.044}
  {\path{doi:10.1016/j.tcs.2015.06.044}}.

\bibitem{MR1165342}
Marek Chrobak and Lawrence~L. Larmore.
\newblock The server problem and on-line games.
\newblock In {\em On-line algorithms ({N}ew {B}runswick, {NJ}, 1991)}, volume~7
  of {\em DIMACS Ser. Discrete Math. Theoret. Comput. Sci.}, pages 11--64.
  Amer. Math. Soc., Providence, RI, 1992.

\bibitem{2serverlower}
Marek Chrobak, Lawrence~L. Larmore, Carsten Lund, and Nick Reingold.
\newblock A better lower bound on the competitive ratio of the randomized
  2-server problem.
\newblock {\em Inf. Process. Lett.}, 63(2):79--83, 1997.
\newblock \href {http://dx.doi.org/10.1016/S0020-0190(97)00099-9}
  {\path{doi:10.1016/S0020-0190(97)00099-9}}.

\bibitem{inftheory}
Thomas~M. Cover and Joy~A. Thomas.
\newblock {\em Elements of information theory {(2.} ed.)}.
\newblock Wiley, 2006.
\newblock \href {http://dx.doi.org/10.1002/047174882X}
  {\path{doi:10.1002/047174882X}}.

\bibitem{MoT}
Imre Csisz{\'{a}}r.
\newblock The method of types.
\newblock {\em {IEEE} Trans. Information Theory}, 44(6):2505--2523, 1998.
\newblock \href {http://dx.doi.org/10.1109/18.720546}
  {\path{doi:10.1109/18.720546}}.

\bibitem{DBLP:journals/siamcomp/DemaineHIP07}
Erik~D. Demaine, Dion Harmon, John Iacono, and Mihai Patrascu.
\newblock Dynamic optimality - almost.
\newblock {\em {SIAM} J. Comput.}, 37(1):240--251, 2007.
\newblock \href {http://dx.doi.org/10.1137/S0097539705447347}
  {\path{doi:10.1137/S0097539705447347}}.

\bibitem{dembo}
Amir Dembo and Ofer Zeitouni.
\newblock {\em Large deviations techniques and applications}, volume~38 of {\em
  Stochastic Modelling and Applied Probability}.
\newblock Springer Berlin Heidelberg, corrected printing of the 1998 second
  edition, 2009.
\newblock \href {http://dx.doi.org/http://dx.doi.org/10.1007/978-3-642-03311-7}
  {\path{doi:http://dx.doi.org/10.1007/978-3-642-03311-7}}.

\bibitem{A4}
Stefan Dobrev, Rastislav Kr\'alovi\v{c}, and Dana Pardubsk{\'a}.
\newblock Measuring the problem-relevant information in input.
\newblock {\em RAIRO - Theor.\ Inf.\ Appl.}, 43(3):585--613, 2009.
\newblock \href {http://dx.doi.org/10.1051/ita/2009012}
  {\path{doi:10.1051/ita/2009012}}.

\bibitem{DBLP:journals/corr/DurrKR16}
Christoph D{\"{u}}rr, Christian Konrad, and Marc~P. Renault.
\newblock On the power of advice and randomization for online bipartite
  matching.
\newblock {\em CoRR}, abs/1602.07154, 2016.

\bibitem{DBLP:journals/algorithmica/EggertKMS1}
Sebastian Eggert, Lasse Kliemann, Peter Munstermann, and Anand Srivastav.
\newblock Bipartite matching in the semi-streaming model.
\newblock {\em Algorithmica}, 63(1-2):490--508, 2012.
\newblock \href {http://dx.doi.org/10.1007/s00453-011-9556-8}
  {\path{doi:10.1007/s00453-011-9556-8}}.

\bibitem{DBLP:journals/tcs/EhmsenL13}
Martin~R. Ehmsen and Kim~S. Larsen.
\newblock Better bounds on online unit clustering.
\newblock {\em Theor. Comput. Sci.}, 500:1--24, 2013.
\newblock \href {http://dx.doi.org/10.1016/j.tcs.2013.07.008}
  {\path{doi:10.1016/j.tcs.2013.07.008}}.

\bibitem{A2}
Yuval Emek, Pierre Fraigniaud, Amos Korman, and Adi Ros{\'e}n.
\newblock Online computation with advice.
\newblock {\em Theor. Comput. Sci.}, 412(24):2642--2656, 2011.
\newblock \href {http://dx.doi.org/10.1016/j.tcs.2010.08.007}
  {\path{doi:10.1016/j.tcs.2010.08.007}}.

\bibitem{UnitClustering}
Leah Epstein and Rob van Stee.
\newblock On the online unit clustering problem.
\newblock {\em {ACM} Transactions on Algorithms}, 7(1):7, 2010.
\newblock \href {http://dx.doi.org/10.1145/1868237.1868245}
  {\path{doi:10.1145/1868237.1868245}}.

\bibitem{DBLP:journals/jcss/FakcharoenpholRT04}
Jittat Fakcharoenphol, Satish Rao, and Kunal Talwar.
\newblock A tight bound on approximating arbitrary metrics by tree metrics.
\newblock {\em J. Comput. Syst. Sci.}, 69(3):485--497, 2004.
\newblock \href {http://dx.doi.org/10.1016/j.jcss.2004.04.011}
  {\path{doi:10.1016/j.jcss.2004.04.011}}.

\bibitem{superazu}
Xiequan Fan, Ion Grama, and Quansheng Liu.
\newblock Hoeffding's inequality for supermartingales.
\newblock {\em Stochastic Processes and their Applications}, 122(10):3545 --
  3559, 2012.

\bibitem{FiatFKRRV98}
Amos Fiat, Dean~P. Foster, Howard~J. Karloff, Yuval Rabani, Yiftach Ravid, and
  Sundar Vishwanathan.
\newblock Competitive algorithms for layered graph traversal.
\newblock {\em {SIAM} J. Comput.}, 28(2):447--462, 1998.
\newblock \href {http://dx.doi.org/10.1137/S0097539795279943}
  {\path{doi:10.1137/S0097539795279943}}.

\bibitem{Paging91}
Amos Fiat, Richard~M. Karp, Michael Luby, Lyle~A. McGeoch, Daniel~Dominic
  Sleator, and Neal~E. Young.
\newblock Competitive paging algorithms.
\newblock {\em J. Algorithms}, 12(4):685--699, 1991.
\newblock \href {http://dx.doi.org/10.1016/0196-6774(91)90041-V}
  {\path{doi:10.1016/0196-6774(91)90041-V}}.

\bibitem{Agraphc}
Michal Fori{\v{s}}ek, Lucia Keller, and Monika Steinov{\'a}.
\newblock Advice complexity of online graph coloring.
\newblock Unpublished manuscript, 2012.

\bibitem{WMAfs}
Yoav Freund and Robert~E. Schapire.
\newblock A decision-theoretic generalization of on-line learning and an
  application to boosting.
\newblock {\em J. Comput. Syst. Sci.}, 55(1):119--139, 1997.
\newblock \href {http://dx.doi.org/10.1006/jcss.1997.1504}
  {\path{doi:10.1006/jcss.1997.1504}}.

\bibitem{DBLP:journals/tcs/FuchsHK05}
Bernhard Fuchs, Winfried Hochst{\"{a}}ttler, and Walter Kern.
\newblock Online matching on a line.
\newblock {\em Theor. Comput. Sci.}, 332(1-3):251--264, 2005.

\bibitem{DBLP:conf/soda/GuntherMMW13}
Elisabeth G{\"{u}}nther, Olaf Maurer, Nicole Megow, and Andreas Wiese.
\newblock A new approach to online scheduling: Approximating the optimal
  competitive ratio.
\newblock In {\em {SODA}}, pages 118--128. {SIAM}, 2013.
\newblock \href {http://dx.doi.org/10.1137/1.9781611973105.9}
  {\path{doi:10.1137/1.9781611973105.9}}.

\bibitem{DBLP:conf/icalp/GuptaL12}
Anupam Gupta and Kevin Lewi.
\newblock The online metric matching problem for doubling metrics.
\newblock In {\em {ICALP} {(1)}}, volume 7391, pages 424--435. Springer, 2012.
\newblock \href {http://dx.doi.org/10.1007/978-3-642-31594-7_36}
  {\path{doi:10.1007/978-3-642-31594-7_36}}.

\bibitem{Asushmita}
Sushmita Gupta, Shahin Kamali, and Alejandro L{\'o}pez-Ortiz.
\newblock On advice complexity of the k-server problem under sparse metrics.
\newblock In {\em SIROCCO}, pages 55--67, 2013.
\newblock \href {http://dx.doi.org/10.1007/978-3-319-03578-9_5}
  {\path{doi:10.1007/978-3-319-03578-9_5}}.

\bibitem{GyarfasL88}
Andr{\'{a}}s Gy{\'{a}}rf{\'{a}}s and Jen{\"{o}} Lehel.
\newblock On-line and first fit colorings of graphs.
\newblock {\em Journal of Graph Theory}, 12(2):217--227, 1988.
\newblock \href {http://dx.doi.org/10.1002/jgt.3190120212}
  {\path{doi:10.1002/jgt.3190120212}}.

\bibitem{MS}
Magn{\'u}s~M. Halld{\'o}rsson and Mario Szegedy.
\newblock Lower bounds for on-line graph coloring.
\newblock {\em Theor. Comput. Sci.}, 130(1):163--174, 1994.
\newblock \href {http://dx.doi.org/10.1016/0304-3975(94)90157-0}
  {\path{doi:10.1016/0304-3975(94)90157-0}}.

\bibitem{HoeffOrg}
Wassily Hoeffding.
\newblock Probability inequalities for sums of bounded random variables.
\newblock {\em Journal of the American statistical association},
  58(301):13--30, 1963.

\bibitem{A3}
Juraj Hromkovi{\v{c}}, Rastislav Kr\'alovi\v{c}, and Richard Kr\'alovi\v{c}.
\newblock Information complexity of online problems.
\newblock In {\em MFCS}, pages 24--36, 2010.
\newblock \href {http://dx.doi.org/10.1007/978-3-642-15155-2_3}
  {\path{doi:10.1007/978-3-642-15155-2_3}}.

\bibitem{ssm}
Sandy Irani, Sandeep~K. Shukla, and Rajesh Gupta.
\newblock Algorithms for power savings.
\newblock {\em {ACM} Transactions on Algorithms}, 3(4), 2007.
\newblock \href {http://dx.doi.org/10.1145/1290672.1290678}
  {\path{doi:10.1145/1290672.1290678}}.

\bibitem{DBLP:journals/jal/KalyanasundaramP93}
Bala Kalyanasundaram and Kirk Pruhs.
\newblock Online weighted matching.
\newblock {\em J. Algorithms}, 14(3):478--488, 1993.
\newblock \href {http://dx.doi.org/10.1006/jagm.1993.1026}
  {\path{doi:10.1006/jagm.1993.1026}}.

\bibitem{DBLP:conf/dagstuhl/KalyanasundaramP96}
Bala Kalyanasundaram and Kirk Pruhs.
\newblock On-line network optimization problems.
\newblock In Amos Fiat and Gerhard~J. Woeginger, editors, {\em Online
  Algorithms, The State of the Art}, volume 1442 of {\em LNCS}, pages 268--280.
  Springer, 1996.
\newblock \href {http://dx.doi.org/10.1007/BFb0029573}
  {\path{doi:10.1007/BFb0029573}}.

\bibitem{KMMO90}
Anna~R. Karlin, Mark~S. Manasse, Lyle~A. McGeoch, and Susan~S. Owicki.
\newblock Competitive randomized algorithms for non-uniform problems.
\newblock In {\em SODA}, pages 301--309. {SIAM}, 1990.
\newblock URL: \url{http://dl.acm.org/citation.cfm?id=320176.320216}.

\bibitem{CompRatio1}
Anna~R. Karlin, Mark~S. Manasse, Larry Rudolph, and Daniel~D. Sleator.
\newblock Competitive snoopy caching.
\newblock {\em Algorithmica}, 3:77--119, 1988.
\newblock \href {http://dx.doi.org/10.1007/BF01762111}
  {\path{doi:10.1007/BF01762111}}.

\bibitem{optbip}
Richard~M. Karp, Umesh~V. Vazirani, and Vijay~V. Vazirani.
\newblock An optimal algorithm for on-line bipartite matching.
\newblock In {\em STOC}, pages 352--358. {ACM}, 1990.
\newblock \href {http://dx.doi.org/10.1145/100216.100262}
  {\path{doi:10.1145/100216.100262}}.

\bibitem{DBLP:journals/tcs/KhullerMV94}
Samir Khuller, Stephen~G. Mitchell, and Vijay~V. Vazirani.
\newblock On-line algorithms for weighted bipartite matching and stable
  marriages.
\newblock {\em Theor. Comput. Sci.}, 127(2):255--267, 1994.
\newblock \href {http://dx.doi.org/10.1016/0304-3975(94)90042-6}
  {\path{doi:10.1016/0304-3975(94)90042-6}}.

\bibitem{DBLP:conf/dagstuhl/Kierstead96}
Hal~A. Kierstead.
\newblock Coloring graphs on-line.
\newblock In {\em Online Algorithms, The State of the Art}, pages 281--305,
  1996.
\newblock \href {http://dx.doi.org/10.1007/BFb0029574}
  {\path{doi:10.1007/BFb0029574}}.

\bibitem{DBLP:conf/stacs/KommKKM14}
Dennis Komm, Rastislav Kr\'alovi\v{c}, Richard Kr\'alovi\v{c}, and Tobias
  M{\"{o}}mke.
\newblock Randomized online algorithms with high probability guarantees.
\newblock In {\em STACS}, volume~25 of {\em LIPIcs}, pages 470--481. Schloss
  Dagstuhl, 2014.
\newblock \href {http://dx.doi.org/10.4230/LIPIcs.STACS.2014.470}
  {\path{doi:10.4230/LIPIcs.STACS.2014.470}}.

\bibitem{Asetcover}
Dennis Komm, Richard Kr\'alovi\v{c}, and Tobias M{\"o}mke.
\newblock On the advice complexity of the set cover problem.
\newblock In {\em CSR}, pages 241--252, 2012.
\newblock \href {http://dx.doi.org/10.1007/978-3-642-30642-6_23}
  {\path{doi:10.1007/978-3-642-30642-6_23}}.

\bibitem{DBLP:journals/csr/Koutsoupias09}
Elias Koutsoupias.
\newblock The k-server problem.
\newblock {\em Computer Science Review}, 3(2):105--118, 2009.
\newblock \href {http://dx.doi.org/10.1016/j.cosrev.2009.04.002}
  {\path{doi:10.1016/j.cosrev.2009.04.002}}.

\bibitem{WFA}
Elias Koutsoupias and Christos~H. Papadimitriou.
\newblock On the k-server conjecture.
\newblock {\em J. {ACM}}, 42(5):971--983, 1995.

\bibitem{Kuhn}
H.~W. Kuhn.
\newblock Extensive games and the problem of information.
\newblock In {\em Contributions to the theory of games, vol. 2}, Annals of
  Mathematics Studies, no. 28, pages 193--216. Princeton University Press,
  1953.

\bibitem{WMA}
Nick Littlestone and Manfred~K. Warmuth.
\newblock The weighted majority algorithm.
\newblock {\em Inf. Comput.}, 108(2):212--261, 1994.
\newblock Preliminary version in FOCS'89.
\newblock \href {http://dx.doi.org/10.1006/inco.1994.1009}
  {\path{doi:10.1006/inco.1994.1009}}.

\bibitem{LPforOP}
Carsten Lund and Nick Reingold.
\newblock Linear programs for randomized on-line algorithms.
\newblock In Daniel~Dominic Sleator, editor, {\em SODA}, pages 382--391.
  {ACM/SIAM}, 1994.

\bibitem{DBLP:conf/stoc/ManasseMS88}
Mark~S. Manasse, Lyle~A. McGeoch, and Daniel~Dominic Sleator.
\newblock Competitive algorithms for on-line problems.
\newblock In {\em STOC}, pages 322--333. {ACM}, 1988.
\newblock \href {http://dx.doi.org/10.1145/62212.62243}
  {\path{doi:10.1145/62212.62243}}.

\bibitem{maschler2013game}
M.~Maschler, S.~Zamir, and E.~Solan.
\newblock {\em Game Theory}.
\newblock Cambridge University Press, 2013.

\bibitem{DBLP:conf/soda/MeyersonNP06}
Adam Meyerson, Akash Nanavati, and Laura~J. Poplawski.
\newblock Randomized online algorithms for minimum metric bipartite matching.
\newblock In {\em {SODA}}, pages 954--959. {ACM} Press, 2006.
\newblock URL: \url{http://dl.acm.org/citation.cfm?id=1109557.1109662}.

\bibitem{Aedge}
Jesper~W. Mikkelsen.
\newblock Optimal online edge coloring of planar graphs with advice.
\newblock In {\em CIAC}, volume 9079 of {\em Lecture Notes in Computer
  Science}, pages 352--364. Springer, 2015.
\newblock \href {http://dx.doi.org/10.1007/978-3-319-18173-8_26}
  {\path{doi:10.1007/978-3-319-18173-8_26}}.

\bibitem{Abipfinal}
Shuichi Miyazaki.
\newblock On the advice complexity of online bipartite matching and online
  stable marriage.
\newblock {\em Inf. Process. Lett.}, 114(12):714--717, 2014.

\bibitem{momke}
Tobias M{\"{o}}mke.
\newblock A competitive ratio approximation scheme for the k-server problem in
  fixed finite metrics.
\newblock {\em CoRR}, abs/1303.2963, 2013.
\newblock URL: \url{http://arxiv.org/abs/1303.2963}.

\bibitem{DBLP:journals/ipl/Newman91}
Ilan Newman.
\newblock Private vs. common random bits in communication complexity.
\newblock {\em Inf. Process. Lett.}, 39(2):67--71, 1991.
\newblock \href {http://dx.doi.org/10.1016/0020-0190(91)90157-D}
  {\path{doi:10.1016/0020-0190(91)90157-D}}.

\bibitem{DBLP:conf/esa/RackeSW02}
Harald R{\"{a}}cke, Christian Sohler, and Matthias Westermann.
\newblock Online scheduling for sorting buffers.
\newblock In {\em ESA}, volume 2461 of {\em Lecture Notes in Computer Science},
  pages 820--832. Springer, 2002.
\newblock \href {http://dx.doi.org/10.1007/3-540-45749-6_71}
  {\path{doi:10.1007/3-540-45749-6_71}}.

\bibitem{Ak-serverRR}
Marc~P. Renault and Adi Ros{\'{e}}n.
\newblock On online algorithms with advice for the k-server problem.
\newblock {\em Theory Comput. Syst.}, 56(1):3--21, 2015.
\newblock \href {http://dx.doi.org/10.1007/s00224-012-9434-z}
  {\path{doi:10.1007/s00224-012-9434-z}}.

\bibitem{CompRatio2}
Daniel~D. Sleator and Robert~E. Tarjan.
\newblock Amortized efficiency of list update and paging rules.
\newblock {\em Commun. ACM}, 28(2):202--208, 1985.
\newblock \href {http://dx.doi.org/10.1145/2786.2793}
  {\path{doi:10.1145/2786.2793}}.

\bibitem{DBLP:journals/jacm/SleatorT85}
Daniel~Dominic Sleator and Robert~Endre Tarjan.
\newblock Self-adjusting binary search trees.
\newblock {\em J. {ACM}}, 32(3):652--686, 1985.
\newblock \href {http://dx.doi.org/10.1145/3828.3835}
  {\path{doi:10.1145/3828.3835}}.

\bibitem{Smula}
Jasmin Smula.
\newblock {\em Information Content of Online Problems, Advice versus
  Determinism and Randomization}.
\newblock PhD thesis, ETH Z{\"u}rich, 2015.
\newblock \href {http://dx.doi.org/10.3929/ethz-a-010497710}
  {\path{doi:10.3929/ethz-a-010497710}}.

\bibitem{Teia93}
Boris Teia.
\newblock A lower bound for randomized list update algorithms.
\newblock {\em Inf. Process. Lett.}, 47(1):5--9, 1993.
\newblock \href {http://dx.doi.org/10.1016/0020-0190(93)90150-8}
  {\path{doi:10.1016/0020-0190(93)90150-8}}.

\bibitem{DBLP:journals/siamcomp/Wilber89}
Robert~E. Wilber.
\newblock Lower bounds for accessing binary search trees with rotations.
\newblock {\em {SIAM} J. Comput.}, 18(1):56--67, 1989.
\newblock \href {http://dx.doi.org/10.1137/0218004}
  {\path{doi:10.1137/0218004}}.

\bibitem{Yao}
Andrew Chi-Chih Yao.
\newblock Probabilistic computations: Toward a unified measure of complexity
  (extended abstract).
\newblock In {\em FOCS}, pages 222--227, 1977.
\newblock \href {http://dx.doi.org/10.1109/SFCS.1977.24}
  {\path{doi:10.1109/SFCS.1977.24}}.

\end{thebibliography}

\newpage
\appendix
\section{Derandomization using advice}
In this appendix, we slightly extend the derandomization result due to B{\"o}ckenhauer et al.\ \cite{Ak-server} (see also the survey \cite{BockenhauerHK14}). In particular, we show that it is possible to derandomize algorithms which uses both advice and randomization. Also, we show how to modify the result for minimization problems so that it can be applied to maximization problems. Finally, we provide lower bounds showing that these derandomization results are essentially best possible.

\label{sec:derand}
\begin{theorem}
\label{thm:derandmin}
Let $\P$ be an online minimization problem and let $R$ be a randomized algorithm with advice complexity $b(n)$. Suppose that there exists a function $I:\mathbb{N}\rightarrow \mathbb{N}$ such that the number of valid $\P$-inputs of length $n$ is at most $I(n)$. Then for every $\varepsilon>0$, there exists a deterministic algorithm with advice, $\ALG$, such that $\ALG(\sigma)\leq(1+\varepsilon)\E[R(\sigma)]$ for every input $\sigma$ and such that $\ALG$ reads at most $$b(n)+O(\log n + \log \log I(n)+\log\varepsilon^{-1})$$ bits of advice on inputs of length $n$. In particular, if $I(n)=\dri$ and $\varepsilon>0$ is a small but fixed constant, then $\ALG$ reads $b(n)+O(\log n)$ bits of advice.
\end{theorem}
\begin{proof}
Let $I_n$ be the set of all $\P$-inputs of length $n$. By assumption, $\ab{I_n}\leq I(n)$. Recall that a randomized algorithm with advice complexity $b$ is a probability distribution over deterministic algorithms with advice complexity (at most) $b$. Let $t$ be an integer and let the random variables $R_1,\ldots , R_t$ be $t$ deterministic algorithms selected independently at random according to the probability distribution specified by the randomized algorithm $R$. Then, for every fixed input $\sigma\in I_n$ such that $\E[R(\sigma)]>0$ and every $1\leq i\leq t$, Markov's inequality implies that
\begin{align}
\label{derandmark}
\Pr\Big[R_i(\sigma)> (1+\varepsilon)\E[R(\sigma)]\Big]\leq \frac{1}{1+\varepsilon}.
\end{align}
Recall that costs are by definition non-negative (\cref{def}). Therefore, if $\E[R(\sigma)]=0$ for some $\sigma\in I_n$, then $R_i(\sigma)=0$ for $1\leq i\leq t$. Thus, for every $\sigma\in I_n$ where $\E[R(\sigma)]=0$, the inequality (\ref{derandmark}) holds for the trivial reason that the probability on the left-hand side is zero.

Since $R_1,\ldots , R_t$ are independent random variables, inequality (\ref{derandmark}) implies that for every fixed input $\sigma\in I_n$, the probability of all $t$ algorithms failing to achieve the desired guarantee is at most
\begin{align*}
\Pr\Big[\min_{1\leq i\leq t} \{R_i(\sigma)\}> (1+\varepsilon)\E[R(\sigma)]\Big]\leq \left(\frac{1}{1+\varepsilon}\right)^t.
\end{align*}
An application of the union bound gives
\begin{equation}
\Pr\Big[\exists\sigma\in I_n : \min_{1\leq i\leq t} \{R_i(\sigma)\}> (1+\varepsilon)\E[R(\sigma)]\Big]\leq I(n)\left(\frac{1}{1+\varepsilon}\right)^t.
\label{existssigma}
\end{equation}
Note that
\begin{align*}
I(n)\left(\frac{1}{1+\varepsilon}\right)^t<1 \Leftrightarrow \frac{\log (I(n))}{\log (1+\varepsilon)}< t.
\end{align*}
Choose $t=t_{n,\varepsilon}$ to be the smallest integer such that $\frac{\log (I(n))}{\log (1+\varepsilon)}<t_{n,\varepsilon}$. Then
\begin{align*}
\Pr\Big[\forall\sigma\in I_n : \min_{1\leq i\leq t_{n,\varepsilon}} \{R_i(\sigma)\}\leq (1+\varepsilon)\E[R(\sigma)]\Big]\geq 1-I(n)\left(\frac{1}{1+\varepsilon}\right)^{t_{n,\varepsilon}}>0.
\end{align*}
This shows that there exists $t_{n,\varepsilon}$ deterministic algorithms such that for every input $\sigma\in I_n$, at least one of these $t_{n,\varepsilon}$ algorithms incurs a cost of at most $(1+\varepsilon)\E[R(\sigma)]$. By assumption, the advice complexity of each of these $t_{n,\varepsilon}$ algorithms is at most $b(n)$.

We are now ready to define $\ALG$. The oracle will write the length, $n$, of the input onto the advice tape using $O(\log n)$ bits. Knowing $n$, $\ALG$ will be able to\footnote{Exactly how this is done depends on the model of online algorithms being used. If an online algorithm is just a strategy for a request-answer game without computational constraints (such as in the model of Ben-David et al.~\cite{BBKTW}, see also \cref{app:rando}), then we do not need to specify how $\ALG$ determines these algorithms (it suffice to know that they exist). In most other models, such as the one used in \cite{Ak-server}, it will be possible for $\ALG$ to use exhaustive search.} determine $t_{n,\varepsilon}$ algorithms $\ALG_1,\ldots , \ALG_{t_{n,\varepsilon}}$ such that for every input $\sigma\in I_n$, it holds that $\min_{i}\{\ALG_i(\sigma)\}\leq (1+\varepsilon)\E[R(\sigma)]$. The oracle will encode the index, $i$, of the best of these $t_{n,\varepsilon}$ algorithms for the current input $\sigma$. This requires $\lceil \log t_{n,\varepsilon}\rceil$ bits. Furthermore, the oracle writes the advice read by $\ALG_i$ (when given $\sigma$ as input) onto the advice tape of $\ALG$. The algorithm $\ALG$ is then able to simulate $\ALG_i$ and, thus, $\ALG(\sigma)=\ALG_i(\sigma)\leq(1+\varepsilon)\E[R(\sigma)]$. The number of advice bits read by $\ALG$ is at most
\begin{align*}
b(n)+\left\lceil \log\left(\frac{\log (I(n))}{\log (1+\varepsilon)}+1\right)\right\rceil+O(\log n)&=b(n)+O(\log n+ \log \log I(n)+\log \varepsilon^{-1}).
\end{align*}
\end{proof}
Note that $\varepsilon$ could be non-constant. Assuming $I(n)=\dri$, we can take, for example, $\varepsilon=1/n$ and get an algorithm $\ALG$ which reads $O(\log n)$ bits of advice while $\ALG(\sigma)\leq (1+1/n)\E[R(\sigma)]$.

We will now show how to obtain a guarantee on the competitive ratio using \cref{thm:derandmin}.
\begin{corollary}
Let $c\geq 1$ be a constant, let $\P$ be a minimization problem and let $R$ be a randomized $c$-competitive algorithm with advice complexity $b(n)$. Suppose that there exists a function $I:\mathbb{N}\rightarrow \mathbb{N}$ such that the number of valid $\P$-inputs of length $n$ is at most $I(n)$. If $I(n)=\dri$ then for every constant $\varepsilon>0$, there exists a deterministic $(c+\varepsilon)$-competitive algorithm with advice complexity $b(n)+O(\log n)$.
\end{corollary}
\begin{proof}
Choose $\varepsilon'=\varepsilon/c$ and apply \cref{thm:derandmin} to $R$. This gives a deterministic algorithm $\ALG$ with the desired advice complexity such that $\ALG(\sigma)\leq (1+\varepsilon')\E[R(\sigma)]$ for every input $\sigma$. Using that $R$ is $c$-competitive, we get that for every input $\sigma$,
\begin{align*}
\ALG(\sigma)\leq (1+\varepsilon')\E[R(\sigma)]\leq (1+\varepsilon')(c\cdot\OPT(\sigma)+\alpha)\leq (c+\varepsilon)\OPT(\sigma)+(1+\varepsilon)\alpha.
\end{align*}
This proves that $\ALG$ is $(c+\varepsilon)$-competitive.
\end{proof}

We will now prove a derandomization result for maximization algorithms. However, there are some technical difficulties involved in this. For a minimization problem, we always know that the cost of a solution is at least zero. For a maximization problem, $\OPT(\sigma)$ could grow arbitrarily fast (as a function of the length of the input) while the profit of a randomized algorithm could remain small. We resolve this by assuming that the strict competitive ratio of the algorithm is bounded by some function $c(n)$, and show that if $c(n)$ is polynomial in $n$, then we get the same conclusion as for minimization problems. In \cref{app:lowerstrict}, we show that this assumption cannot be removed (without replacing it with something else).

\begin{theorem}
\label{derandmax}
Let $\P$ be an online maximization problem and let $R$ be a randomized algorithm with advice complexity $b(n)$. Suppose that there exists a function $I:\mathbb{N}\rightarrow \mathbb{N}$ such that the number of valid $\P$-inputs of length $n$ is at most $I(n)$. Furthermore, suppose that $R$ is strictly $c(n)$-competitive, that is, $\OPT(\sigma)\leq c(n)\cdot E[R(\sigma)]$ for every input $\sigma$ of length at most $n$.

Then for every $0<\varepsilon<1$, there exists a deterministic algorithm with advice, $\ALG$, such that $\ALG(\sigma)\geq (1-\varepsilon)\E[R(\sigma)]$ for every input $\sigma$ and such that $\ALG$ reads at most $$b(n)+O(\log n + \log \log I(n)+\log c(n) + \log\varepsilon^{-1})$$ bits of advice on inputs of length $n$. In particular, if $I(n)=\dri$, $c(n)=n^{O(1)}$, and $\varepsilon$ is a small but fixed constant, then $\ALG$ reads $b(n)+O(\log n)$ bits of advice.
\end{theorem}
\begin{proof}
Consider the set $I_n$ of all $\P$-inputs of length $n$. By assumption, $\ab{I_n}\leq I(n)$. Recall that a randomized algorithm with advice complexity $b$ is a probability distribution over deterministic algorithms with advice complexity (at most) $b$. Let $t$ be an integer and let the random variables $R_1,\ldots , R_t$ be $t$ deterministic algorithms selected independently at random according to the probability distribution specified by the randomized algorithm $R$. Then, for every fixed input $\sigma\in I_n$ such that $\E[R(\sigma)]>0$ and every $1\leq i\leq t$, Markov's inequality (applied to the positive random variable $\OPT(\sigma)-R_i(\sigma)$) implies that
\begin{align}
\notag \Pr\Big[R_i(\sigma)< (1-\varepsilon)\E[R(\sigma)]\Big]&\leq \frac{\OPT(\sigma)-\E[R(\sigma)]}{\OPT(\sigma)-(1-\varepsilon)\E[R(\sigma)]}\\
\notag &=\frac{\OPT(\sigma)/\E[R(\sigma)]-1}{\OPT(\sigma)/\E[R(\sigma)]-(1-\varepsilon)}\\
\label{maxderandmark}&\leq\frac{c(n)-1}{c(n)-(1-\varepsilon)}.
\end{align}
The same upper bound trivially holds if $E[R(\sigma)]=0$, since then the probability on the left-hand side of (\ref{maxderandmark}) is zero (because profits are by definition non-negative). Since $R_1,\ldots , R_t$ are independent random variables, inequality (\ref{maxderandmark}) implies that for every fixed input $\sigma\in I_n$, the probability of all $t$ algorithms failing to achieve the desired performance guarantee is at most
\begin{align*}
\Pr\Big[\min_{1\leq i\leq t} \{R_i(\sigma)\}< (1-\varepsilon)\E[R(\sigma)]\Big]\leq \left(\frac{c(n)-1}{c(n)-(1-\varepsilon)}\right)^{t}.
\end{align*}
An application of the union bound gives
\begin{equation}
\Pr\Big[\exists\sigma\in I_n : \min_{1\leq i\leq t} \{R_i(\sigma)\}< (1-\varepsilon)\E[R(\sigma)]\Big]\leq I(n) \left(\frac{c(n)-1}{c(n)-(1-\varepsilon)}\right)^{t}.
\label{existssigmamax}
\end{equation}
Note that
\begin{align*}
I(n)\cdot \left(\frac{c(n)-1}{c(n)-(1-\varepsilon)}\right)^{t}<1 \Leftrightarrow \frac{\log (I(n))}{\log \left(\frac{c(n)-(1-\varepsilon)}{c(n)-1}\right)}< t.
\end{align*}
Using that $\log(1+x)^{-1}\leq 1/x$ for $0<x\leq 1$ and $\log(1+x)^{-1}\leq 1$ for $x>1$, we get that
\begin{align*}
\log \left(\frac{c(n)-(1-\varepsilon)}{c(n)-1}\right)^{-1}&=\log\left(1+\frac{\varepsilon}{c(n)-1}\right)^{-1}\\
&\leq \frac{c(n)-1}{\varepsilon}+1.
\end{align*}
Choose $t=t_{n,\varepsilon}$ to be the smallest integer such that ${\log I(n)}\cdot \left(\frac{c(n)-1}{\varepsilon}+1\right)<t_{n,\varepsilon}$. Then
\begin{align*}
\Pr\Big[\forall\sigma\in I_n : \min_{1\leq i\leq t_{n,\varepsilon}} \{R_i(\sigma)\}\geq (1-\varepsilon)\E[R(\sigma)]\Big]>0.
\end{align*}
This shows that there exist $t_{n,\varepsilon}$ deterministic algorithms such that for every input $\sigma\in I_n$, at least one of these $t_{n,\varepsilon}$ algorithms obtains a profit of at least $(1-\varepsilon)\E[R(\sigma)]$. By assumption, the advice complexity of each of these $t_{n,\varepsilon}$ algorithms is at most $b(n)$.

We are now ready to define $\ALG$. The oracle will write the length $n$ of the input onto the advice tape using $O(\log n)$ bits. Knowing $n$, $\ALG$ will be able to determine $t_{n,\varepsilon}$ algorithms $\ALG_1,\ldots , \ALG_t$ such that $\min_{i}\{\ALG_i(\sigma)\}\geq (1-\varepsilon)\E[R(\sigma)]$ for every $\sigma\in I_n$. The oracle will encode the index, $i$, of the best of the $t_{n,\varepsilon}$ algorithms for the current input $\sigma$. This requires $\lceil \log t_{n,\varepsilon}\rceil$ bits. Furthermore, the oracle writes the advice read by $\ALG_i$ (when given $\sigma$ as input) onto the advice tape of $\ALG$. The algorithm $\ALG$ is then able to simulate $\ALG_i$ and, thus, $\ALG(\sigma)=\ALG_i(\sigma)\geq(1-\varepsilon)\E[R(\sigma)]$. The number of advice bits read by $\ALG$ is at most
\begin{align*}
b(n)+&\log\left({\log I(n)}\cdot\left({\frac{c(n)-1}{\varepsilon}}+1\right)+1\right)+O(\log n)\\
&=b(n)+O(\log n+ \log \log I(n)+\log c(n)+\log \varepsilon^{-1}).
\end{align*}
\end{proof}

\begin{corollary}
\label{derandmax:cor}
Let $c\geq 1$ be a constant. Furthermore, let $\P$ be an online maximization problem and let $R$ be a randomized algorithm with advice complexity $b(n)$. Suppose that there exists a function $I:\mathbb{N}\rightarrow \mathbb{N}$ such that the number of valid $\P$-inputs of length $n$ is at most $I(n)$. If $I(n)=\dri$, then for every $\varepsilon>0$, if there exists a $c$-competitive randomized algorithm $R$ with advice complexity $b(n)$, then there exists a $(c+\varepsilon)$-competitive deterministic algorithm with advice complexity $b(n)+O(\log n)$.
\end{corollary}
\begin{proof}
By assumption, there exists an additive constant $\alpha$ such that $\OPT(\sigma)\leq c\E[R(\sigma)]+\alpha$ for every input $\sigma$. Let $I_{2\alpha}$ those $\P$-inputs for which $2\alpha\leq \OPT(\sigma)$. Note that if $\sigma\in I_{2\alpha}$, then $c\E[R(\sigma)]+\alpha \geq \OPT(\sigma)\geq 2\alpha$. Thus, $c\E[R(\sigma)]\geq \alpha$, and hence $2c\E[R(\sigma)]\geq \OPT(\sigma)$ for every $\sigma\in I_{2\alpha}$. In particular, $R$ is strictly $2c$-competitive when we restrict to inputs from $I_{2\alpha}$. Choose $\varepsilon'=\varepsilon/(c+\varepsilon)$. By applying \cref{derandmax} to the algorithm $R$ restricted to inputs from $I_{2\alpha}$, we obtain a deterministic algorithm $\ALG'$ with advice complexity $b(n)+O(\log n)$ such that $\ALG'(\sigma)\geq (1-\varepsilon)\E[R(\sigma)]$ for every $\sigma\in I_{2\alpha}$.

Our deterministic $(c+\varepsilon)$-competitive algorithm $\ALG$ works as follows: A single bit of advice is used to indicate if $\sigma\in I_{2\alpha}$ or not. If $\sigma\notin I_{2\alpha}$, then $\ALG$ computes an arbitrary valid output. If $\sigma\in I_{2\alpha}$, the algorithm simulates $\ALG'$ on $\sigma$. By the choice of $\varepsilon'$, we get that for every $\sigma\in I_{2\alpha}$,
\begin{align*}
\OPT(\sigma)\leq c\E[R(\sigma)]+\alpha\leq \frac{c}{1-\varepsilon'} \ALG'(\sigma)+\alpha=(c+\varepsilon)\ALG'(\sigma)+\alpha=(c+\varepsilon)\ALG(\sigma)+\alpha.
\end{align*}
Since $\OPT(\sigma)\leq 2\alpha$ for every $\sigma\notin I_{2\alpha}$, it follows that $\ALG$ is $(c+\varepsilon)$-competitive with an additive constant of $2\alpha$.
\end{proof}

\subsection{Lower bounds for derandomization using advice}

In this section, we show that \cref{thm:derandmin} is essentially tight by providing (for any choice of $I(n)$) an online problem where $\Omega(\log \log I(n))$ bits of advice are necessary for a deterministic online algorithm to achieve a competitive ratio anywhere near that of the best randomized algorithm.
\begin{theorem}
Let $I(n):\mathbb{N}\rightarrow\mathbb{N}$ be an increasing function. There exists a minimization problem \P with at most $I(n)$ inputs of length $n$ such that a randomized algorithm without advice can be $1$-competitive, while a deterministic algorithm with advice needs to read at least $\Omega(\log \log I(n))$ bits of advice to achieve a constant competitive ratio.
\end{theorem}

\begin{proof}
Define $f(n)=\lfloor\sqrt{\log I(n)}\rfloor$ and $u(n)=2^{f(n)}$. Furthermore, for each $n\geq 2$, let $M_n$ be a set containing $u(n)$ elements. The problem \P is defined as follows: In the first $n-2$ rounds, a dummy request appears which requires no response. In round $n-1$, a special request, $?$, is revealed and the algorithm must respond by choosing an element $x$ from $M_n$ (note that the algorithm can deduce $n$ from the number of dummy requests). In round $n$, a set $M'\subset M_n$ such that $\ab{M'}=f(n)$ is revealed (no response is required from the algorithm). If $x\in M'$, the algorithm incurs a cost of $f(n)$. Otherwise, the algorithm incurs a cost of $1$. Formally, the input to \P is a sequence $(d_1,\ldots , d_{n-2}, ?, M')$ where $d_i$ denotes a dummy request. Note that since $f(n)<u(n)$, it follows that $\OPT(\sigma)=1$ for any input. Also, the number of inputs of length $n$ is at most $\binom{u(n)}{f(n)}\leq u(n)^{f(n)}=2^{f(n)^2}\leq 2^{\log I(n)}=I(n)$.

Consider the randomized algorithm $R$ which chooses $x$ uniformly at random from $M_n$ in round $n-1$. Since the adversary is oblivious, the expected cost incurred by $R$ on an input $\sigma$ of length $n$ is $$\E[R(\sigma)]\leq f(n)\cdot \frac{f(n)}{u(n)}+1\cdot\left(1-\frac{f(n)}{u(n)}\right)\leq \frac{f(n)^2}{u(n)} + 1\leq 3. $$
It follows that $\E[R(\sigma)]\leq \OPT(\sigma)+2$ and hence $R$ is $1$-competitive.

Let $\ALG$ be a deterministic algorithm reading at most $b(n)$ bits of advice on inputs of length $n$. Since the first $n-1$ request are identical in all inputs, the element $x$ chosen in round $n-1$ can only depend on $n$ and the advice. Thus, $\ALG$ can choose at most $2^{b(n)}$ different elements from $M_n$ on inputs of length $n$. If for some $n$ we have that $2^{b(n)}\leq f(n)$, then an adversary can choose $M'$ to contain all of the elements that could possibly be outputted by $\ALG$. Thus, $\ALG$ will incur a cost of $f(n)$. In particular, if $\ALG$ is to achieve a constant competitive ratio, then there must be an $N$ such that for all $n\geq N$, we have that $2^{b(n)}>f(n)$ and, hence, $$b(n)>\log (f(n))=\log\left( \lfloor \sqrt{\log I(n)}\rfloor\right) =\Omega(\log \log I(n)).$$
\end{proof}

For maximization problems, it is natural to ask if the $O(\log c(n))$-term in \cref{derandmax} is necessary or just an artifact of our proof-strategy. We show in the following example that it is indeed necessary, by designing an online problem and a strictly $c$-competitive randomized algorithm such that any deterministic algorithm reading $o(\log c(n))$ bits of advice will, on some inputs, obtain a profit much smaller than the randomized algorithm.

\begin{theorem}
\label{app:lowerstrict}
Let $c:\mathbb{N}\rightarrow \mathbb{N}$ be a function. There exists an online maximization problem $\P$ with $I(n)$ inputs of length $n$ and a randomized $\P$-algorithm $R$ without advice such that: 
\begin{enumerate}[(i)]
\item The algorithm $R$ is strictly $c(n)$-competitive and $\E[R(\sigma)]=n$ for every input $\sigma$ of length $n$.
\item If $\ALG$ is a deterministic algorithm reading $O(\log n+ \log \log I(n))$ bits of advice, then there exists an $n_0$ such that for every $n\geq n_0$, there is an input $\sigma$ of length $n$ such that $\ALG(\sigma)=0$.
\item If $\ALG$ is a deterministic algorithm such that $\ALG(\sigma)\geq 1$ for every input $\sigma$, then $\ALG$ must read at least $\Omega(\log c(n))$ bits of advice.
\end{enumerate}
\end{theorem}

\begin{proof}
Let $\P$ be the following modification of the string guessing problem with unknown history (cf. \cite{sg}): The problem $\P$ consists of $n$ rounds. The number of rounds is revealed to the online algorithm at the beginning. In round $1\leq i\leq n$, the algorithm has to guess a character from the alphabet $\{1,2,\ldots , c(n)\}$. Only after all $n$ rounds, the correct characters are revealed (thus, the algorithm does not learn the correct character right after guessing). For each character correctly guessed by the algorithm, it obtains a profit of $c(n)$. For each wrong guess, the algorithm obtains a profit of $0$.

Note that $\OPT(\sigma)=n\cdot c(n)$ and that $I(n)=c(n)^n$. Let $R$ be the algorithm which selects a character uniformly at random in each round. The probability of $R$ guessing right in round $i$ is $1/c(n)$. Since the profit of a correct guess is $c(n)$, this means that the expected profit of $R$ in each round is $1$. By linearity of expectation, $\E[R(\sigma)]=n$ for every input $\sigma$ of length $n$. It follows that $R$ is strictly $c(n)$-competitive.

Let $\ALG$ be a deterministic algorithm with advice complexity $b(n)$. Suppose that for some $n$, it holds that $b(n)\leq \log c(n)-1$. This means that $\ALG$ can produce at most $2^{b(n)}\leq 2^{\log( c(n))-1}=c(n)/2<c(n)$ different outputs on inputs of length $n$. In particular, for each $1\leq i\leq n$, there exists a character $x_i\in\{1,2,\ldots , c(n)\}$ such that none of the outputs that $\ALG$ can produce has $x_i$ as the answer in round $i$. It follows that there exists an input string $x=x_1\ldots x_n$ of length $n$ such that $\ALG(x)=0$. This proves that (iii) holds. In order to see that (ii) holds, note that $\log \log I(n)=\log \log (c(n)^n)=O(\log n + \log \log c(n))$. In particular, if $b(n)=O(\log n+\log \log I(n))$, then $b(n)=o(\log c(n))$ which shows that (ii) holds.  

\end{proof}

\section{Concentration inequalities}
\refstepcounter{subsection}
\subsection*{\thesubsection \quad Chernoff bounds}
\begin{theorem}[{\cite[Lemma II.2]{MoT}}]
Let $0<p<1$ and let $X_1,X_2,\ldots $ be a sequence of independent random variables such that for each $i$, $\Pr[X_i=1]=p$ and $\Pr[X_i=0]=1-p$. Then for every $\varepsilon>0$,
\begin{equation}
\frac{2^{-K_{p}(p+\varepsilon)\cdot n}}{n+1}\leq \Pr\left[ \sum_{i=1}^n X_i \geq (p+\varepsilon)n\right]\leq 2^{-K_{p}(p+\varepsilon)\cdot n}.
\label{chernoff}
\end{equation}
\label{chernoffthm}
\end{theorem}
The upper bound in (\ref{chernoff}) is a variant of the Chernoff bound. The lower bound in (\ref{chernoff}) shows that this variant of the Chernoff bound is essentially tight. \cref{chernoffthm} is well-known and can, for example, be easily derived as a special case of Lemma II.2 in \cite{MoT}. 
\refstepcounter{subsection}
\subsection*{\thesubsection \quad The Azuma-Hoeffding inequality}
\label{app:sec:az}
The Azuma-Hoeffding inequality \cite{AzuOrg, HoeffOrg} gives a bound similar to (\ref{chernoff}) for bounded-difference martingales.

\begin{theorem}[\cite{dembo, superazu}]
Let $Z_0,Z_1,\ldots Z_n$ be a (real-valued) supermartingale. Assume that there exist constants $d,\sigma\geq 0$ such that, for $0\leq i<n$,
\begin{align*}
Z_{i+1}-Z_i&\leq d,\\
\E[(Z_{i+1}-Z_i)^2\vert Z_{1},\ldots , Z_{i}] &\leq \sigma^2.
\end{align*}
Then, for every $r\geq 0$,
\begin{equation}
\label{azumai}
\Pr[Z_n-Z_0\geq t n]\leq \exp_2\left(-K_{\mbox{\Large $\frac{\gamma}{1+\gamma}$}}\left(\frac{\alpha+\gamma}{1+\gamma}\right)n\right),
\end{equation}
where $\gamma=\dfrac{\sigma^2}{d^2}$ and $\alpha=\dfrac{t}{d}$ (and $\exp_2(x)=2^x$).
\label{azuma}
\end{theorem} 
\cref{azuma} is almost identical to Corollary 2.4.7 in \cite{dembo}, except that the corollary in \cite{dembo} is only stated for martingales, not supermartingales. However, one may verify that the same proof applies to supermartingales. Alternatively, \cref{azuma} can be derived from Theorem 2.1 in \cite{superazu}, which is stated and proved for supermartingales. Indeed, under the assumptions of \cref{azuma}, it follows immediately from \cite{superazu} that
\begin{equation}
\Pr[Z_n-Z_0\geq tn]\leq\left(\left(\frac{\gamma}{\alpha+\gamma}\right)^{n(\alpha+\gamma)}\left(\frac{1}{1-\alpha}\right)^{n(1-\alpha)}\right)^{\frac{1}{1+\gamma}}.\label{superazui}
\end{equation}
where $\gamma=\sigma^2/d^2$ and $\alpha=t/d$.
It can be shown by a straightforward algebraic calculation that the right-hand side of (\ref{superazui}) is identical to the right-hand side of (\ref{azumai}).

\cref{azuma} provides a very good bound (for example, it can be used to obtain the upper bound of \cref{chernoffthm}) but it is sometimes a little cumbersome to work with. If one is willing to settle for a weaker bound, the following version of Azuma-Hoeffding is easier to use:
\begin{theorem}[\cite{AzuOrg}]
\label{stdazuma}
Let $Z_0,Z_1,\ldots Z_n$ be a (real-valued) supermartingale. Assume that there exists a constant $d$ such that for $0\leq i<n$, $\ab{Z_{i+1}-Z_i}\leq d$. Then
\begin{equation}
\Pr[Z_n-Z_0\geq tn]\leq \exp\left(-\frac{t^2}{2d^2}n\right).
\end{equation}
\end{theorem}
\section{Models of online computation}
In this appendix, we will formally define what we mean by a randomized online algorithm with advice.

\refstepcounter{subsection}
\subsection*{\thesubsection \quad Randomized online algorithms}
\label{app:rando}

At STOC'90, Ben-David, Borodin, Karp, Tardos, and Wigderson presented a formal model of deterministic and randomized online algorithms, and proved several fundamental results regarding this model \cite{BBKTW}. We will refer to this model as the \emph{BBKTW-model}. The BBKTW-model became a popular model for formalizing randomized online computation. For example, the model is used in the textbook by Borodin and El-Yaniv \cite{BE98b}. In the BBKTW-model, a deterministic online algorithm is simply a set of mappings which defines a strategy for a Request-Answer game (see \cite{BBKTW} for the details). In particular, no computational restrictions are imposed. The model only considers what is mathematically possible to achieve in an online environment with incomplete information and does not take computational resources into account.
A randomized online algorithm is then defined as a probability distribution over deterministic online algorithms.

A probability distribution over deterministic algorithms corresponds to a \emph{mixed strategy}. Another very natural way to define a randomized algorithm is to allow the online algorithm to flip a number of coins before answering each request. Such a definition of a randomized online algorithms corresponds to a \emph{behavioral strategy}. Note that in an online problem, the algorithm has complete information of all previous requests and answers. It follows from Kuhn's theorem~\cite{Kuhn,KuhnInf} that because of this complete information, mixed and behavioral strategies are equivalent. Intuitively, one can think of the probability distribution of a mixed strategy as fixing the coin-flips of the algorithm. In the case of finite online problems, the textbook \cite{BE98b} by Borodin and El-Yaniv contains a full proof of this equivalence. However, online problems are usually not finite, since the length of the input is unbounded. For an accessible proof of the equivalence between mixed and behavioral strategies for infinite games, we refer to e.g.\ \cite{maschler2013game}. 

One possible alternative to our definition of a randomized algorithm with advice would be to define the advice complexity as the \emph{expected} number of advice bits read by the algorithm over inputs of length at most $n$. We suspect that most of our lower bounds would also hold for such slightly more powerful randomized algorithms with advice but leaves this as an open question.

\refstepcounter{subsection}
\subsection*{\thesubsection \quad The advice-at-beginning model}
\label{app:aab}
Recall that according to \cref{cadef}, the advice is provided to the online algorithm on an infinite tape in order to prevent the algorithm to learn anything from the length of the advice. We will now introduce another model of online algorithms with advice in which the advice is given as a finite bitstring. All of our lower bounds holds in this slightly more powerful model.

\begin{definition}[Advice-at-beginning]
An online algorithm with advice in the \emph{advice-at-beginning} model is defined as a triple $\ALG=\big(\{D_x\}_{x\in\{0,1\}^*},b,\varphi\big)$ where $b:\mathbb{N}\rightarrow \mathbb{N}\cup\{0\}$ is a non-decreasing function, $\varphi:I\rightarrow \{0,1\}^*$ is a mapping such that $\ab{\varphi(\sigma)}=b(n)$ where $n=\ab{\sigma}$ is the number of requests in $\sigma$, and, for each bit-string $x\in \{0,1\}^*$, $D_x$ is a deterministic online algorithm in the Request-Answer model. On input $\sigma$, the algorithm $\ALG$ produces the same output as the deterministic algorithm $D_{\varphi(\sigma)}$. In particular, $\score(\ALG(\sigma))=\score(D_{\varphi(\sigma)}(\sigma))$. The function $b(n)$ is called the \emph{advice complexity} of $\ALG$.
\end{definition}

Please note that we do not allow the algorithm to receive as advice a bitstring which is shorter than $b(n)$. This is in order to ensure that an algorithm with advice complexity $b(n)$ corresponds to at most $2^{b(n)}$ different deterministic algorithms. If the advice string was allowed to be smaller than $b(n)$, the algorithm would correspond to $\sum_{i=0}^{b(n)}2^i=2^{b(n)+1}-1$ different deterministic algorithms. This is not particularly important but slightly simplifies the model.
\paragraph{Relationship to the advice-on-tape model.}
An algorithm in the advice-on-tape model using $b$ bits of advice can be converted to an algorithm in the advice-at-beginning model using $b$ bits of advice.  This is why a lower bound in the advice-at-beginning model immediately translates into a lower bound in the advice-on-tape model.

An algorithm in the advice-at-beginning model using $b$ bits of advice can be converted to an algorithm in the advice-on-tape model as follows: On input $\sigma$, the oracle computes $b(n)$ where $n=\ab{\sigma}$ and it computes $\varphi(\sigma)$. It then writes the integer $b(n)$ onto the advice tape in a self-delimiting way followed by the string $\varphi(\sigma)$. This requires at most $b+2\lceil \log (b+1)\rceil+1=b+O(\log b)$ bits of advice. This shows that the two models are equivalent up to an additive logarithmic difference in the advice complexity.

\newpage

\section{Modeling online problems as task systems}
\label{app:modellazy}
In this appendix, we show that the following online problems can be modeled as a finite lazy task system (\cref{def:lts}):
\begin{enumerate}
\item The MTS problem on a finite metric space.
\item The $k$-server problem on a finite metric space.
\item The list update problem.
\item The dynamic binary search three problem on a set of $N$ nodes.
\end{enumerate}
We also show that reordering buffer management is $\Sigma$-repeatable.

\paragraph{MTS, $k$-server, and list update.}

If $d$ is a metric, then $d$ also satisfies the conditions imposed on the distance function in the definition of a lazy task system. This proves the following lemma.
\begin{lemma}
If $(\gs, \gt)$ is a MTS on a finite metric space, then $(\gs, \gt)$ is also a finite lazy task system.
\end{lemma}
It is easy to see that the $k$-server problem can be modeled as a MTS. It is shown in e.g.\ \cite{BE98b} how to model the list update problem as a MTS.
\begin{lemma}
The $k$-server problem on a finite metric space can be modeled as a finite lazy task system.
\end{lemma}
\begin{lemma}
The list update problem can be modeled as a finite lazy task system.
\end{lemma}
\paragraph{Dynamic binary search trees.}

We will show how to model the dynamic binary search tree problem as a lazy task system (when defined as in \cref{app:dbs}).

\begin{lemma}
\label{app:BSTlazy}
The dynamic binary search three problem on a set of $N$ nodes can be modeled as a finite lazy task system.
\end{lemma}
\begin{proof}

Fix $N$ nodes. The set of states $\gs$ consists of all possible states $s$ of the form $s=(T, v_r,v_e)$ where $T$ is a binary search tree on $N$ nodes and $v_r, v_e$ are nodes in $T$. We say that $v_r$ is the \emph{request node} of $s$ and that $v_e$ is the \emph{end note} of $s$. We will now define the set of tasks $\gt$. For each node $v$, there is an associated task $t_{v}$ defined as follows: \begin{align*}t_{v}(s)=\begin{cases}0 & \text{if $v$ is the request node of $s$}\\ \infty & \text{if $v$ is not the request node of $s$} \end{cases}\end{align*}

The distance $d(s,s')$ between two states $s=(T, v_r, v_e),s'=(T', v_r', v_e')$ is defined as follows: $d(s,s')$ is the minimum number of unit-cost operations that must be performed such that the following conditions are satisfied:
\begin{enumerate}
\item The pointer is initialized at the node $v_e$ in the tree $T$ (recall that initialization is a unit-cost operation).
\item The rotations performed during the sequence of unit-cost operations transforms $T$ into $T'$.
\item The node $v_r'$ is visited during the sequence of unit-cost operations.
\item The pointer ends at the node $v_e'$ in the tree $T'$.
\end{enumerate}

We will argue that the LTS defined above models the binary search tree problem. Consider an input sequence $\sigma=(x_1,\ldots , x_n)$ and an initial state $T_0$ for the binary search tree problem. The corresponding request sequence in $(\gs, \gt)$ will be $(t_{x_1}, \ldots , t_{x_n})$ and the initial state will be any one of the states for which $T_0$ is the associated tree. Suppose that a $(\gs, \gt)$-algorithm is in state $s=(T, v_r, v_e)$ just before $t_{x_i}$ arrives. Then the algorithm will be forced into a state $s'=(T',v_r', v_e')$ for which $v_r'=x_i$ is the request node. The algorithm incurs a cost of $d(s,s')$ for serving $t_{x_i}$. By definition, this cost is equal to the lowest possible cost for transforming $T$ into $T'$ when the pointer must begin in $v_e$ and must end in $v_e'$, and when the node containing the key $x_i$ must be visited along the way. This shows that $(\gs, \gt)$ models the dynamic binary search tree problem on the fixed set of $N$ nodes.

We still need to prove that $(\gs, \gt)$ is a LTS. Since $d(s,s')\geq 1$ for all $s,s'$ (due to the initialization cost), we only need to show that $d$ satisfies the triangle inequality. To this end, fix three states $s=(T, v_r, v_e),s'=(T', v_r', v_e')$, and $s''=(T'', v_r'', v_e'')$. Assume that $d(s,s')=t_1$ and $d(s',s'')=t_2$. This means that with a pointer starting at $v_e$, we can transform $T$ into $T'$, find $v_r'$ and end at $v_e'$ by performing $t_1$ unit-cost operations. Similarly, with the pointer starting at $v_e'$, we can transform $T'$ into $T''$, find $v_r''$ and end at $v_e''$ with $t_2$ unit-cost operations. We will prove that $d(s,s'')\leq t_1+t_2$. First, we perform the minimum number of unit-cost operations needed to transform $T$ into $T'$ and end in $v_e'$. This requires at most $t_1$ such operations (maybe fewer, since we do not need to find $v_r'$). After doing so, we perform the minimum number of unit-cost operations needed to transform $T'$ into $T''$ while also finding $v_r''$ along the way and ending at $v_e''$. This requires at most $t_2$ unit-cost operations (in fact, it requires at most $t_2-1$ operations since we do not need to perform the initialization operation which is included in $t_2=d(s',s'')$). Thus, we have transformed $T$ into $T''$, found $v_r''$ and ended up at $v_e''$ with at most $t_1+t_2$ unit-cost operations. Hence, $d(s,s'')\leq t_1+t_2 = d(s,s') + d(s',s'')$ as required.
 
\end{proof}

\paragraph{Reordering buffer management.} See \cite{DBLP:conf/esa/RackeSW02,STOCrbm} for a definition of the problem. Note that the ending of a RBM-input is special, since the algorithm has to empty its buffer at the end (even though no new item arrives). Since we need to discuss the repeated version of the problem, we need to specify exactly how this should be done and how to encode into the request sequence that the input ends. The precise way in which this is done is not important, but we need to settle on some definition. We will assume that there is a special item, $x_{\text{stop}}$, which will always be the last item of the input sequence. When a RBM-algorithm tries to load a new item into the buffer and discovers that the only item left is $x_{\text{stop}}$, it will not load this item into the buffer. Instead, the algorithm now knows that no further items will enter the buffer, and so all there is left for the algorithm is to empty its current buffer. Thus, when the item $x_{\text{stop}}$ is revealed, an algorithm incurs a cost equal to the minimum number of color switches needed to empty its buffer.
\begin{lemma}
\label{lem:rbmrep}
The reordering buffer management problem with a buffer of capacity $k$ is $\Sigma$-repeatable.
\end{lemma}
\begin{proof}
We will show that the reordering buffer management problem is $\Sigma$-repeatable with parameters $(0,k,0)$.

Let $\sigma^*=(\sigma_1;\ldots;\sigma_r)\in I^*$ be an instance of the repeated buffer management problem, RBM$^*$. For $1\leq i\leq r$, the $i$th-round input sequence has the form $\sigma_i=(x_1,\ldots , x_{n-1}, x_{\text{stop}})$ where $x_{\text{stop}}$ is a special request marking the end of the input stream (note that after $x_{\text{stop}}$ arrives, a RBM-algorithm still has to pay some cost in order to empty the buffer). For $i<r$, let $\sigma_i'=\sigma_i-\{x_{\text{stop}}\}=(x_1,\ldots , x_{n-1})$, and let $\sigma_r'=\sigma_r$. The mapping $g:I^*\rightarrow I$ is defined by $g(\sigma^*)=\sigma_1'\sigma_2'\ldots\sigma_r'$. It is clear that $g$ satisfies (\ref{re1}) with $k_1=0$ since $\ab{g(\sigma^*)}=\ab{\sigma^*}-(r-1)\leq \ab{\sigma^*}$. 

In order to prove (\ref{re2}), fix a deterministic RBM-algorithm $\ALG$ (without advice). Let $\sigma^*\in I^*$. We will now define a RBM$^*$-algorithm $\ALG^*$: The algorithm $\ALG^*$ simulates $\ALG$ on $g(\sigma^*)$. In each round, the algorithm $\ALG^*$ makes the exact same color switches (at the same time) as $\ALG$ does. When an $x_{\text{stop}}$ item is revealed, $\ALG^*$ pauses its simulation of $\ALG$ and empties its current buffer using the minimum number of color switches possible. Recall that $x_{\text{stop}}$ is always the last request of a round. When the next round begins, $\ALG^*$ resumes its simulation of $\ALG$. Note that the set of items in the buffer of $\ALG^*$ is always a subset of the set of items in the buffer of $\ALG$. Also, $\ALG^*$ pays a cost of at most $k$ for emptying its buffer at the end of each round. Thus, $\ALG^*(\sigma^*)\leq \ALG(g(\sigma^*))+kr$.

We will show that (\ref{re3}) holds. Let $\sigma^*\in I^*$. Consider an optimal sequence of color switches made by $\OPT^*_{\Sigma}$ for serving $\sigma^*$. The exact same color switches will also serve all items in $g(\sigma^*)$, since $g(\sigma^*)$ contains the same items in the same order as $\sigma^*$. The fact that $\sigma^*$ contains some additional $x_{\text{stop}}$ items which forces $\OPT^*_{\Sigma}$ to empty its buffer at the end of each round is only a disadvantage for $\OPT^*_{\Sigma}$. Thus, $\OPT^*_{\Sigma}(\sigma^*)\geq \OPT(g(\sigma^*))$.

\end{proof}

\section{Omitted proofs}
\label{sec:omitted}
\subsection{Assuming same initial state}
We prove \cref{lem:wlogsames} which justifies an assumption regarding initial states made in e.g.\ \cref{maint1}.
Define $\text{init}(\sigma)$ to be the initial state of the input $\sigma$.  Let $\P$ be an online problem and supose that we are given a probability distribution $p$ such that $\E_{p}[\DET(\sigma)]\geq c\E_p[\OPT(\sigma)]+\alpha$ for every deterministic algorithm $\DET$ and such that $\supp(p)$ is finite. Let $S_p$ be the set of all initial states appearing in at least one input in $\supp(p)$. For $s\in S_p$, let $p_s$ be the probability distribution obtained from $p$ by conditioning on the event that the initial state of the input is $s$, that is:
\begin{equation*}
p_s(\sigma)=\begin{cases}0 & \text{if $\text{init}(\sigma)\neq s$}\\
\frac{p(\sigma)}{\Pr_{\sigma\sim p}[\text{init}(\sigma)=s]} & \text{else\footnotemark} \end{cases}
\end{equation*}
\footnotetext{Note that since $\supp(p)$ is finite, so is $S_p$ and hence $\Pr_{\sigma\sim p}[\text{init}(\sigma)=s]>0$ for all $s\in S_p$.}
\begin{lemma}
Let $p, p_s$, and $S_p$ be as defined above. Then there must exist some $s\in S_p$ such that $\E_{p_s}[\DET(\sigma)]\geq c\E_{p_s}[\OPT(\sigma)]+\alpha$ for every deterministic $\P$-algorithm $\DET$.
\label{lem:wlogsames}
\end{lemma}
\begin{proof}
 Assume by way of contradiction that this is not the case. Then, for every $s\in S_p$, there exists a deterministic algorithm $\DET_s$ such that $\E_{p_s}[\DET_s(\sigma)] < c\E_{p_s}[\OPT(\sigma)]+\alpha$. We now define a $\P$-algorithm $\DET'$ as follows: Let $s$ be the initial state of the input (which is known to the algorithm before the first request arrives). The algorithm $\DET'$ uses the algorithm $\DET_s$ to serve the input. Using the law of total expectation, we get
\begin{align*}
\E_{\sigma\sim p}[\DET'(\sigma)-c\OPT(\sigma)]&=\E_{s\in S_p}\left[\E_{\sigma\sim p_s}[\DET'(\sigma)-c\OPT(\sigma)] \right]\\
&=\E_{s\in S_p}\left[\E_{\sigma\sim p_s}[\DET_s(\sigma)-c\OPT(\sigma)] \right]\\
&<\alpha,
\end{align*}
which contradicts the assumption that $\E_{\sigma\sim p}[\DET(\sigma)-c\OPT(\sigma)]\geq \alpha$ for every deterministic algorithm $\DET$.
\end{proof}

\subsection{Proof of \texorpdfstring{\cref{thm:upperanti}}{Theorem 11}}
\label{app:antilowerup}
Before giving the proof of \cref{thm:upperanti}, we introduce some terminology and prove a lemma. For $x,y\in[q]^n$, we denote the Hamming distance between $x$ and $y$ by $d(x,y)$.
\begin{definition}
Let $n,q,r\in\mathbb{N}$ and assume that $q\geq 2$. Let $M$ be a set consisting of strings (codewords) from $[q]^n$. We say that $M$ is a \emph{$q$-ary anti-covering of radius $r$ and length $n$} if for every $x\in[q]^n$, there exists a codeword $y\in M$ such that $d(x,y)\geq r$.
\end{definition}

\begin{lemma}
\label{anticodelem}
Let $n,q\in\mathbb{N}$ and assume that $q\geq 2$. Furthermore, let $0< \alpha <1/q$. There exists a $q$-ary anti-covering code of radius $(1-\alpha)n$ and length $n$ of size at most $$2^{K_{1/q}(\alpha)n}(n+1)\ln(q)n+1.$$
\end{lemma}
\begin{proof}
We apply the probabilistic method. Fix $n,q\in\mathbb{N}$ and $0<\alpha<1/q$. Let $M$ be a set of $m$ strings chosen uniformly at random and independently from $[q]^n$. Our goal is to show that if $m$ is sufficiently large, then with some non-zero probability the set $M$ is an anti-cover of radius $(1-\alpha)n$.

Fix $x\in [q]^n$. For any $y\in M$, the probability that $d(x,y)\geq (1-\alpha)n$ can be bounded from below using an inverse Chernoff bound (\cref{chernoffthm}):
\begin{align*}
\Pr[d(x,y)\geq (1-\alpha)n]\geq \frac{2^{-K_{1/q}(\alpha)n}}{n+1}:=s.
\end{align*}
Thus, the probability that there does not exists some $y\in M$ such that $d(x,y)\geq (1-\alpha)n$ is at most
\begin{align*}
\Pr[\forall y\in M : d(x,y)<(1-\alpha)n]\leq (1-s)^m.
\end{align*}
Applying the union bound over all strings in $[q]^n$ shows that the probability that $M$ is not an anti-covering of radius $(1-\alpha)n$ is at most
\begin{align*}
\Pr[M\text{ is not an anti-covering}]&\leq q^n(1-s)^m\\
&\leq q^ne^{-ms}=e^{\ln (q)n-ms}.
\end{align*}
Here, we used that $(1-a)^b\leq e^{-ba}$ for $b>1, a\leq 1$. Set $m=\lfloor s^{-1}\ln (q)n+1\rfloor$. Then $m$ is an integer strictly larger than $s^{-1}\ln (q)n$ and hence the probability that $M$ is not an anti-covering of radius $r$ is strictly smaller than $1$. Thus, the probability that $M$ is in fact such an anti-covering is strictly positive. This finishes the proof since
\begin{align*}
m&=\lfloor s^{-1}\ln(q)n+1\rfloor \\
&\leq 2^{K_{1/q}(\alpha)n}(n+1)\ln(q)n+1.
\end{align*} 
\end{proof}

\begin{proof}[Proof of \cref{thm:upperanti}]
Let $x=x_1\ldots x_n$ be the input string. The algorithm computes a $q$-ary anti-covering code, $M^*$, of radius $(1-\alpha)n$ and length $n$ with as few codewords as possible (recall that $n$ is revealed to the algorithm before it has to output its first answer). This is simply done by an exhaustive search. The oracle computes the same anti-covering code $M^*$ as the algorithm. Then, the oracle selects a codeword $y=y_1,\ldots , y_n\in M^*$ such that $d(x,y)\geq (1-\alpha)n$ and writes the (lexicographical) index of $y$ in $M^*$ onto the advice tape. The algorithm learns the index of $y$, and thereby also $y$ itself, from the advice. In round $i$, the algorithm answers $y_i$. Since $d(x,y)\geq (1-\alpha)n$, the cost incurred by the algorithm will be at most $\alpha n$.   

From \cref{anticodelem}, we get that the number of codewords in $M^*$ is at most $2^{K_{1/q}(\alpha)n}(n+1)\ln(q)n+1$. Thus, the number $b$ of advice bits needed to encode the index of a codeword in $M^*$ is at most
\begin{align*}
b&\leq \log\left(2^{K_{1/q}(\alpha)n}(n+1)\ln(q)n+1 \right)\\
&=K_{1/q}(\alpha)n+\log((n+1)\ln(q)n+1)=K_{1/q}(\alpha)n+O(\log n+\log\log q).
\end{align*}
\end{proof}

\section{Statement of main results for maximization problems}
\label{app:sec:max}
In what follows, we will state the main results (and definitions) of the paper for the case of $\Sigma$-repeatable maximization problems. The proofs are omitted since they can be obtained from the proofs of the corresponding results for minimization with only minor changes. When dealing with a maximization problem, \cref{rround} and \cref{rround2} still applies, except that the word ``cost'' should be replaced with the word ``score''. The definition of \Prs (\cref{def:prs}) applies as stated to maximization problems. The information-theoretic direct product theorem (\cref{mainthm}) and its application looks as follows for maximization problems.

\begin{theorem}[\cref{mainthm} for maximization problems]
\label{mainthm:max}
Let $\P$ be an online maximization problem and let $p^r$ be an $r$-round input distribution with associated cost-functions $\cost_i$. Furthermore, let $\ALG$ be a deterministic algorithm reading at most $b$ bits of advice on every input in the support of $p^r$.

Assume that there exists a concave and increasing function $f:[0,\infty]\rightarrow\mathbb{R}$ such that for every $w\in\mathcal{W}_i$ where $\Pr[W_i=w]>0$, the following holds: 
\begin{equation}
\E [\score_i(\ALG)\vert W_i=w]\leq f\left(D_{KL}(p_{i\vert w}\| p_i)\right)
\label{mainthmas:max}
\end{equation}
Then,
\begin{equation}
\E[\score(\ALG)]\leq  r f(b/r).
\label{mainthmconc:max}
\end{equation}
\end{theorem}

\begin{lemma}[\cref{techlemma} for maximization problems]
\label{techlemma:max}
Fix $r\geq 1$. Let $\P$ be a maximization problem. Let $p:I\rightarrow [0,1]$ be an input distribution, where $I_s$ is a finite set of $P$-inputs with the same initial state $s$ and length at most $n'$. Assume that for every deterministic $\P$-algorithm $\ALG$ without advice, it holds that $\E_{\sigma\sim p}[\ALG(\sigma)]\leq t$. Also, let $M$ be the largest score that any $\P$-algorithm can achieve on any input from $I_s$.

Then, there exists an $r$-round input distribution, $p^r$, over $\Prs$-inputs with at most $rn'$ requests in total, such that any deterministic $\Prs$-algorithm reading at most $b$ bits of advice (on inputs of length at most $rn'$) has an expected score of at most $r(t+2M\sqrt{b/r})$. Furthermore, $\E_{\sigma^*\sim p^r}[\OPT^*_{\Sigma}(\sigma^*)]=r\E_{\sigma\sim p}[\OPT(\sigma)]$.
\end{lemma}

We need to make some obvious changes to \cref{re} in order to obtain a proper definition of being $\Sigma$-repeatable for maximization problems.
\begin{definition}[\cref{re} for maximization problems]
Let $k_1,k_2,k_3\geq 0$. An online maximization problem $P$ is \emph{$\Sigma$-repeatable with parameters $(k_1,k_2,k_3)$} if there exists a mapping $g:I^*\rightarrow I$ with the following properties:
\begin{enumerate}
\item For every $\sigma\in I^*$, 
\begin{equation}
\ab{g(\sigma^*)}\leq \ab{\sigma^*}+k_1r,
\label{re1:max}
\end{equation}
where $r$ is the number of rounds in $\sigma^*$.
\item For any deterministic $\P$-algorithm $\ALG$, there is a deterministic $\Prs$-algorithm $\ALG^*$ such that for every $\sigma^*\in I^*$, 
\begin{equation}
\ALG^*(\sigma^*)\geq \ALG(g(\sigma^*))-k_2r,
\label{re2:max}
\end{equation}
where $r$ is the number of rounds in $\sigma^*$. 
\item For every $\sigma^*\in I^*$,
\begin{equation}
\OPT^*_{\Sigma}(\sigma^*)\leq \OPT(g(\sigma^*))+k_3r,
\label{re3:max}
\end{equation}
where $r$ is the number of rounds in $\sigma^*$.
\end{enumerate}
\end{definition}

\begin{theorem}[\cref{maint1} for maximization problems]
\label{maint1:max}
Let $\P$ be a $\Sigma$-repeatable online maximization problem. Suppose that for every $\varepsilon>0$ and every $\alpha$, there exists an input distribution $\pea:I\rightarrow [0,1]$ with finite support such that $\E_{\pea}[\OPT(\sigma)]\geq (c-\varepsilon)\E_{\pea}[\DET(\sigma)]+\alpha$ for every deterministic algorithm $\DET$ without advice. Then, every randomized algorithm reading at most $o(n)$ bits of advice on inputs of length $n$ has a competitive ratio of at least $c$.
\end{theorem}

\begin{definition}[\cref{compdef} for maximization problems]
\label{app:max:compdef}
Let $\P$ be a maximization problem and let $c>1$ be such that the expected competitive ratio of every randomized $\P$-algorithm is at least $c$. We say that $\P$ is \emph{compact} if for every $\varepsilon>0$ and every $\alpha\geq 0$, there exists an input distribution $\pea$ such that

\begin{enumerate}
\item If $\ALG$ is a deterministic online algorithm (without advice) then \begin{equation}\E_{\pea}[\OPT(\sigma)]\geq (c-\varepsilon)\cdot \E_{\pea}[\ALG(\sigma)]+\alpha. \label{app:max:invyaoeq}\end{equation}
\item The support $\supp (\pea)$ of $\pea$ is finite.
\end{enumerate}
\end{definition}

The maximization version of \cref{mmtt} follows easily from \cref{app:max:compdef} and \cref{maint1:max} (and \cref{derandmax:cor}).
\begin{theorem}[\cref{mmtt} for maximization problems]
Let $\P$ be compact and $\Sigma$-repeatable maximization problem with at most $\dri$ inputs of length $n$, and let $c$ be a constant not depending on $n$. The following are equivalent:
\begin{enumerate}
\item For every $\varepsilon>0$, there exists a randomized $(c+\varepsilon)$-competitive $\P$-algorithm without advice.
\item For every $\varepsilon>0$, there exists a deterministic $(c+\varepsilon)$-competitive $\P$-algorithm with advice~complexity $o(n)$.
\end{enumerate}
\end{theorem}
\end{document}